\documentclass{article}
\usepackage{fullpage}
\usepackage[utf8]{inputenc}

\usepackage[export]{adjustbox} %
\usepackage{amsmath}
\usepackage{amssymb}
\usepackage{appendix}
\usepackage{array}
\usepackage[utf8]{inputenc} %
\usepackage[english]{babel}
\usepackage{booktabs}
\usepackage{color, colortbl}
\usepackage{csquotes}
\usepackage{efbox}
\usepackage{enumitem}
\usepackage[shortcuts]{extdash} %
\usepackage[tt=false]{libertine} %
\usepackage[T1]{fontenc} %
\usepackage[symbol]{footmisc}
\usepackage[
  showframe = false,%
  pass      = true, %
]{geometry}
\usepackage{graphicx}
\usepackage{hyphenat}
\usepackage{ifthen}
\usepackage{lipsum}
\usepackage{multirow}
\usepackage{parnotes}
\usepackage{pdflscape} %
\usepackage{pifont}
\usepackage{ragged2e}
\usepackage{seqsplit}
\usepackage[group-separator={,}]{siunitx}
\usepackage{subcaption}
\usepackage{tabularx}
\usepackage{xcolor}
\usepackage{xspace}
\usepackage{url}

\usepackage[
  breaklinks    = true,
  colorlinks    = true,
  hypertexnames = false,
  pdfpagelabels = false,
  citecolor     = {blue!80!black},
  linkcolor     = {blue!80!black},
  urlcolor      = {blue!80!black},
]{hyperref} %

\usepackage{hyperref}
\usepackage{url}
\usepackage{tabularx}
\usepackage{booktabs}
\usepackage{bbm}
\usepackage{quoting}
\usepackage[normalem]{ulem}

\usepackage{enumitem}

\usepackage{geometry}
\usepackage{setspace}
\usepackage{tabularx}
\usepackage{color}
\usepackage{amsmath}
\usepackage{amssymb}
\usepackage{amsthm}
\usepackage{mathtools}
\usepackage{algorithm}
\usepackage{algorithmic}
\usepackage{bbm}
\usepackage{quoting}
\usepackage[normalem]{ulem}
\usepackage{outlines}
\usepackage{import}

\newcommand{\cM}{\mathcal{M}}

\DeclareMathOperator*{\E}{\mathbb{E}}

\makeatletter
\newtheorem*{rep@theorem}{\rep@title}
\newcommand{\newreptheorem}[2]{%
\newenvironment{rep#1}[1]{%
 \def\rep@title{#2 \ref{##1}}%
 \begin{rep@theorem}}%
 {\end{rep@theorem}}}
\makeatother

\theoremstyle{definition}
\newtheorem{theorem}{Theorem}[section]
\newreptheorem{theorem}{Theorem}
\newtheorem{lemma}[theorem]{Lemma}
\newreptheorem{lemma}{Lemma}
\newtheorem{definition}[theorem]{Definition}

\newtheorem*{remark}{Remark}

\makeatletter
\newcommand\point{\@startsection{paragraph}{4}{\z@}%
                                    {3.25ex \@plus1ex \@minus.2ex}%
                                    {-1em}%
                                    {\normalfont\normalsize\slshape}}
\makeatother

\usepackage{bigstrut}

\newcommand{\cX}{\mathcal{X}}
\newcommand{\cY}{\mathcal{Y}}
\newcommand{\cR}{\mathcal{R}}
\newcommand{\cT}{\mathcal{T}}
\newcommand{\cF}{\mathcal{F}}

\DeclareMathOperator*{\argmin}{arg\,min}
\DeclareMathOperator*{\argmax}{arg\,max}
\DeclareMathOperator*{\err}{err}
\newcommand{\as}{\texttt{AS}}
\newcommand{\osas}{\texttt{OSAS}}
\newcommand{\rp}{\texttt{RP}}
\newcommand{\rap}{\texttt{RAP}}
\newcommand{\bma}{\texttt{All-0}}
\newcommand{\rnm}{\texttt{RNM}}
\newcommand{\gm}{\texttt{GM}}

\usepackage{authblk}

\usepackage{abstract}

\title{Pushing the Boundaries of Private, Large-Scale Query Answering}

\makeatletter
\newcommand\email[2][]%
   {\newaffiltrue\let\AB@blk@and\AB@pand
      \if\relax#1\relax\def\AB@note{\AB@thenote}\else\def\AB@note{\relax}%
        \setcounter{Maxaffil}{0}\fi
      \begingroup
        \let\protect\@unexpandable@protect
        \def\thanks{\protect\thanks}\def\footnote{\protect\footnote}%
        \@temptokena=\expandafter{\AB@authors}%
        {\def\\{\protect\\\protect\Affilfont}\xdef\AB@temp{#2}}%
         \xdef\AB@authors{\the\@temptokena\AB@las\AB@au@str
         \protect\\[\affilsep]\protect\Affilfont\AB@temp}%
         \gdef\AB@las{}\gdef\AB@au@str{}%
        {\def\\{, \ignorespaces}\xdef\AB@temp{#2}}%
        \@temptokena=\expandafter{\AB@affillist}%
        \xdef\AB@affillist{\the\@temptokena \AB@affilsep
          \AB@affilnote{}\protect\Affilfont\AB@temp}%
      \endgroup
       \let\AB@affilsep\AB@affilsepx
}
\makeatother

\author[*]{Brendan Avent}
\author[**]{Aleksandra Korolova}
\affil[*]{University of Southern California}
\email{\url{bavent@usc.edu}}
\affil[**]{Princeton University}
\email{\url{korolova@princeton.edu}}

\date{}

\begin{document}

\maketitle

\begin{abstract}
\normalsize
\noindent We address the problem of efficiently and effectively answering large numbers of queries on a sensitive dataset while ensuring differential privacy (DP).
We separately analyze this problem in two distinct settings, grounding our work in a state-of-the-art DP mechanism for large-scale query answering: the Relaxed Adaptive Projection (\rap) mechanism.

The first setting is a classic setting in DP literature where all queries are known to the mechanism in advance.
Within this setting, we identify challenges in the \rap\ mechanism's original analysis, then overcome them with an enhanced implementation and analysis.
We then extend the capabilities of the \rap\ mechanism to be able to answer a more general and powerful class of queries ($r$-of-$k$ thresholds) than previously considered.
Empirically evaluating this class, we find that the mechanism is able to answer orders of magnitude larger sets of queries than prior works, and does so quickly and with high utility.

We then define a second setting motivated by real-world considerations and whose definition is inspired by work in the field of machine learning.
In this new setting, a mechanism is only given partial knowledge of queries that will be posed in the future, and it is expected to answer these future-posed queries with high utility.
We formally define this setting and how to measure a mechanism's utility within it.
We then comprehensively empirically evaluate the \rap\ mechanism's utility within this new setting.
From this evaluation, we find that even with weak partial knowledge of the future queries that will be posed, the mechanism is able to efficiently and effectively answer arbitrary queries posed in the future.
Taken together, the results from these two settings advance the state of the art on differentially private large-scale query answering.
\end{abstract}

\clearpage
\setcounter{tocdepth}{2}
\tableofcontents
\clearpage

\section{Overview} \label{sec:many-queries-introduction}
Many data analysis and machine learning algorithms, at their core, involve answering \textit{statistical queries}.
Statistical queries are the class of queries that answer the question: ``What fraction of entries in a given dataset have a particular property $P$?''
Because of their ubiquity, developing differentially private mechanisms to effectively answer statistical queries has been one of the most well studied problems in DP~\cite{dinur2003revealing, blum2005practical, dwork2006calibrating, blum2008learning, dwork2009complexity, dwork2010boosting, roth2010interactive, hardt2010multiplicative, hardt2010simple, gupta2012iterative}.
Early DP research primarily focused on designing mechanisms to answer specific, individual statistical queries in an interactive setting.
In that setting, queries are posed and answered one at a time with the goal of answering each query with minimal error while ensuring privacy.
However, most practical data-driven algorithms do not pose only a single query.
Instead, they pose a large number of queries, referred to as a \textit{query workload}.
When a query workload is available in advance (i.e., prespecified), it is possible to design DP mechanisms that take advantage of the relationships between the queries to achieve higher utility relative to answering the individual queries independently.
In this work, we address the problem of privately answering a large number of queries by answering the following high-level research question.
\begin{quoting}
\textit{In the two following settings, to what extent are differentially private mechanisms able to answer a large number of statistical queries efficiently and with low error?
\begin{itemize}
\item[] Setting 1: All queries are prespecified; i.e., known in advance.
\item[] Setting 2: Only partial knowledge of the queries is available in advance.\\
\end{itemize}
}
\end{quoting}

\subsubsection*{Motivating Example}
A motivating data analysis example for this work is The American Community Survey (ACS), a demographics survey program conducted by the U.S.\ Census Bureau~\cite{bureau2016american}.
The ACS regularly gathers information such as ancestry, citizenship, educational attainment, income, language proficiency, migration, disability, employment, and housing characteristics.
The Census Bureau aggregates the individual ACS responses (microdata), then generates population estimates which are available to the public via online data tools.
The most popular tool, Public Use Microdata Sample (PUMS), enables researchers to generate custom cross-tabulations of responses to the ACS questions.
To protect the privacy of the ACS respondents, PUMS data are sampled, anonymized, and only available for sufficiently populous geographics regions.
However, studies have found that the ad hoc anonymization techniques used are not entirely sufficient to protect the privacy of individual respondents (e.g., via re-identification attacks)~\cite{abowd2018staring, christ2022differential}.
As a result, the Census Bureau has announced plans to incorporate differential privacy into the American Community Survey, and declared that it is researching ``a new fully-synthetic data product'' with a development period ending in 2025~\cite{rodriguez2021synacs,daily_2022}.

One promising and active direction within DP research is synthetic data generation~\cite{mckenna2019graphical, vietri2020new, liu2021iterative}.
The hope is that once a synthetic dataset is generated via a differentially private mechanism, researchers and analysts can pose an arbitrary number of queries against the synthetic dataset without increasing the privacy risk to those who contributed the original underlying data.
DP synthetic data generation mechanisms seek to strike a balance between distilling the information in the underlying dataset most useful to analysts while simultaneously ensuring privacy of the underlying dataset.
Thus, to maximize the eventual usefulness of the synthetic dataset, synthetic data generation mechanisms must tailor the generated dataset to the specific class of downstream tasks (e.g., a particular class of queries) that analysts are most likely interested in.
This is typically done by providing a set of queries (the query workload) to the DP mechanism, so that the mechanism can tailor the synthetic dataset to answering these queries (and, ideally, to other similar queries).
Much of DP synthetic data research has focused on designing mechanisms to generate synthetic data which can provide accurate answers (under a variety of metrics, most commonly $\ell_\infty$ error) to the subset of statistical queries known as $k$-way marginal queries~\cite{barak2007privacy, thaler2012faster, gupta2013privately, chandrasekaran2014faster, cormode2018marginal, mckenna2019graphical, vietri2020new, nixon2022latent}.
Informally, a $k$-way marginal query is one which answers the question: ``What fraction of people in the private dataset have all of the following $k$ attributes: ...?''
In this work, we focus on a strict generalization of $k$-way marginal queries known as $r$-of-$k$ threshold queries~\cite{kearns1987learnability,littlestone1988learning,hatano2004learning, thaler2012faster,ullman2013privacy,aydore2021differentially} under the $\ell_\infty$ error metric.
Informally, $r$-of-$k$ threshold queries answer the question: ``What fraction of people in the private dataset have at least r of the following $k$ attributes: ...?''.

As a simplified example of where such queries can be used, we consider the scenario where a social scientist is interested in using ACS data to determine what portion of a community has a substandard quality of living.
Suppose the scientist wants to examine the four following attributes for each person in the community: is their income level below the poverty line, are they unemployed, are they homeless, do they have a low net worth?
Clearly, a person having any single attribute does not necessarily mean that they have a substandard quality of living.
Similarly, a person does not need to have all four attributes to have a substandard quality of living.
Thus, the social scientist can formulate this as an $r$-of-$k$ threshold query with $r=3$, $k=4$; i.e., a person has a substandard quality of living if they have at least three of the four attributes.

This social scientist may have many such queries, and other researchers may have sets of queries of their own that they wish to pose. Thus, a natural algorithm design question is: how should the U.S. Census Bureau answer everyone's queries with low error while still ensuring the ACS respondents’ privacy?
The simplest option is to use a portion of the DP budget to individually answer each query, independent of all other queries.
This would likely be unsatisfactory utility-wise, since it both limits how many queries can be answered and ignores any relationships between queries (which would likely lead to large $\ell_\infty$ error over the set of answers).
However, we posit two potentially superior alternatives whose performance we will investigate.
\begin{enumerate}
\item One alternative is to collect a large group of queries, and then use a state-of-the-art DP query answering mechanism to answer them all simultaneously.
This is an example of answering queries in the ``prespecified queries'' setting (studied in Sections~\ref{sec:improving-evaluation} and \ref{sec:extending-applicability}).
With careful DP mechanism design or selection, this alternative typically leads to lower $\ell_\infty$ error over the set of answers than answering each query independently.
\item A separate alternative is along the lines of synthetic data generation, and is applicable to the Census Bureau if queries which have been posed in the past are in some sense similar to queries which analysts will likely pose in the future.
Concretely, we hypothesize that the Census Bureau can leverage those past queries in conjunction with a state-of-the-art DP synthetic data generation mechanism to privately generate a synthetic dataset.
Researchers can then pose their own queries directly against the synthetic dataset without needing to go through the Census Bureau, and without needing to worry about the original ACS respondents’ privacy.
This is an example of answering queries in the ``partial knowledge'' setting (studied in Section~\ref{sec:understanding-generalizability}), as knowledge from the past is being used to inform the future.
If the queries posed in the past are indeed similar to the queries posed in the future, then a synthetic dataset generated using the past queries has the potential to answer the future queries with low $\ell_\infty$ error.
\end{enumerate}

\subsection{Prior Work on Large-Scale Query Answering} \label{sec:many-queries-prior-work}
To address answering a large number of queries under differential privacy in an improved manner over the naive interactive approach, two separate lines of research previously emerged: synthetic data generation, and workload evaluation.
We describe both lines of research, then briefly introduce the state-of-the-art mechanism which we build upon in this work.

\paragraph{Synthetic Data Generation:}
One line of research studies the problem of answering a large number of queries via private synthetic dataset generation.
In differentially private synthetic dataset generation, a DP mechanism is applied to the original, sensitive data in order to generate a synthetic dataset.
The synthetic dataset's purpose is then to directly answer arbitrary queries posed in the future, without the further need to account for potential privacy leakage or manage differential privacy budgets.
In this setting, aside from knowing the general query class, \textit{no knowledge is typically assumed about which specific queries will be posed in the future}.
The proven advantage of this approach is that DP synthetic datasets are theoretically capable of accurately answering an exponentially larger number of queries relative to the aforementioned interactive approach~\cite{gupta2012iterative, cheraghchi2012submodular, hardt2012private, gupta2013privately}.
However, actually generating a synthetic dataset which accurately answers exponentially many queries has been proven intractable~\cite{dwork2009complexity, ullman2011pcps, ullman2016answering}, even for simple subclasses of statistical queries (e.g., 2-way marginals).
Thus, a significant recent research focus has been on designing efficient mechanisms for privately generating synthetic datasets which accurately answer increasingly large numbers of queries~\cite{gaboardi2014dual, mckenna2019graphical, vietri2020new, liu2021iterative}.

\paragraph{Workload Evaluation:}
A separate line of research focuses on the problem of answering a large number of queries when the concrete query workload is prespecified; i.e., \textit{when all queries are known in advance}.
Pre-specifying the query workload enables researchers to design DP mechanisms to take advantage of the workload's structure in order to answer the queries with lower error relative to the interactive approach or the private synthetic dataset approach.
Early research in this setting produced mechanisms with optimal or near-optimal error guarantees, but with impractical (typically exponential) running times for even modestly sized real-world problems~\cite{hardt2010multiplicative, hardt2010simple, gupta2012iterative, li2015matrix}.
As a result, recent research has focused on designing computationally efficient mechanisms to answer prespecified workloads with low error on real-world datasets~\cite{mckenna2018optimizing, snoke2018pmse, aydore2021differentially}, at the cost of losing the strong theoretical utility guarantees of prior works and thus necessitating thorough empirical utility evaluations to demonstrate their value.

\paragraph{\textit{Relaxed Adaptive Projection} Mechanism:}
Our approach for evaluating suitable (i.e., efficient and accurate) mechanisms in both our settings of interest builds on Aydore et al.'s~\cite{aydore2021differentially} recently introduced \textit{Relaxed Adaptive Projection} (\rap) mechanism.
\rap\ is the current state-of-the-art mechanism for answering large sets of statistical queries in the setting where the query workload is prespecified.
At a high-level, \rap\ works by:
\begin{enumerate}
\item Initializing a synthetic dataset $D'$ in a relaxed data space (e.g., by relaxing a binary feature in the original dataset to the interval $[0,1]$ in the synthetic dataset).
\item For each original prespecified query, specifying a surrogate query which is equivalent to the original in the unrelaxed data space, but which is differentiable everywhere in the relaxed space. 
\item Iteratively applying an \textit{Adaptive Selection} (\as) step followed by a \textit{Relaxed Projection} (\rp) step. In the \as\ step, adaptivity is introduced to allow the subset of queries with the highest error on $D'$ to be privately selected. In the \rp\ step, these selected queries' surrogates are used to optimize $D'$ using standard gradient-based optimization techniques.
\item Finally, answering the original set of queries using the optimized synthetic dataset $D'$.
\end{enumerate}
For $k$-way marginals, a canonical subclass of statistical queries~\cite{barak2007privacy, thaler2012faster, gupta2013privately, chandrasekaran2014faster, cormode2018marginal} (formally defined in Section~\ref{sec:preliminaries}), Aydore et al.\ theoretically and empirically demonstrate that \rap\ outperforms prior state-of-the-art mechanisms.
Theoretically, they provide an ``oracle efficient'' (i.e., assuming the optimization procedure achieves a global minima) utility result characterizing \rap's error, showing that \rap\ achieves strictly lower error than the previous practical state-of-the-art mechanism~\cite{vietri2020new}.
Experimentally, they compare the \rap\ mechanism with prior state-of-the-art mechanisms~\cite{mckenna2019graphical, vietri2020new}, demonstrating that \rap\ answers prespecified sets of queries with lower error.\\

\subsection{Our Contributions} \label{sec:contributions}
To answer this work's high-level research question, we make the following contributions in both settings of interest.
In the classic setting where all queries are known in advance, our contributions are as follows.
\begin{itemize}
	\item We overcome memory hurdles in \rap's initial implementation by reimplementing \rap\ in a memory-efficient way, thus enabling the evaluation of significantly larger query spaces than previously considered.
	\item We utilize the new implementation to enhance \rap's evaluation, evaluating \rap\ on larger query spaces (answering approximately 50x more queries) than in its initial evaluation, and conclusively determining the role that adaptivity from the \as\ step plays in \rap's utility.
	\item We extend \rap's applicability by expanding the class of queries that it evaluates, finding that it can efficiently and effectively answer more complex query classes than previously considered.
\end{itemize}

\noindent As a realistic intermediate setting that lies between the two classic extremes of no-knowledge vs.\ full-knowledge of which queries will be posed, we propose a new setting where partial knowledge of the future queries is available.
In this new setting, our contributions are as follows.
\begin{itemize}
	\item We concretely define this setting as well as how to measure utility within it. Specifically, we assume that a set of historical queries was independently drawn from some unknown distribution $\cT_H$, and that the mechanism has access to these historical queries. In the future, the mechanism will be posed an arbitrary number of queries sampled from a distribution $\cT_F$, which may be related to $\cT_H$. We define the utility of the mechanism in terms of its generalization error; i.e., its expected error across these future queries drawn from $\cT_F$ having been given access to the historical queries from $\cT_H$.
	\item We assess how suitable \rap\ is for this new setting by formulating query distributions according to real-world phenomena, then empirically evaluating \rap's generalization error on these distributions. When future queries are drawn from the same distribution as the historical queries that \rap\ used to learn its synthetic dataset (i.e., $\cT_H=\cT_F$), we find that regardless of what the distribution is, \rap\ is able to achieve high utility. When the distribution of future queries diverges from the distribution of historical queries, we find that \rap's utility slowly and gracefully declines.
\end{itemize}

\noindent These contributions, in both the prespecified queries setting and the partial knowledge setting, definitively demonstrate the practical value of \rap\ and improve \rap's adoptability for real-world uses.\\

The remainder of this work is structured as follows.
Beginning in Section~\ref{sec:preliminaries}, we provide a comprehensive overview of the relevant technical terminology and definitions, and detail the \rap\ mechanism that we build upon.
In Section~\ref{sec:improving-evaluation}, we perform a focused but thorough reproducibility study on Aydore et al.'s~\cite{aydore2021differentially} evaluation of the \rap\ mechanism.
To accomplish this, we first improve \rap's implementation from the ground up, and then leverage the new implementation to enhance \rap's initial evaluation in order to strengthen our comprehension of its utility.
Building on the improved \rap\ implementation, in Section~\ref{sec:extending-applicability} we expand the class of queries that \rap\ is able to accommodate.
We then empirically evaluate \rap\ on this new class of queries, finding that it is able to efficiently answer large numbers of queries from this class while maintaining high utility.
In Section~\ref{sec:understanding-generalizability}, we concretely define our newly proposed setting where a mechanism is given partial knowledge of the queries that will be posed in the future.
We define how we assess \rap's performance in this setting, and detail the distinct new ways that \rap's performance may be affected in this new setting.
We then empirically evaluate \rap\ in this setting, finding that even with only partial knowledge of which queries will be posed in the future, \rap\ is able to efficiently and effectively achieve high utility.
Finally, in Section~\ref{sec:many-queries-related-works}, in addition to the related works already discussed in this section, we describe other important relevant works and the future directions they motivate related to this work.

\section{Technical Preliminaries} \label{sec:preliminaries}
In this section, we define the requisite technical terminology.
The fundamental concepts introduced here were primarily presented in prior works~\cite{gaboardi2014dual, vietri2020new, aydore2021differentially}.
We restate them to aid in understanding and contextualizing Aydore et al.'s \rap\ mechanism, which we use to answer this work's research questions.
Towards this, we first define statistical queries and their subclasses that are relevant to this work.
We then define what it means to be a ``surrogate'' query for one of these statistical queries.
Next, we describe what workloads are and how we use them.
Finally, we detail the \rap\ mechanism that we build on in this work.
Because this work is notationally dense, Table~\ref{tab:symbols} serves as a reference for the various symbols that we define.

\begin{table}
\setlength\tabcolsep{5pt}
\begin{tabularx}{\columnwidth}{rc|X}
& \textbf{Symbol} & \textbf{Usage} \\
\midrule
\midrule
& $\epsilon, \delta$ & Differential privacy parameters. \\
\hline
& $\cX,\ d,\ \cX_i$ & Data space $\cX$ for any possible record consisting of $d$ features.  $\cX_i$ is the domain of feature $i$.\\
& $D,\ n$ & Dataset $D$ containing $n$ records from $\cX$. \\
\hline
& $q_{\phi}$ & Statistical query $q$ defined by the mean of the predicate $\phi$ over a set of records from $\cX$. \\
& $Q,\ m,\ a$ & $Q$ is a vector of $m$ queries, and $a$ represents the answers to the vector of queries over the dataset $D$ such that $Q(D) = a = (a_1,\dots,a_m)$. \\
& $W$ & Threshold workload $W$ which defines the concrete query vector $Q$. \\
& $q_{\phi_{S,y,k}}$ & $k$-way marginal query specified by set $S$ of $k$ features and values $y$ for each feature. \\
& $q_{\phi_{S,y,1}}$ & 1-of-$k$ threshold query specified by set $S$ of $k$ features and values $y$ for each feature. \\
$\star$ & $q_{\phi_{S,y,r}}$ & $r$-of-$k$ threshold query specified by set $S$ of $k$ features and values $y$ for each feature, and threshold $r$. \\
\hline
& $\cY,\ d' $ & Data space $\cY$ consisting of $d'$ features, which is a relaxation of the one-hot encoded $\cX$ data space. \\
& $D',\ n'$ & Synthetic dataset $D'$ containing $n'$ features from $\cY$. \\
& $\hat{q}_{\hat{\phi}}$ & Surrogate query $\hat{q}$ defined by the mean of the function $\hat{\phi}$ over a set of records from $\cY$. \\
& $\hat{Q}$ & Vector of surrogate queries. \\
& $\hat{q}_{\hat{\phi}_T}$ & Product query, specified by a set of features $T$. \\
$\star$ & $\hat{q}_{\hat{\phi}_{T_+,T_-}}$ & Generalized product query, specified by a set of positive and negated features $T_+$ and $T_-$. \\
$\star$ & $\hat{q}_{\hat{\phi}_{T,r}}$ & Polynomial threshold query, specified by a set of features $T$ and integer $r$. \\
\hline
$\star$ & $\err_P$ & Measure of a mechanism's present error, used when all queries are known in advance.\\
$\star$ & $\err_F$ & Measure of mechanism's future error, used when only partial knowledge of queries is available in advance.\\
$\star$ & $\cF, \cT$ & Distribution $\cT$ from which thresholds in a random workload are sampled i.i.d.\ in order to form a corresponding vector of consistent queries. The threshold distribution may be formed by a distribution over features $\cF$. \\
\hline
& \rap, \as, \rp & Relaxed Adaptive Projection mechanism, with its primary subcomponents: the Adaptive Selection and Relaxed Projection mechanisms. \\ 
& \rnm & Report Noisy Max mechanism, used by the \as\ mechanism to select high-error queries. \\
& \gm & Gaussian noise-addition mechanism, used as both a baseline mechanism as well as a subcomponent of \rap\ to privately answer queries directly. \\
$\star$ & \osas & Oneshot Adaptive Selection mechanism, introduced as more efficient a drop-in replacement for \rap's \as\ mechanism. \\
& \bma & Baseline mechanism that returns only 0 for all queries. \\
\end{tabularx}
\caption{Comprehensive list of notation. Lines marked with a $\star$ indicate new concepts not found in~\cite{aydore2021differentially}.} \label{tab:symbols}
\end{table}

\subsection{Statistical Queries and their Subclasses} \label{sec:stat-qs}
The general class of queries that we are interested in (which the \rap\ mechanism can, in theory, be used to answer) are statistical queries.
\begin{definition}[Statistical query]
A \textit{statistical query} $q_{\phi}$ is parameterized by a predicate $\phi: \cX \rightarrow \{0,1\}$; i.e., the predicate takes as input a record $x$ of a dataset $D$, and outputs a boolean value.
The statistical query is then defined as the normalized count of the predicate over all $n$ records of the input dataset; i.e.,
$$q_{\phi}(D) = \frac{\sum_{x\in D} \phi(x)}{n}.$$
Given a vector of $m$ statistical queries $Q$, we define $Q(D)=(a_1,\dots,a_m)$ to be the answers to each of the queries on $D$; i.e., $a_i = q_{\phi_i}(D)$ for all $i \in [m]$.
\end{definition}

We now formally define the specific subclasses of statistical queries that we reference throughout this work.
Let the space for each record in the dataset consist of $d$ categorical features $\cX = (\cX_1 \times \cdots \times \cX_d)$, where each $\cX_i$ is the discrete domain of feature $i$, and let $x_i \in \cX_i$ denote the value of feature $i$ of record $x \in \cX$.
Prior works have primarily evaluated the subclass of statistical queries known as $k$-way marginals (also known as $k$-way contingency tables or $k$-way conjunctions)~\cite{barak2007privacy, thaler2012faster, gupta2013privately, chandrasekaran2014faster, cormode2018marginal, mckenna2019graphical, vietri2020new}, and typically focused specifically on 3-way and 5-way marginals.
\begin{definition}[$k$-way marginal] \label{def:kw}
A \textit{$k$-way marginal query} $q_{\phi_{S,y,k}}$ is a statistical query whose predicate $\phi_{S,y,k}$ is specified by a set $S$ of $k$ features $f_1 \neq \cdots \neq f_k \in [d]$ and a target $y \in (\cX_{f_1} \times \cdots \times \cX_{f_k})$, given by
\[ \phi_{S,y,k}(x)=
   \begin{cases} 
      1 & \text{if }\ x_{f_1} = y_1 \ \wedge \cdots \ \wedge x_{f_k} = y_k \\
      0 & \text{otherwise.} 
   \end{cases}
\]
Informally, a row satisfies the predicate if \textit{all} of its values match the target on the specified features.
A \textit{$k$-way marginal} is then specified by a set $S$ of $k$ features, and consists of all ($\Pi_{i=1}^k |\cX_{f_i}|$) $k$-way marginal queries with feature set $S$.
\end{definition}

1-of-$k$ thresholds (also known as $k$-way disjunctions) were briefly evaluated in \cite{aydore2021differentially}, and are defined similarly.
\begin{definition}[1-of-$k$ threshold] \label{def:1k}
A \textit{1-of-$k$ threshold query} $q_{\phi_{S,y,1}}$ is a statistical query whose predicate $\phi_{S,y,1}$ is specified by a set $S$ of $k$ features $f_1 \neq \cdots \neq f_k \in [d]$ and a target $y \in (\cX_{f_1} \times \cdots \times \cX_{f_k})$, given by
\[ \phi_{S,y,1}(x)=
   \begin{cases} 
      1 & \text{if }\ x_{f_1} = y_1 \ \vee \cdots \ \vee x_{f_k} = y_k \\
      0 & \text{otherwise.} 
   \end{cases}
\]
Informally, a row satisfies the predicate if \textit{any} of its values match the target on the specified features.
A \textit{1-of-$k$ threshold} is then specified by a set $S$ of $k$ features, and consists of all ($\Pi_{i=1}^k |\cX_{f_i}|$) 1-of-$k$ threshold queries with feature set $S$.
\end{definition}

Finally, in this work, we evaluate a generalization of both of these subclasses of statistical queries: $r$-of-$k$ thresholds~\cite{kearns1987learnability,littlestone1988learning,hatano2004learning, thaler2012faster,ullman2013privacy,aydore2021differentially}.
\begin{definition}[$r$-of-$k$ threshold] \label{def:rk}
An \textit{$r$-of-$k$ threshold query} $q_{\phi_{S,y,r}}$ is a statistical query whose predicate $\phi_{S,y,r}$ is specified by a positive integer $r \le k$, a set $S$ of $k$ features $f_1 \neq \cdots \neq f_k \in [d]$, and a target $y \in (\cX_{f_1} \times \cdots \times \cX_{f_k})$.
The predicate is then given by
\[\phi_{S,y,r}(x) = \mathbbm{1}\left[ \sum_{i=1}^k \mathbbm{1}[x_{f_i} = y_i] \ge r \right].\]
Informally, a row satisfies the predicate if \textit{at least} $r$ of its values match the target on the specified features.
An \textit{$r$-of-$k$ threshold} is then specified by positive integer $r \le k$ and a set $S$ of $k$ features, and consists of all ($\Pi_{i=1}^k |\cX_{f_i}|$) $r$-of-$k$ threshold queries with feature set $S$.
This class generalizes $k$-way marginals when $r=k$, and generalizes 1-of-$k$ thresholds when $r=1$.
\end{definition}
The expressiveness of $r$-of-$k$ thresholds make them more useful than $k$-way marginals, as they enable more nuanced queries to be easily and intuitively posed.
This is particularly useful when the implications behind categories of distinct features in a dataset have some overlap.
For instance, in the motivating U.S. Census example, there were several features with categories that were indicative of a substandard quality of living.
Requiring someone to belong to \textit{all} of the categories (as a $k$-way marginal requires) is overly restrictive, and $r$-of-$k$ thresholds allow this restrictiveness to be relaxed.

\begin{remark}
We say that any $r$-of-$k$ threshold query (and, by extension, any $k$-way marginal query or 1-of-$k$ threshold query) specified by $r$, $k$, $S$, and $y$ is \textit{consistent} with the $r$-of-$k$ threshold specified by $r$, $k$, and $S$.
That is, we often refer to an $r$-of-$k$ threshold simply as the features it specifies, whereas a query \textit{consistent with} that $r$-of-$k$ threshold is one which specifies concrete target values corresponding to those features.
\end{remark}

\subsection{Surrogate Queries}
Aydore et al.~\cite{aydore2021differentially} introduce surrogate queries to replace the original statistical queries with queries that are similar, but that are amenable to first-order optimization methods.
These first-order optimization methods, thanks to significant recent advances in hardware and software tooling, can enable highly efficient learning of synthetic datasets.

\begin{definition}[Surrogate Query]
A \textit{surrogate query} $\hat{q}_{\hat{\phi}}$ is parameterized by function $\hat{\phi}: \cY \rightarrow \mathbb{R}$; i.e., the function takes as input a record $x \in \cY$ from a dataset $D'$, and outputs a real value.
The surrogate query is then defined as the normalized count of the function over all $n'$ records of the input dataset; i.e., 
\[ \hat{q}_{\hat{\phi}}(D') = \frac{\sum_{x\in D'} \hat{\phi}(x)}{n'}.\]
The only distinctions between the definitions of a surrogate query with $\hat{\phi}$ and a statistical query with $\phi$ are that $\hat{\phi}$'s domain may be different than $\phi$'s, and $\hat{\phi}$'s codomain is the entire real line instead of $\{0,1\}$.
\end{definition}

We are interested in surrogate queries that are \textit{equivalent extended differentiable queries} (EEDQs) as defined in~\cite{aydore2021differentially}.
\begin{definition}[Equivalent Extended Differentiable Query]
Let $q_{\phi}$ be an arbitrary statistical query parameterized by $\phi(x):\cX \rightarrow \{0,1\}$, and let $\hat{q}_{\hat{\phi}}$ be a surrogate query parameterized by $\hat{\phi}: \cY \rightarrow \mathbb{R}$.
We say that $\hat{q}_{\hat{\phi}}$ is an \textit{equivalent extended differentiable query} to $q_{\phi}$ if it satisfies the following properties:
\begin{enumerate}
\item $\hat{\phi}$  is differentiable over $\cY$. I.e., for every $x \in \cY,\ \nabla \hat{\phi}(x)$ is defined.
\item $\hat{\phi}$ agrees with $\phi$ on every possible database record that results from a one-hot encoding. I.e., for every $x \in \cX$ where $h(x)$ represents a one-hot encoding\footnote{A one-hot encoding of a categorical feature $\cX_i$ with $t_i$ categories is a mapping from each category to a unique $1 \times t_i$ binary vector that has exactly 1 non-zero coordinate.} of $x$: $\phi(x) = \hat{\phi}(h(x))$.
\end{enumerate}
\end{definition}

\paragraph{Notation of Feature Spaces:}
Recall the original feature space $\cX = (\cX_1 \times \dots \times \cX_d)$, where each $\cX_i$ is the discrete domain of feature $i$, and let $t_i$ be the number of distinct values/categories that $\cX_i$ can attain.
A one-hot encoding $h(x)$ of any record $x$ results in a binary vector $\{0,1\}^{d'}$, where $d' = \sum_{i=1}^d t_i$.
Just as in~\cite{aydore2021differentially}, we are interested in constructing a synthetic dataset that lies in a continuous relaxation of this binary feature space.
A natural relaxation of $\{0,1\}^{d'}$ is $[0,1]^{d'}$, so we adopt $\cY = [0,1]^{d'}$ as the relaxed space for the remainder of this work.\\

As an illustrative example of an EEDQ, we define the class of EEDQ's used by Aydore et al.\ for $k$-way marginals.
Concretely, \cite{aydore2021differentially} defines the class of surrogate queries known as \textit{product queries}, and shows how to construct an EEDQ product query for any given $k$-way marginal.
\begin{definition}[Product Query] \label{def:pq}
Given a subset of features $T \subseteq [d']$, the \textit{product query} $\hat{q}_{\hat{\phi}_T}$ is a surrogate query parameterized by function $\hat{\phi}_T$ which is defined as $\hat{\phi}_T(x) = \prod_{i \in T} x_i$.
\end{definition}

\begin{lemma}[\cite{aydore2021differentially}, Lemma 3.3] \label{lem:pq}
Every $k$-way marginal query $q_{\phi_{S,y,k}}$ has an EEDQ in the class of product queries. By construction, every $\hat{\phi}_T$ satisfies the requirement that it is defined over the entire relaxed space $\cY$ and is differentiable.
Additionally, for every $q_{\phi_{S,y,k}}$, there is a corresponding product query $\hat{q}_{\hat{\phi}_T}$ with $|T|=k$ such that for every $x \in \cX: \phi_{S,y,k}(x) = \hat{\phi}_T(h(x))$. We construct this $T$ in the following straightforward way: for every $i \in S$, we include in $T$ the coordinate corresponding to $y_i \in \cX_{f_i}$.
\end{lemma}

\subsection{Threshold Workloads}
It was standard in prior works to evaluate \textit{workloads} of $k$-way marginals~\cite{li2015matrix, mckenna2018optimizing, mckenna2019graphical, vietri2020new, liu2021leveraging, liu2021iterative}.
A $k$-way marginal workload $W$ is specified by a set of $k$-way marginals, $W = \{S_1, \dots, S_{|W|}\}$ such that each $S_i \in W$ is a set of $k$ features.
This workload $W$ defines a concrete query vector $Q$ which consists of all queries consistent with each marginal in $W$.
Since $Q$ is defined by the marginal workload defines, $Q$ is commonly referred to as the \textit{query workload}.
For example, a workload may be specified by the following two $3$-way marginals, $W = \{(1,2,5), (2,3,7)\}$, and would therefore define the query vector $Q$ containing all marginal queries consistent with those feature sets.
The number of queries in this query vector would then be $|Q|=|\cX_1||\cX_2||\cX_5| + |\cX_2||\cX_3||\cX_7|$.

Since our work extends the class of queries from marginals to $r$-of-$k$ thresholds, rather than a workload being specified by a set of marginals, we say that a workload $W$ is specified by a set of $r$-of-$k$ thresholds.
$W$ similarly defines the concrete query vector $Q$ which consists of all $r$-of-$k$ threshold queries consistent with each $r$-of-$k$ threshold in $W$.
For example, when $r=1$ and $k=3$, we can specify a similar workload as before $W = \{(1,2,5), (2,3,7)\}$ which defines query workload $Q$ containing the same number of consistent queries as before ($|Q|=|\cX_1||\cX_2||\cX_5| + |\cX_2||\cX_3||\cX_7|$) --- however, here each $q \in Q$ is a 1-of-3 threshold query instead of a 3-way marginal query.

Lastly, we let $\hat{Q}$ denote the corresponding vector of surrogate queries for $Q$.
We use threshold workloads (and their corresponding vector of all consistent queries) for the empirical evaluations of our mechanisms.

\subsection{\textit{Relaxed Adaptive Projection} (\rap) Mechanism}
We now describe the details of the \rap\ mechanism, including how it works as well as its DP guarantee.

Algorithm~\ref{alg:rap} formally defines the \rap\ mechanism.
The input to the mechanism is the dataset $D$ of sensitive user data, the desired size of the synthetic dataset $n'$, privacy parameters $(\epsilon, \delta)$, a vector of $m$ statistical queries $Q$ and their corresponding surrogate queries $\hat{Q}$, adaptiveness parameters $T,K \in [m]$.
The final outputs are (1) an $n'$-row synthetic dataset, and (2) estimates to the original queries $Q$ obtained by evaluating their surrogate queries on the synthetic dataset; i.e., \rap\ outputs (1) $D'$ and (2) $\hat{Q}(D')$.

\begin{algorithm}
\caption{Relaxed Adaptive Projection (\rap) Mechanism}
     \vspace{0.4em} \hspace*{\algorithmicindent} \textbf{Input} \vspace*{-0.6em}
     \begin{itemize}[leftmargin=1.5em] \setlength\itemsep{-0.2em}
     	\item $D$: Dataset of $n$ records from space $\cX$.
     	\item $Q, \hat{Q}$: A vector of $m$ statistical queries and their corresponding surrogate queries.
     	\item $n', \cY$: Desired size of synthetic dataset with records from relaxed space $\cY$.
     	\item $T$: Number of rounds of adaptiveness.
     	\item $K$: Number of queries to select per round of adaptiveness.
     	\item $\epsilon, \delta$: Differential privacy parameters.
     \end{itemize}
     \hspace*{\algorithmicindent} \textbf{Body}
\begin{algorithmic}[1]
    \STATE Let $\rho = \epsilon + 2\left(\log(\frac{1}{\delta}) - \sqrt{\log(\frac{1}{\delta})(\epsilon+\log(\frac{1}{\delta}))}\right)$.
    \STATE Independently uniformly randomly initialize $D' \in \cY^{n'}$.
    \IF{ $T=1,\ K=m$ } 
    	\FOR{$i=1,2,\dots,m$}
    		\STATE Let $\tilde{a}_i = \gm(D, q_i, \rho / m)$.
    	\ENDFOR
    	\STATE Let $D' = \rp(\hat{Q}, \tilde{a}, D')$.
    \ELSE 
    	\STATE Let $Q_s = \emptyset$.
    	\FOR{$t=1,2,\dots,T$}
			\STATE Let $Q_s, \tilde{a} = \as(D, D', Q, \hat{Q}, Q_s, K, \frac{\rho}{T})$.
			\STATE Let $\hat{Q}_{s} = (\hat{q}_i: q_i \in Q_s)$.
			\STATE Let $D' = \rp(D', \hat{Q}_s, \tilde{a})$.
    	\ENDFOR
    \ENDIF
    \STATE {\bfseries Return:} Final synthetic dataset $D'$ and answers $\hat{Q}(D')$.
\end{algorithmic}
\label{alg:rap}
\end{algorithm}

\begin{algorithm}
\caption{Adaptive Selection (\as) Mechanism}
     \vspace{0.4em} \hspace*{\algorithmicindent} \textbf{Input} \vspace*{-0.6em}
     \begin{itemize}[leftmargin=1.5em] \setlength\itemsep{-0.2em}
     	\item $D, D'$: Dataset of $n$ records from space $\cX$, and synthetic dataset of $n'$ records from relaxed space $\cY$.
     	\item $Q, \hat{Q}$: Vector of all statistical queries and their corresponding surrogate queries.
     	\item $Q_s$: Set of already selected queries.
     	\item $K$: Number of new queries to select.
     	\item $\rho$: Differential privacy parameter.
     \end{itemize}
     \hspace*{\algorithmicindent} \textbf{Body}
\begin{algorithmic}[1]
     \FOR{$j = 1, 2, \dots, K$}
   	   	\STATE Let $\Delta = (|\hat{q}_{i}(D) - \hat{q}_{i}(D')| : q_i \in Q \setminus Q_s)$.
     	\STATE Let $i = \rnm(\Delta, \frac{\rho}{2K})$
     	\STATE Add $q_i$ into $Q_s$.
     	\STATE Let $\tilde{a}_i = \gm(D, q_i, \frac{\rho}{2K})$.
     \ENDFOR
     \STATE {\bfseries Return:} $Q_s$ and $\tilde{a} = (\tilde{a}_i: q_i \in Q_s)$.
\end{algorithmic}
\label{alg:as}
\end{algorithm}

\begin{algorithm}
\caption{Relaxed Projection (\rp) Mechanism}
     \vspace{0.4em} \hspace*{\algorithmicindent} \textbf{Input} \vspace*{-0.6em}
     \begin{itemize}[leftmargin=1.5em] \setlength\itemsep{-0.2em}
     	\item $D'$: Synthetic dataset of $n'$ records from relaxed space $\cY$.
     	\item $\hat{Q}$: Vector of surrogate queries.
     	\item $\tilde{a}$: Vector of ``true'' privatized answers corresponding to each surrogate query.
     \end{itemize}
     \hspace*{\algorithmicindent} \textbf{Body}
\begin{algorithmic}[1]
    \STATE Use any iterative differentiable optimization technique (SGD, Adam, etc.) to attempt to find:
    $$D' = \argmin_{D' \in \cY^{n'}} ||\hat{Q}(D') - \tilde{a}||_2^2,$$
    applying the Sparsemax transformation to every feature encoding in each row of $D'$ between each iteration.
    \STATE {\bfseries Return:} $D'$.
\end{algorithmic}
\label{alg:rp}
\end{algorithm}

\paragraph{Non-Adaptive Case:}
In its most basic form ($T=1, K=m$), \rap\ employs no adaptivity.
Here, the vector of $m$ queries are first privately answered directly on the sensitive dataset $D$ using the \textit{Gaussian Mechanism} (\gm).
These answers, along with the vector of surrogate queries $\hat{Q}$ and a uniformly randomly initialized $n'$-row synthetic dataset $D'$, are passed to the \textit{Relaxed Projection} mechanism (\rp, Algorithm~\ref{alg:rp}).
The \rp\ subcomponent utilizes an iterative gradient-based optimization procedure (such as SGD) to update $D'$ by minimizing the disparity between the surrogate queries answers on $D'$ and the privatized answers on the sensitive dataset $D$.
After iterative update, the Sparsemax transformation is applied to every feature encoding in each row of $D'$.
Once the procedure reaches a stopping condition (e.g., $\hat{Q}(D')$ is within a certain tolerance of $\tilde{a}$, or a certain number of iterations have occurred), \rp\ returns the final $D'$.
\rap\ then returns $D'$ along with estimated answers to the query workload $\hat{Q}(D')$.

\paragraph{Adaptive Case:}
In the more general case, \rap\ proceeds in $T > 1$ rounds.
In each round $t$, \rap\ uses the \textit{Adaptive Selection} (\as) mechanism to select $K$ new queries to add to the set $Q_s$.
\as\ iteratively uses the Gumbel noise \textit{Report Noisy Max} (\rnm)~\cite{chen2016truthful, durfee2019practical} and \gm\ mechanisms together to privately choose the $K$ queries that have the largest disparity between their current answers on the synthetic dataset $D'$ and their answers on the true dataset $D$.
The \rp\ mechanism is then applied only to this subset $Q_s$ containing $tK$ queries in each round, rather than applying \rp\ in 1 round on the full vector of privately answered queries $Q$ (as in the non-adaptive case).
Aydore et al.\ claim that the aim of incorporating this adaptivity is to expend the privacy budget more wisely by selectively answering only the $TK \ll m$ total worst-performing queries.

\subsubsection*{Concentrated Differential Privacy}
To state and understand \rap's DP guarantee, we must briefly discuss \textit{zero-concentrated differential privacy} (zCDP)~\cite{bun2016concentrated}.

Although \rap\ is given $\epsilon$ and $\delta$ values as input and in turn guarantees $(\epsilon, \delta)$-DP, its DP sub-mechanisms and corresponding privacy proof are in terms of $\rho$-zCDP.
Zero-concentrated differential privacy is a different definition of DP that provides a weaker guarantee than pure DP but a stronger guarantee than approximate DP.
It is formally defined as follows.
\begin{definition}[\cite{bun2016concentrated}]
A randomized mechanism $\cM$ is $\rho$-zCDP if and only if for all neighboring input datasets $D$ and $D'$ that differ in precisely one individual's data and for all $\alpha \in (1,\infty)$, the following inequality is satisfied:
$$\mathbb{D}_\alpha(\cM(D) || \cM(D')) \le \rho \alpha,$$
where $\mathbb{D}_\alpha(\cdot || \cdot)$ is the $\alpha$-R\'enyi divergence.
\end{definition}

We omit a detailed discussion of zCDP in this work, referring an interested reader to Bun and Steinke's work~\cite{bun2016concentrated} for more details.
However, its value for \rap\ comes from the fact that zCDP has better composition properties than approximate DP, yet \rap's final composed zCDP guarantee (parameterized by $\rho$) can be converted back into an $(\epsilon, \delta)$-DP guarantee.
This converted $(\epsilon, \delta)$-DP guarantee is better than if standard composition results of approximate DP had been directly applied.

We now informally state these composition and conversion properties.
zCDP's composition property ensures that if two mechanisms satisfy $\rho_1$-zCDP and $\rho_2$-zCDP, then a mechanism that sequentially composes them satisfies $\rho$-zCDP with $\rho = \rho_1 + \rho_2$.
zCDP's conversion property ensures that if a mechanism satisfies $\rho$-zCDP, then for any $\delta > 0$, the mechanism also satisfies $(\epsilon, \delta)$-DP with $\epsilon = \rho + 2\sqrt{\rho\log(1/\delta)}$.

Finally, we define the two fundamental DP mechanisms used in \rap\ --- \gm\ and \rnm\ --- and state their DP guarantees in terms of zCDP.
The first mechanism is the Gaussian mechanism, which we restate here in terms of zCDP and for the particular use case of answering a single statistical query.
\begin{definition}
The Gaussian mechanism \gm$(D, q_i, \rho)$ takes as input a dataset $D \in \cX^n$, a statistical query $q_i$, and a zCDP parameter $\rho$.
It outputs $a_i = q_i(D) + Z$, where $Z \sim \textrm{Normal}(0,\sigma^2)$ and $\sigma^2 = \frac{1}{2n^2\rho}$.
\end{definition}
\begin{lemma}[\cite{bun2016concentrated}]
For any query $q_i$ and $\rho > 0$, the \gm$(D, q_i, \rho)$ satisfies $\rho$-zCDP.
\end{lemma}
The second fundamental mechanism that \rap\ uses is the Gumbel noise Report Noisy Max (\rnm) mechanism.
\begin{definition}
The Report Noisy Max mechanism \rnm$(D, \Delta, \rho)$ takes as input a dataset $D \in \cX^n$, a vector of real values $\Delta$, and a zCDP parameter $\rho$.
It outputs the index of the highest noisy value in $\Delta$; i.e., $i^* = \argmax_i \Delta_i + Z_i$, where each $Z_i \sim \textrm{Gumbel}\left(\frac{1}{\sqrt{2\rho|D|^2}}\right)$.
\end{definition}
\begin{lemma}[\cite{durfee2019practical}]
For any real vector $\Delta$ and $\rho > 0$, the \rnm$(D, \Delta, \rho)$ satisfies $\rho$-zCDP.\\
\end{lemma}

With these fundamental mechanisms and their zCDP guarantees defined, we are now able to formally reproduce Aydore et al.'s original theorem and proof of \rap's DP guarantee.
\begin{theorem}[\cite{aydore2021differentially}]
For any class of queries and surrogate queries $Q$ and $\hat{Q}$, and for any set of parameters $n'$, $T$, and $K$, the \rap\ mechanism satisfies $(\epsilon, \delta)$-DP.
\end{theorem}
\begin{proof}
First, consider the non-adaptive case where $T=1, K=m$.
Here, the sensitive dataset $D$ is only accessed via $m$ invocations of the Gaussian mechanism, each with privacy $\rho/m$.
Therefore, by the composition property of zCDP, \rap\ satisfies $\rho$-zCDP.
Thus, by our choice of $\rho$ in line 1, we conclude that \rap\ satisfies $(\epsilon, \delta)$-DP.

Next, assume $T > 1$.
\rap\ executes $T$ iterations of its loop, only accessing the sensitive dataset $D$ via the Adaptive Selection (\as) mechanism each iteration.
Thus, we seek to prove that the \as\ mechanism satisfies $\rho/T$-zCDP.
Each invocation of the \as\ mechanism receives as input the privacy parameter $\rho'=\rho/T$, and accesses the sensitive dataset via $K$ invocations of \rnm\ and $K$ invocations of \gm.
Each invocation of either mechanism ensures $\frac{\rho'}{2K}$-zCDP, and therefore by the composition property of zCDP, the total $2K$ mechanism invocations ensure $\rho'$-zCDP.
Thus, the \as\ mechanism satisfies $\rho/T$-zCDP.
Leveraging zCDP's composition property again, because \rap\ invokes \as\ $T$ times, \rap\ therefore satisfies $\rho$-zCDP.
Finally, by our choice of $\rho$ in line 1, we conclude that \rap\ satisfies $(\epsilon, \delta)$-DP.
\end{proof}

\section{Enhancing \rap's Evaluation} \label{sec:improving-evaluation}
In this section, we address our first two contributions in the setting where all queries are prespecified: we strengthen and clarify our understanding of \rap's utility by performing a thorough reproducibility study on two important aspects of Aydore et al.'s evaluation of \rap.
These two aspects are:
\begin{enumerate}
\item The benefit of \rap's adaptive component relative to its non-adaptive component was unclear in its initial evaluation. We conclusively determine and quantify this component's utility benefit, finding that it is crucial for enabling \rap\ to achieve high utility.
\item \rap\ was initially only evaluated on highly reduced portions of the query space. We instead evaluate \rap's utility across the entire query space, answering up to 50x more queries than in its initial evaluation.
\end{enumerate}
The first aspect is significant because it improves our understanding of how \rap's adaptivity parameters affect its utility and establishes whether \rap's adaptive component is necessary in order to achieve high utility.
The second aspect is important because \rap's initial evaluation on highly reduced portions of the query space yielded potentially biased utility results.
By instead evaluating \rap\ across the entire query space, we establish \rap's unbiased utility and determine what impact reducing the query space has on \rap's utility.
In order to evaluate both aspects, we must reimplement \rap\ from the ground up in order to improve its efficiency for evaluating large sets of prespecified queries.
We then use the new implementation to evaluate both aspects, clarifying the value of the \rap\ mechanism and thus improving its adoptability for practical uses.

To make the description of our improved evaluation precise, in Section~\ref{sec:present-err} we define the utility metric used by Aydore el al.\ and by the prior state-of-the-art mechanisms for answering prespecified queries, which we also use in our evaluations.
We then discuss in Section~\ref{sec:reevaluation-focus} the details and implications of the two aspects of Aydore et al.'s initial evaluation of \rap\ that we are improving upon.
In Section~\ref{sec:reimplementing}, we detail the particular obstacle in \rap's initial implementation which prevents its use for our improved evaluation.
To overcome this obstacle, we reimplement \rap\ from the ground up and make its implementation publicly available\footnote{\href{https://github.com/bavent/large-scale-query-answering}{{https://github.com/bavent/large-scale-query-answering}}.}.
Finally, in Section~\ref{sec:reevaluation-experiments}, we describe how we use our improved implementation to perform our enhanced evaluation of \rap.

With regards to the role of adaptivity in \rap, we not only find that it is crucial to achieving high-utility, we also quantitatively and definitively measure how \rap's adaptivity parameters ($T$ and $K$) affect its utility.
This motivates new, more efficient search strategies to find optimal $T$ and $K$ values, thus reducing \rap's computational burden and privacy cost in practice.
With regards to evaluating \rap\ on the full query space, we find that Aydore et al.'s initial evaluation of \rap\ on a reduced portion of the query space likely \textit{underestimated} \rap's utility.
This was due to their reduced query space having less ``sparsity'' in the query answers (i.e., a larger portion of the queries they evaluated had non-0 answers).
This finding motivates a new line of research on mechanisms for the separate cases of when query answers are and are not sparse.
Together, the improved \rap\ implementation combined with the enhanced evaluation clarifies the value of the \rap\ mechanism, and thus improves \rap's adoptability and usability in practice.

\subsection{Measuring Utility of Prespecified Queries} \label{sec:present-err}
We define the concrete utility measure used in prior works to evaluate DP mechanisms that answer prespecified sets of statistical queries.
Prior works in this setting measured the utility of DP mechanisms in terms of a mechanism's maximum error over the answers to all queries in the prespecified query set~\cite{mckenna2019graphical, vietri2020new, liu2021iterative, aydore2021differentially}.
We refer to this measure of utility as \textit{present utility}, since it is the error on the set of presently available queries, and measure it in terms of the negative of \textit{present error}; i.e., a mechanism with low present error has high present utility, and vice versa.
This error measure is formally defined as follows.
\begin{definition}[Present error] \label{def:present-err}
Let $a = Q(D) = (a_1,\dots,a_m)$ be the true answers to a given query vector $Q$ on dataset $D$, and let $\tilde{a} = (\tilde{a}_1,\dots,\tilde{a}_m)$ be mechanism $M$'s corresponding answers to the query vector.
Then $\err_P$ is the present error of the mechanism, defined as $\err_P(M,D,Q) = \E_{M(D)} \Vert a - \tilde{a}\Vert_\infty$, where the expectation is over the randomness of the mechanism.
\end{definition}

We choose the $\ell_\infty$ norm as the base metric for present error because of its use in Aydore et al.'s evaluation of \rap\ and because it is the most popular norm utilized in the most closely related literature~\cite{mckenna2019graphical, vietri2020new, liu2021iterative, aydore2021differentially}.
However, other norms (e.g., $\ell_1$ and $\ell_2$) and even definitions of error may be equally valid in the prespecified queries setting depending on the practical use case~\cite{tao2021benchmarking}.
Thus, although we do not empirically evaluate \rap\ on such alternative definitions, investigating how the findings in this work change based on the error definition is an excellent direction for future work.

\subsection{Focus of \rap's Reevaluation} \label{sec:reevaluation-focus}
We now detail the two primary aspects of Aydore et al.'s evaluation of \rap\ that we enhance in this work, and how their origins trace back to a particular challenge in \rap's initial implementation.

\paragraph{Adaptivity Evaluation:}
The first aspect that we address in \rap's reevaluation is how \rap's adaptive component affects its utility.
To provide context, we briefly describe the non-adaptive form of \rap.
We then describe the adaptive form of \rap\ and the motivation behind its design.
Finally, we detail how Aydore et al.'s evaluation of \rap\ omitted studying the adaptive component's effect on utility, and we describe why that is an issue.

In its non-adaptive form, the \rap\ mechanism essentially reduces to privately answering the full query vector $Q$ with the Gaussian Mechanism, then applying the \rp\ mechanism to generate a synthetic dataset.
This non-adaptive form of the \rap\ mechanism is a novel reimagining of the classic \textit{Projection Mechanism}~\cite{nikolov2013geometry}, a near-optimal but computationally intractable mechanism for answering prespecified queries.
By leveraging a relaxation of the query space and utilizing EEDQs, Aydore et al.\ describe how their non-adaptive \rap\ mechanism can use modern tools (e.g., GPU-accelerated optimization) to efficiently generate a relaxed synthetic dataset which can hypothetically answer the prespecified queries with low (albeit non-optimal) error.
Moreover, they prove a theoretical result (Theorem 4.1,~\cite{aydore2021differentially}) which confirms the power of the non-adaptive \rap\ mechanism, achieving a $\sqrt{d'}$ factor of utility improvement over the prior state-of-the-art mechanism.

Aydore et al.\ go on to describe the full adaptive form of \rap\ parameterized by $T$ and $K$.
This adaptive form of \rap\ optimizes the synthetic dataset iteratively over $T$ separate rounds, in each round adaptively selecting $K$ new queries to incorporate into the optimization procedure.
Their stated motivation for introducing adaptivity into \rap\ was to more wisely expend the privacy budget by adaptively optimizing over a small number of ``hard'' queries, and they conjecture (without a result similar to that of their Theorem 4.1) that such adaptivity will result in higher utility than that achieved by the non-adaptive form of \rap.

Aydore et al.\ then perform an empirical evaluation of \rap\ across a range of parameters and datasets, and establish that it achieves state-of-the-art utility --- however, the utility benefits of \rap's adaptivity are left unanalyzed.
Specifically, in all evaluations they report the best utility of \rap\ across $2 \le T \le 50$ and $5 \le K \le 100$.
There are two issues related to this.
\begin{enumerate}
\item The values of $T$ and $K$ that achieved the maximum utility are not reported, only what that maximum utility was. Thus, it is unclear how these parameters affect utility. This is problematic in practice because not only is evaluating \rap\ on multiple choices of $T$ and $K$ computationally expensive, but because each evaluation consumes a portion of the overall differential privacy budget.
\item The non-adaptive form of \rap\ is not empirically evaluated. Without evaluating the non-adaptive \rap\ mechanism as a baseline, there is no meaningful way to understand or measure the benefit of adaptivity.
\end{enumerate}
Combined, these two issues leave open the question of how valuable the adaptive component of \rap\ is, and to what extent its adaptivity affects utility.

\paragraph{Query Space Evaluation:}
The second aspect that we address in \rap's reevaluation is how reducing the query space affects \rap's utility for answering $k$-way marginals.
To begin, we describe the motivation behind evaluating this aspect: that for computational ease, Aydore et al.\ only evaluated \rap\ on a reduced portion of the query space.
We then detail how this reduction may have biased their evaluation's results.

Aydore et al.'s empirical evaluation focuses on \rap's utility for answering $k$-way marginals, specifically 3-way and 5-way marginals.
Reviewing the code of their published \rap\ implementation, we determined that a heuristic filtering criterion of the query space was being applied to remove any ``large'' marginals from possible evaluation.
Specifically, any marginal which had more consistent queries than the number of records in the dataset ($n$) was not considered for evaluation.
The impact that filtering had on the evaluated workloads varied depending on $k$ and $n$.
For instance, with 3-way marginals on the ADULT dataset, the filtering criterion removed the top $24\%$ largest 3-way marginals which accounted for over $90\%$ of all consistent queries.
With 5-way marginals on the ADULT dataset, this filtering criterion removed the top $92\%$ largest 5-way  marginals which accounted for over $99.99\%$ of all consistent queries.

Discussing this discrepancy directly with the authors\footnote{\href{https://github.com/amazon-research/relaxed-adaptive-projection/issues/2}{https://github.com/amazon-research/relaxed-adaptive-projection/issues/2}} revealed that the filtering criterion was an intentional choice meant to reduce the computational burden during experimentation, and they conjectured that removing this criterion and rerunning all experiments would yield results comparable to those obtained by increasing the workload size.
Since all baseline mechanisms were evaluated on the same query vectors, the filtering criterion does not result in favorable utility for \rap\ relative to the prior state-of-the-art mechanisms that serve as their baselines.
However, for marginals with a significantly larger number of consistent queries than $n$, most queries will evaluate to 0 by a Pigeonhole principle argument.
Thus, the filtering criterion may result in favorable utility for \rap\ relative to the naive baseline mechanism that they consider in their work: \bma, the mechanism which outputs 0 as the answer to every query.
This leaves open the question of \rap's utility on large, unfiltered query spaces, both in absolute terms and relative to the baseline \bma\ mechanism.

\subsection{Reimplementing \rap} \label{sec:reimplementing}
We now describe why these two aspects cannot be evaluated using Aydore et al.'s initial \rap\ implementation: briefly, the amount of memory required by the implementation is inordinate.
We then detail how we overcome this challenge by reimplementing \rap\ in a way that trades-off a significant amount of memory usage for a potential increase in runtime.

Conceptually, both aspects could be evaluated using Aydore et al.'s published code.
However, evaluating either the non-adaptive form of \rap\ or evaluating a larger portion of the query space both lead to the same obstacle: Aydore et al.'s \rap\ implementation requires an inordinate amount of memory to answer the corresponding large number of queries.
We have identified several portions of their code where this memory bottleneck occurs, all of which fail to execute either when the total number of consistent queries is ``too large'' or when any marginal has ``too many'' consistent queries.
Consequently, Aydore et al.\ were unable to evaluate either the non-adaptive form of \rap\ or a significant portion of the $k$-way marginals' consistent query space.

The high-level idea behind our approach for overcoming this implementation challenge is to trade-off some of \rap's required memory for a potential increase in its runtime.
Our motivation for this approach is inspired by recent advances in differentially private deep learning literature.
In particular, the canonical DP-SGD mechanism~\cite{abadi2016deep} for training machine learning models with differential privacy had been plagued by poor computational performance due to several of its underlying operations (e.g., per-example gradient clipping, uniformly random batch sampling without replacement, etc.) not being natively supported by modern machine learning frameworks.
More recently however, several highly performant DP-SGD implementations~\cite{papernot2019machine, opacus, subramani2021enabling} have been deployed which dramatically decrease the mechanism's runtime in exchange for a mild increase in its memory usage.
To our knowledge, our high-level approach is the first in DP literature to make practical use of this trade-off in the opposite direction: decreasing the mechanism's memory requirement by increasing its runtime.

Concretely, we overcome this implementation challenge by reimplementing \rap\ via the following high-level steps.
First, we reduce the maximal memory requirement in \rap's original implementation caused by the original implementation's implicit evaluation all marginals (or, more generally, all thresholds) in parallel.
We accomplish this by evaluating each marginal (or threshold) sequentially in order to distribute the computational burden.
To further reduce the overall memory requirement, rather than explicitly enumerating and storing every query consistent with each marginal (threshold), we represent the queries implicitly and only convert a query to its explicit representation when it is needed for evaluation.
To evaluate arbitrary sets of such individual queries, we implement the core EEDQ evaluation function from the ground up by designing a simple, direct function to efficiently evaluate arbitrary predicates.
With such a function implemented, we then leverage a combination of powerful language features --- namely vectorizing maps and just-in-time compilation in JAX~\cite{jax2018github} --- to enable efficient evaluation, summation, and differentiation of large sets of predicates without exceeding memory constraints.

In addition to these implementation improvements which primarily serve to reduce \rap's memory requirement, we additionally incorporate an algorithmic improvement based on recent theoretical findings to help offset the increased runtime from our aforementioned deparallelization step.
Specifically, by trivially adapting the \textit{Oneshot Top-$K$ Selection with Gumbel Noise} mechanism~\cite{durfee2019practical, cesar2021bounding} to our setting, we replace \rap's iterative Adaptive Selection (\as) mechanism with the more efficient Oneshot Adaptive Selection (\osas) mechanism in Alg.~\ref{alg:osas}.
The results of \cite{durfee2019practical} prove that the \osas\ mechanism is probabilistically equivalent to \as\ (i.e., both mechanisms have identical output distributions, and thus achieve identical privacy and utility), but \osas\ requires only 1 pass over a set of values in order to select the top-$K$ instead of the $K$ passes that \as\ requires.

\begin{algorithm}
\caption{Oneshot Adaptive Selection (\osas) Mechanism}
     \vspace{0.4em} \hspace*{\algorithmicindent} \textbf{Input} \vspace*{-0.6em}
     \begin{itemize}[leftmargin=1.5em] \setlength\itemsep{-0.2em}
     	\item $D, D'$: The dataset and synthetic dataset.
     	\item $Q, \hat{Q}$: A vector of all statistical queries and their corresponding surrogate queries.
     	\item $Q_s$: A set of already selected queries.
     	\item $K$: The number of new queries to select $K$.
     	\item $\rho$: Differential privacy parameter.
     \end{itemize}
     \hspace*{\algorithmicindent} \textbf{Body}
\begin{algorithmic}[1]
   	\STATE Let $\Delta = (|\hat{q}_{i}(D) - \hat{q}_{i}(D')| : q_i \in Q \setminus Q_s)$.
   	\STATE Let $I$ denote the indices of the top-$K$ values of: $\Delta_i + Z_i$, where $Z_i \overset{\text{iid}}{\sim} \texttt{Gumbel}\left(\sqrt{\frac{K}{2\rho |D|^2}}\right)$.
   	\STATE Let $\tilde{a}_i = \gm(D, q_i, \frac{\rho}{2K}) \quad \forall i \in I$.
	\STATE Let $Q_s = Q_s \cup \{q_i\}_{i \in I}$.
    \STATE {\bfseries Return:} $Q_s$ and $\tilde{a} = (\tilde{a}_i: q_i \in Q_s)$.
\end{algorithmic}
\label{alg:osas}
\end{algorithm}

Figure~\ref{fig:rap-runtime} compares our new implementation to Aydore et al.'s original implementation without filtering out any large marginals.
Specifically, this figure shows the runtimes of both implementations executing the non-adaptive and adaptive variants of \rap\ given the same amount of GPU memory on two datasets across a range of workload sizes\footnote{The runtimes for both implementations (and all subsequent evaluations in this work) were executed on an Nvidia RTX 3090 consumer GPU with 24 GB VRAM.}.
We find that for the non-adaptive variant of \rap, the original implementation was only able to evaluate tiny workloads, while our new reimplementation was able to evaluate massive workloads (albeit, with a very high runtime); this represents a 500x improvement in memory efficiency for our reimplementation.
For the adaptive variant of \rap\ (specifically, with $T$=16 and $K$=4), we find the our reimplementation's runtime is comparable to the original implementation's --- outperforming it slightly on one dataset, while being outperformed slightly on the other.
On the ADULT dataset, both implementations were able to exhaustively evaluate the complete space of marginals.
On the LOANS dataset, the original implementation was able to consistently evaluate marginal workloads of size 256, but was unable to consistently evaluate the largest workload size of 1024; this represents up to a 4x improvement in memory efficiency for our reimplementation.

\begin{figure}
\centering
\begin{tabular}{c|c}
\hspace*{-0.5cm}
  \includegraphics[width=.4\linewidth]{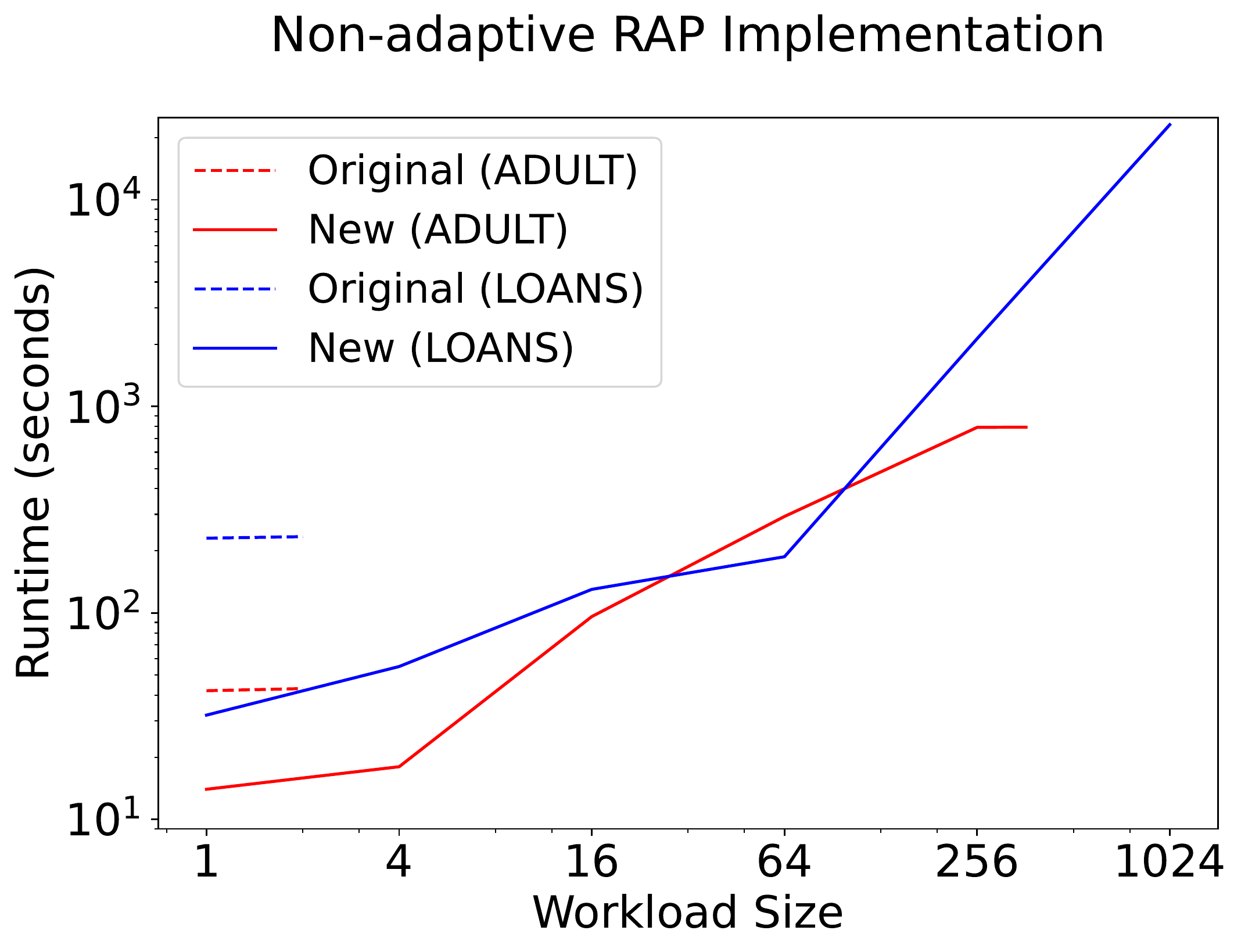} & \includegraphics[width=.4\linewidth]{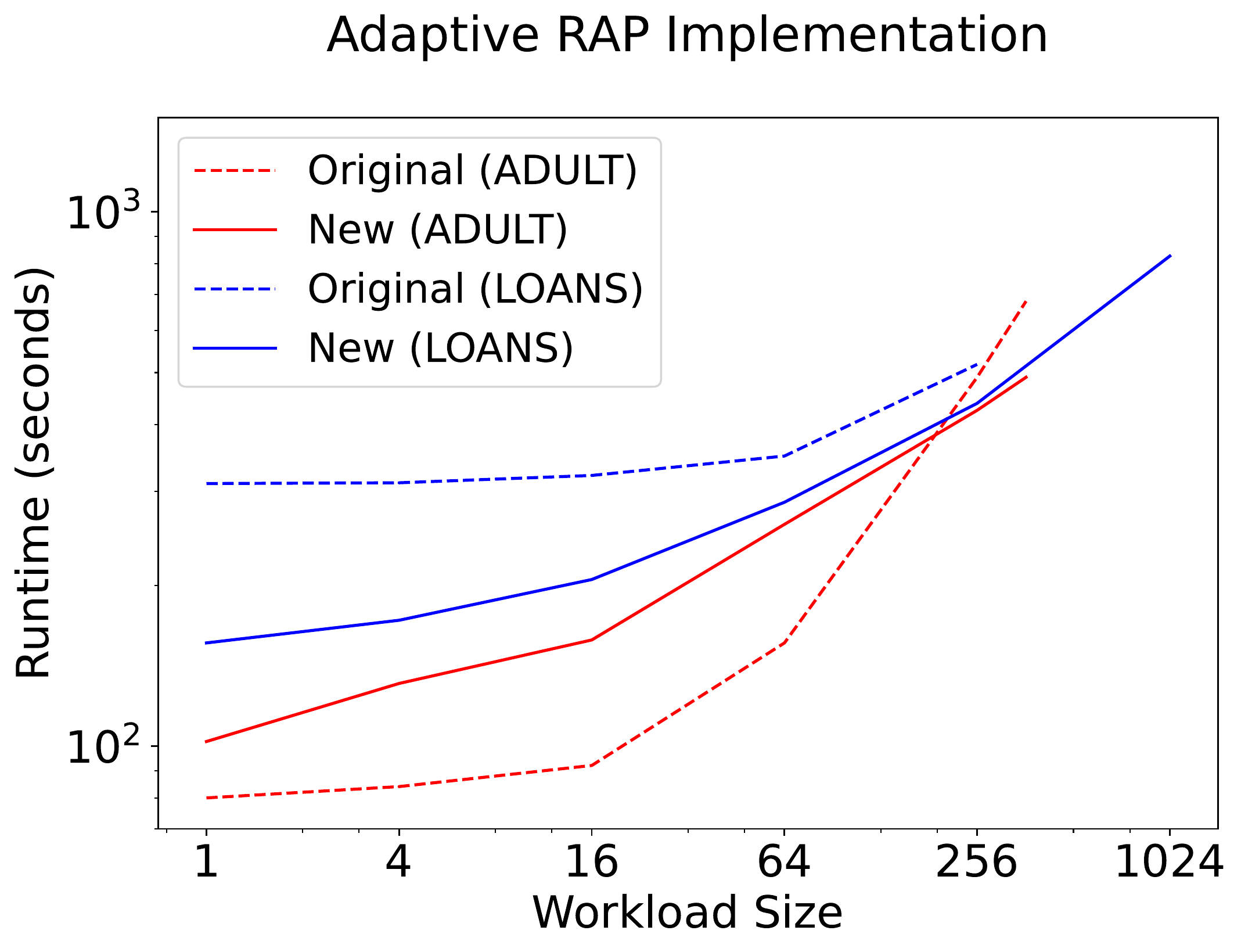}
\end{tabular}
\caption{Runtime evaluations of non-adaptive and adaptive \rap\ variants on the original implementation and reimplementation, on both ADULT and LOANS datasets.}
\label{fig:rap-runtime}
\end{figure}

\subsection{Reevaluating \rap} \label{sec:reevaluation-experiments}
Using our new implementation, we reevaluate both the adaptivity and query space aspects of \rap, enabling new findings.
We start by simply establishing \rap's present utility for answering $k$-way marginals on unbiased random samples of the full marginal space (i.e., without filtering out any ``large'' marginals).
This results in \rap\ answering approximately 50x more queries at its peak than in Aydore et al.'s initial evaluation on filtered marginals.
We then use these results to analyze the role that adaptivity plays in \rap's utility.
Finally, we address the question of whether filtering the large marginals out of \rap's evaluation significantly impacts its utility in order to determine if the filtering criterion is a reasonable heuristic to apply to reduce \rap's computational burden in future evaluations.
This improved implementation and reevaluation, taken together, conclusively demonstrates that \rap\ is a feasible and valuable mechanism for practical, real-world use cases.
Furthermore, in conjunction with our improved implementation, our findings enable new capabilities such as more efficient search strategies for optimal $T$ and $K$ parameters.

\subsubsection*{Evaluation Datasets} \label{sec:rap-datasets}
As in prior works on evaluating DP mechanisms that answer statistical queries \cite{aydore2021differentially, vietri2020new, mckenna2019graphical}, all empirical evaluations use the ADULT~\cite{frank2010uci} and LOANS~\cite{vietri2020new} datasets with the same preprocessing.
Table~\ref{tab:rap-datasets} contains a high level description of each dataset.

\begin{table}
    \centering
    {\begin{tabular}{c|c|c|c}
    \bf{Dataset} & \bf{Records} & \bf{Features} & \bf{Binarized Features} \\
     \hline
     ADULT &  48,842 &  14 & 588 \\
     LOANS &  42,535 &  48 & 4,427 \\
    \end{tabular}}
\caption{Datasets for empirical evaluations. Binarized features represent the features after a transformation via one-hot encoding.}
\label{tab:rap-datasets}
\end{table}

\subsubsection{$k$-way Marginal Evaluation of \rap}
To begin \rap's reevaluation, we concretely establish its utility on a larger portion of the query space than previously considered by Aydore et al.
Specifically, we evaluate \rap's present error for answering uniformly random workloads of 3-way marginals across a range of parameters on both the ADULT and LOANS datasets, and we do so \textit{without any} thresholding criterion to filter out ``large'' marginals.
This results in \rap\ answering approximately 50x as many queries as in its original evaluation by Aydore et al.
Table~\ref{tab:evaluation-experiments} provides a reference for the parameter ranges in this experiment.
For each setting of parameters, we evaluate the adaptive variant of \rap\ across a range of $T$ and $K$ values and report the combinations that achieve minimal present error.
We separately evalaute the non-adaptive ($T=1, K=m$) variant of \rap\ across the same range of parameters in order answer the question of whether or not there is any benefit to \rap's adaptivity.
Additionally, as baselines, we evaluate the present utility of the \bma\ and \gm\ mechanisms, enabling us to put the utility of \rap\ into context.
The results of this experiment are visualized in Figure~\ref{fig:reevaluating-lines}.

\begin{table}
\centering
\begin{tabular}{ |c|c| } 
 \hline
 Primary Mechanism & \rap \\
 \hline
 Baseline Mechanisms & $\bma, \gm$ \\
 \hline
 Utility Measure & $\err_P$ \\
 \hline
 $D$ & ADULT, LOANS \\
 \hline
 $\epsilon$ & $0.01, 0.1, 1$ \\
 \hline
 $\delta$ & $1/|D|^{2}$ \\
 \hline
 $|W|$ & $1, 4, 16, 64, 256$ \\
 \hline
 $n^\prime$ & $10^3$ \\
 \hline
 $T$ & $1, 4, 16, 64$ \\
 \hline
 $K$ & $4, 16, 64, 256, m$ \\ 
 \hline
 $k$ & $3$ \\
 \hline
\end{tabular}
\caption{Experimental reference table for reevaluating \rap's utility on $k$-way marginals.}
\label{tab:evaluation-experiments}
\end{table}

\begin{figure}
\centering
\begin{tabular}{c|c}
\hspace*{-0.5cm}
  \includegraphics[width=.4\linewidth]{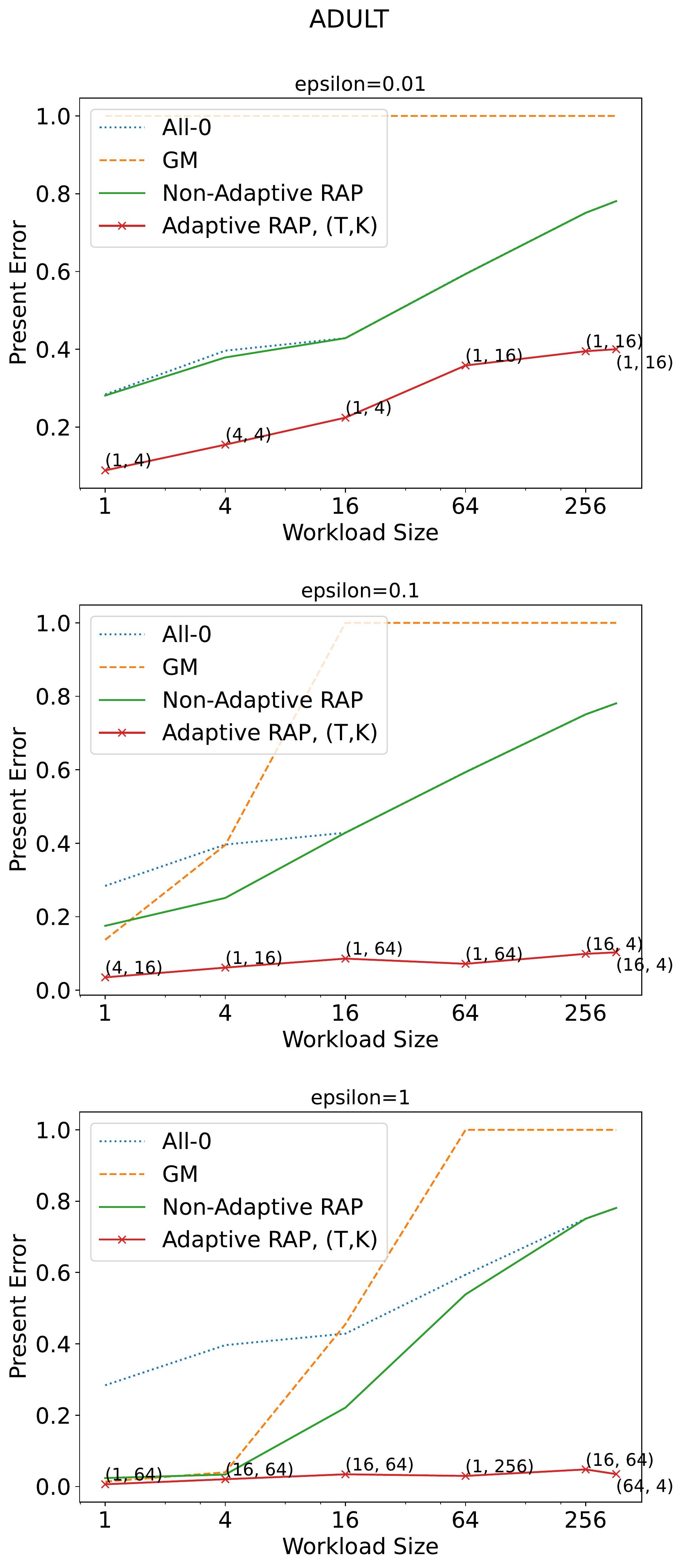} & \includegraphics[width=.4\linewidth]{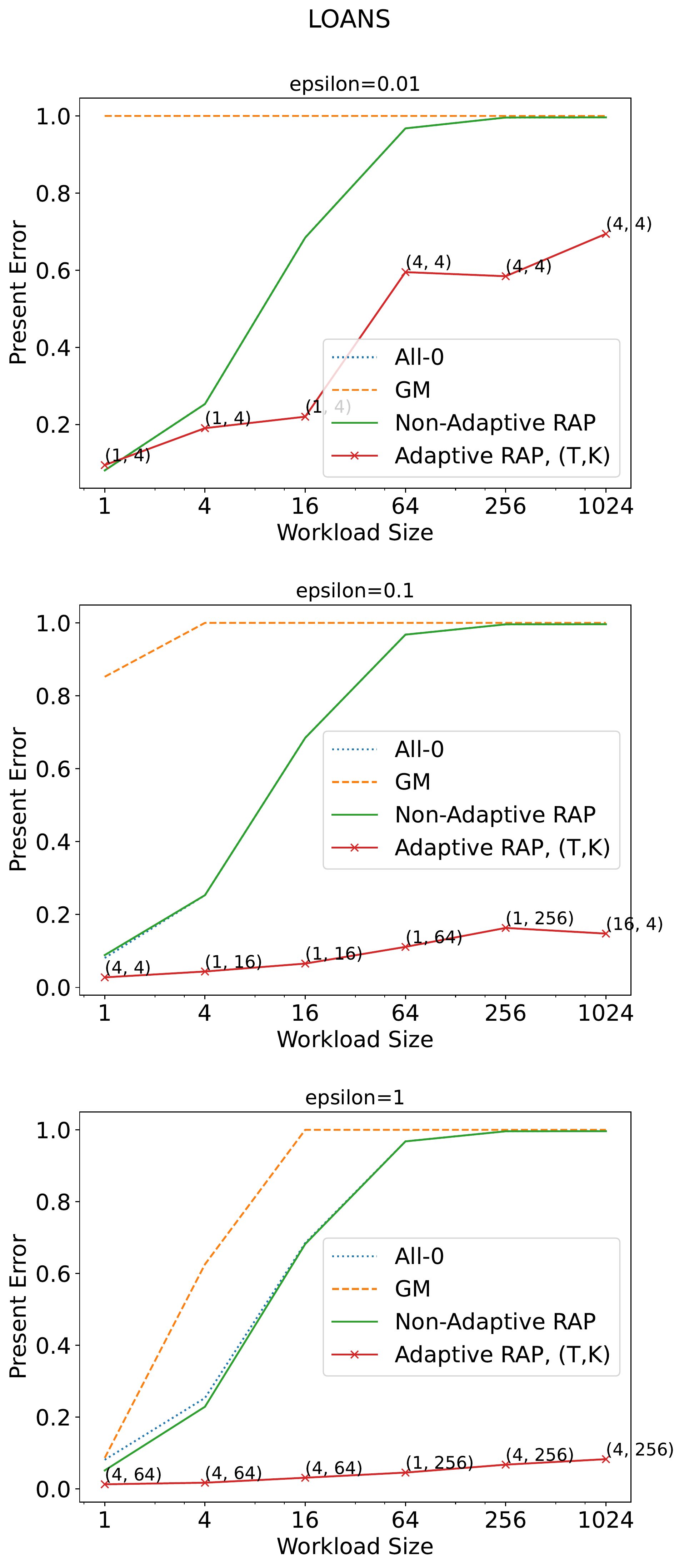}
\end{tabular}
\caption{Present error across a range of parameters and datasets for the adaptive and non-adaptive variants of \rap, the \gm\ baseline, and the \bma\ baseline. Present error for the adaptive variant of \rap\ is computed as the minimal error across the range of $T$ and $K$ values (with the specific $(T,K)$ pair that achieved the minima reported at each point).}
\label{fig:reevaluating-lines}
\end{figure}

There are several immediate conclusions that can be drawn from these results.
The first is that while the non-adaptive variant of \rap\ achieves lower error than the \gm\ baseline, its utility is nearly identical to the \bma\ baseline for all but the smallest workload sizes.
This result likely stems from the fact that the answers to the large majority of a marginal's consistent queries are 0 or nearly 0, with only a small percentage of answers having larger values.
Since the non-adaptive variant of \rap\ first privatizes the answers to all queries, in the synthetic dataset optimization procedure it is likely unable to distinguish between the few answers that are truly larger than 0 vs.\ the outliers that are only large due to random chance.
The second conclusion is that the adaptive variant of \rap\ achieves significantly lower present error than the non-adaptive \rap\ variant as well as the baselines.
This implies that \rap's adaptivity is critical for achieving low error, and thus warrants a more thorough investigation into $T$ and $K$'s precise impact on utility.

\subsubsection{Role of Adaptivity} \label{sec:reevaluating-role-of-adaptivity}
In this next experiment, we seek to understand the precise impact that $T$ and $K$ have on \rap's utility.
From Figure~\ref{fig:reevaluating-lines}, we are only able to glean that \rap\ typically achieves minimal error via smaller values of $T$ in conjunction with relatively larger values of $K$.
However, these values of $T$ and $K$ vary dramatically across parameter settings and datasets.
Moreover, Figure~\ref{fig:reevaluating-lines} provides no information about \rap's utility for $T$ and $K$ combinations that did not achieve minimal error.
To better understand the role these parameters play in \rap's utility, we examine the present error of the adaptive variant of \rap\ for every $(T,K)$ pair across the same parameter settings from Table~\ref{tab:evaluation-experiments}.
The results of this experiment are shown in Figures~\ref{fig:reevaluating-adaptivity-ws} and \ref{fig:reevaluating-adaptivity-eps}.

\begin{figure}
\centering
\begin{tabular}{c}
\hspace*{-0.65cm}
  \includegraphics[width=1.05\linewidth]{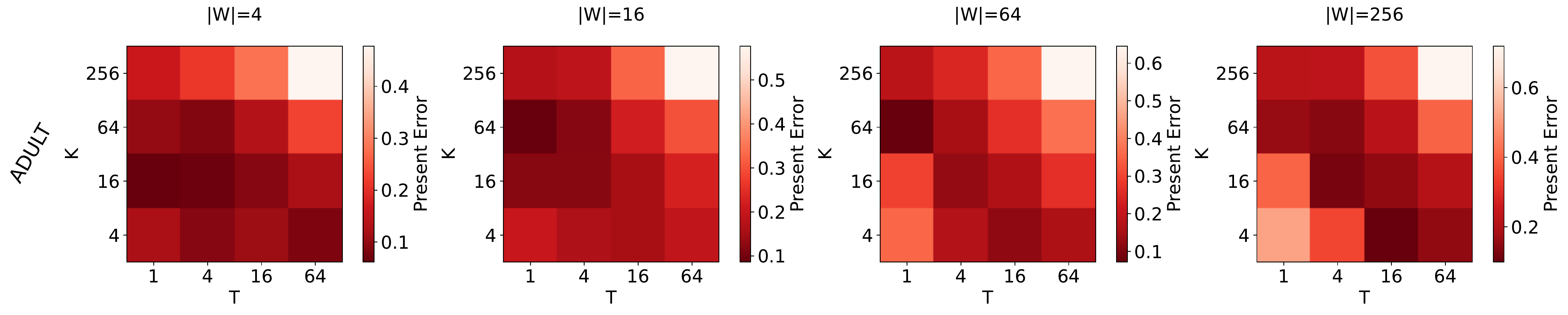} \\
\hline
\hspace*{-0.65cm}
  \includegraphics[width=1.05\linewidth]{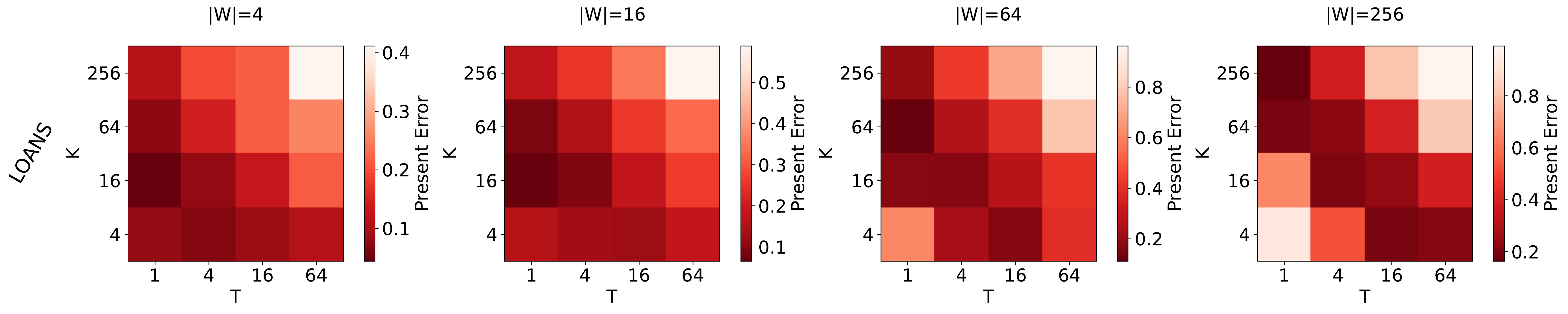}
\end{tabular}
\caption{Present error across a range of workload sizes with $\epsilon = 0.1$ for the adaptive variant of \rap\ at every combination of $T$ and $K$ value considered.}
\label{fig:reevaluating-adaptivity-ws}
\end{figure}

\begin{figure}
\centering
\begin{tabular}{c}
\hspace*{-0.65cm}
  \includegraphics[width=\linewidth]{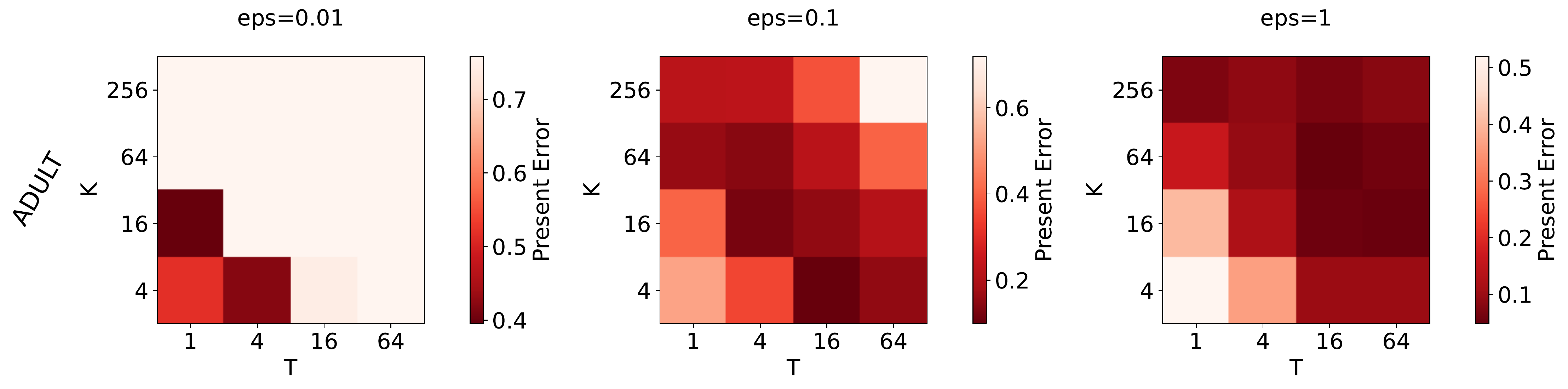} \\
\hline
\hspace*{-0.65cm}
  \includegraphics[width=\linewidth]{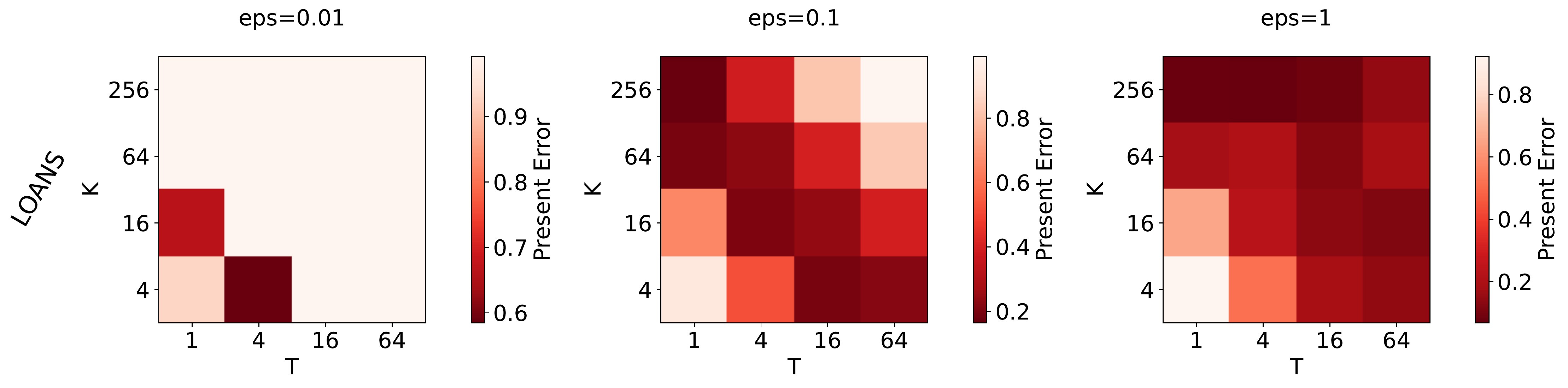}
\end{tabular}
\caption{Present error across a range of $\epsilon$ values with $|W|=256$ for the adaptive variant of \rap\ at every combination of $T$ and $K$ value considered.}
\label{fig:reevaluating-adaptivity-eps}
\end{figure}

The heatmaps in both figures provide interesting insight into \rap's adaptivity.
In Figure~\ref{fig:reevaluating-adaptivity-ws}, with $\epsilon$ fixed at $0.1$, we see that there is no single $(T,K)$ value or region that consistently achieves minimal error across all workload sizes.
Instead, we notice that at each workload size, there is some diagonal banding at around a fixed region of $T \cdot K$ that achieves approximately minimal error.
That is, for any particular workload size, let $(T^*, K^*)$ denote the $T$ and $K$ value that induces minimal error for \rap\ across our considered range of $T,K$ values, and let $c^*\coloneqq T^* \cdot K^*$.
We see that for other $(T,K)$ pairs such that $T \cdot K \approx c^*$, the corresponding error is typically comparable to the minimal error.
Moreover, we see that as $T\cdot K$ diverges from $c^*$, \rap's error increases essentially monotonically.
We hypothesize that for $T \cdot K \ll c^*$, \rap's error is relatively high because \rap\ had not answered and optimized over a sufficient number of queries.
For $T \cdot K \gg c^*$, we hypothesize that \rap's error is relatively high because the privacy budget is spread too thin across across answering a large number of queries, resulting in \rap\ utilizing overly noisy queries to optimize its underlying synthetic dataset.

These hypotheses are supported by the results in Figure~\ref{fig:reevaluating-adaptivity-eps}.
Specifically, as $\epsilon$ becomes larger, not only does the minimal error of \rap\ decrease, but the $T$ and $K$ values that achieve the minimal error (along with their corresponding diagonal bands) are pushed to increasingly large values.
Taken together, these results imply that in order to achieve low error, \rap\ primarily requires answering and optimizing over a \textit{specific number} of queries --- it is less important whether those queries are answered in small batches over a large number of adaptive rounds or in large batches over a small number of adaptive rounds.

This finding is important to \rap's usefulness in practice, as it motivates improved search strategies for optimal $(T,K)$ values.
Improved search strategies (beyond the naive $N\times N$ grid search that we performed) are important for two reasons.
\begin{enumerate}
\item Evaluating \rap\ across a range of $T$ and $K$ values can be computationally expensive. Thus, improved search strategies would decrease the computational cost. Alternatively, at a fixed computational cost, improved search strategies would allow \rap\ to be evaluated across a larger set of $T$ and $K$ values.
\item In practice, each evaluation of \rap\ on any $(T,K)$ setting consumes a portion of the privacy budget, even if only the optimal setting is chosen in the end. Thus, reducing the total number of evaluated $(T,K)$ settings enables more efficient use of the overall privacy budget.
\end{enumerate}
We provide one example of an improved search strategy over the naive $N\times N$ grid search strategy as follows.
First, the observed monotonicity of present error about $c^*$ could be leveraged to binary search for a $c \coloneqq T\cdot K$ setting along the positive diagonal that achieves approximately minimal error.
Then, a linear search across all $(T^\prime, K^\prime)$ settings such that $T^\prime \cdot K^\prime = c$ could be performed to compute the setting that achieves minimal error.
Relative to the grid search, this strategy would yield an $\mathcal{O}(N)$ factor improvement both in the portion of the privacy budget consumed as well as in the computational cost.

\subsubsection{Utility Impact of Filtering Marginals}
In the final experiment, we analyze what impact filtering out marginals with ``too many'' consistent queries has on \rap's utility.
Recall that in Aydore et al.'s evaluation, as a heuristic to reduce the computational burden of experimentally evaluating \rap, any marginal was removed from consideration if it contained more consistent queries than the number of records in the underlying dataset.
Here, we compare how \rap's utility is affected by this marginal filtering criterion.
We initiate this comparison by reevaluating \rap\ with and without the filtering criterion.
We do so across the range of parameters in Table~\ref{tab:evaluation-experiments}, and we record the minimal present error of \rap\ at each parameter setting across all $(T,K)$ pairs.
We then perform two analyses on these results, one focusing on how the workload size affects \rap's present error with and without marginal filtering, and another analyzing how the total number of queries affects \rap's present error. 
We conclusively determine that \rap's present error is impacted by filtering large marginals.
More specifically, we find that when holding the number of queries that \rap\ evaluates constant, filtering large marginals \textit{increases} \rap's present error.

\paragraph{Influence of Workload Size on Utility}
Aydore et al.\ hypothesized that removing the marginal filtering criterion would cause \rap's present error to increase comparably to the error increase induced by increasing the workload size.
To test this hypothesis, we perform a standard nested regression analysis~\cite{gelman2006data} on the \rap\ evaluation results.
For brevity, we state the steps of this analysis and then immediately jump to the results, deferring the regression details to Appendix~\ref{app:many-queries-appendix}.

At the high level, the steps for this analysis are as follows.
For the ADULT and LOANS datasets separately, we define a full regression model to account for the following three variables' (and their interactions') impact on \rap's present error: the DP level $\epsilon$, the workload size $|W|$, and whether the marginal filtering criterion was applied.
We also define a restricted regression model that accounts for $\epsilon$ and $|W|$, but does not distinguish whether or not a result had the marginal filtering criterion applied.
Following the standard approach for a nested regression analysis, we first determine whether the full regression model is a good fit for the \rap\ evaluation results (based on the fitted model's adjusted $r^2$ value, $F$-statistic $p$-value, and omnibus $p$-value).
We then compare the fit of the full model to the fit of the restricted model by performing a likelihood ratio test, analyzing the $p$-value of the resulting $\chi^2$ statistic.
Since the full model only differs from the restricted model in that it accounts for whether the marginal filtering criterion was applied, we can conclude that if the fit of the full model is both statistically sound and statistically significantly better than that of the restricted model, then the marginal filtering criterion impacts \rap's present error.

\begin{figure}
\centering
\begin{tabular}{cc}
\hspace*{-0.65cm}
  \includegraphics[width=.45\linewidth]{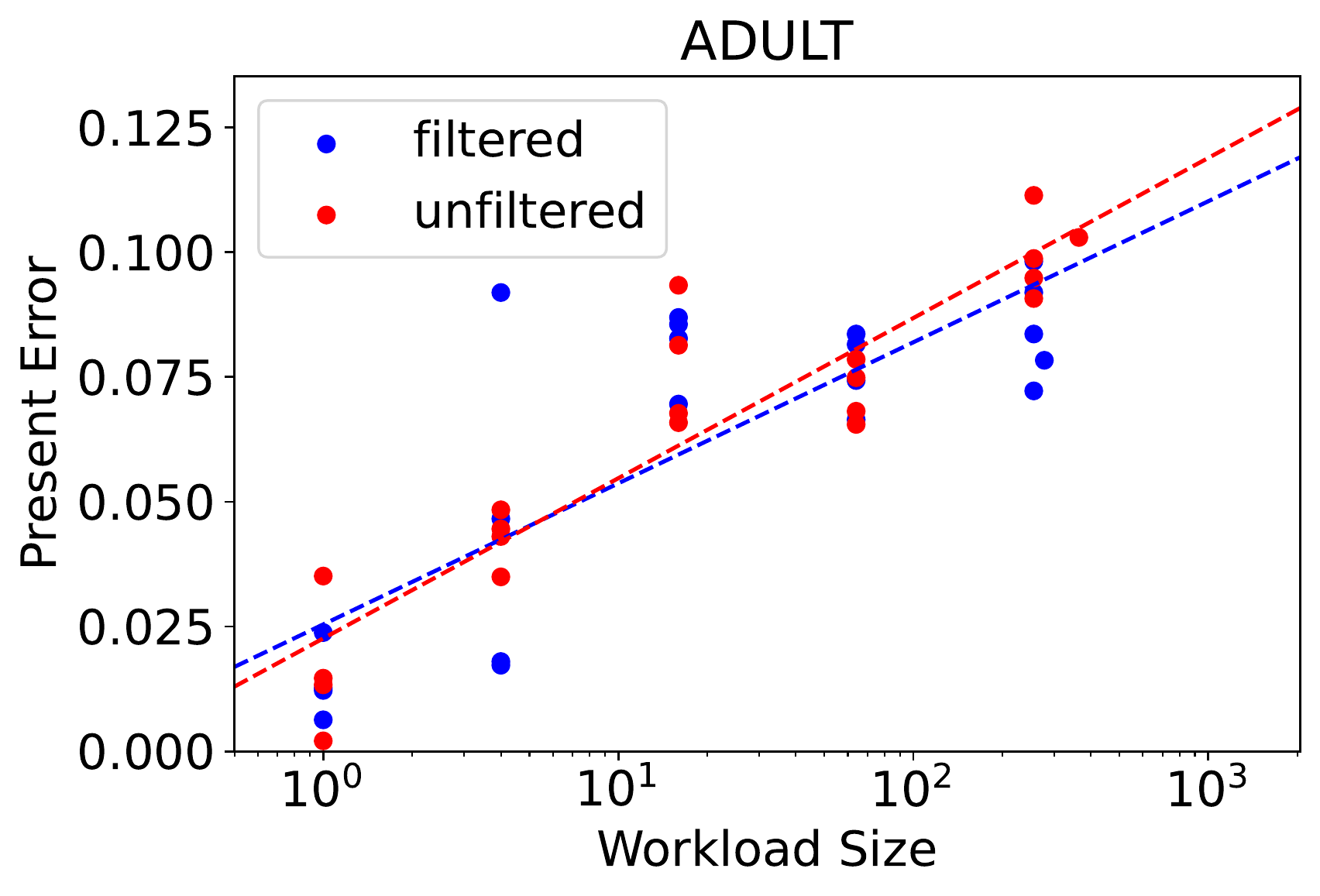} & \includegraphics[width=.45\linewidth]{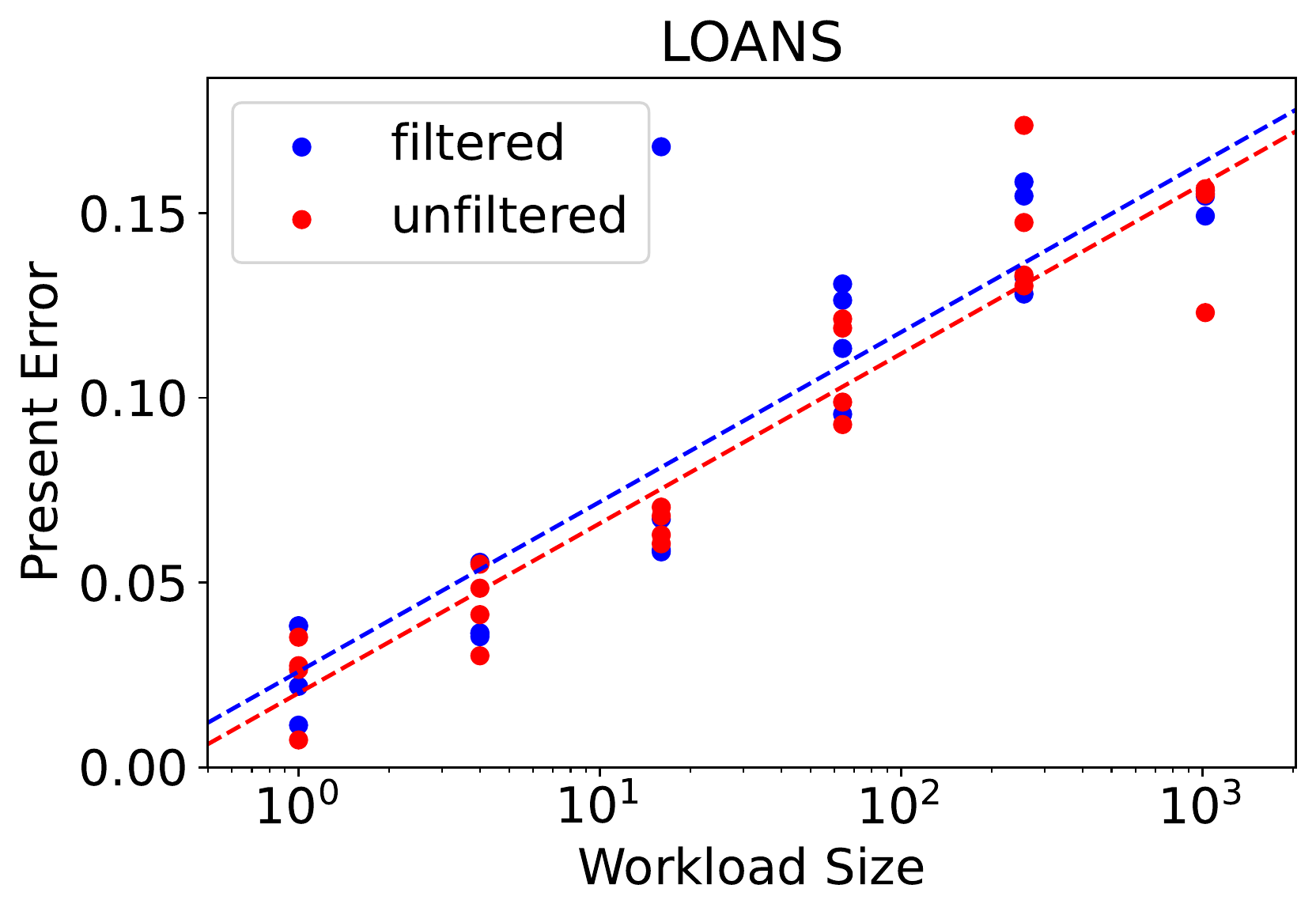}
\end{tabular}
\caption{Regression models for each dataset of \rap's present error vs.\ workload size for results from filtered and unfiltered marginals, at $\epsilon=0.1$.}
\label{fig:ws-regressions}
\end{figure}

From this analysis, Figure~\ref{fig:ws-regressions} shows the fitted full regression model on both datasets with $\epsilon$ fixed at $0.1$.
We find that the full regression models for both datasets fit the \rap\ evaluation results well.
Thus, we perform the aforementioned likelihood ratio test against the restricted models for each dataset.
The corresponding $p$-values for the models on the ADULT and LOANS \rap\ evaluations were $0.026$ and $0.623$ respectively.\footnote{We report the individual $p$-values for all statistical hypotheses tested. However, we control the family-wise error rate $\alpha$ (i.e., the probability $\alpha$ that at least one ``false positive'' finding will occur) using the Holm–Bonferroni method~\cite{holm1979simple}. At the $\alpha=0.05$ level, no conclusions based on the individual $p$-values change when the Holm–Bonferroni method is applied.}\enlargethispage{-\baselineskip}
The small $p$-value for the model corresponding to the \rap\ evaluations on the ADULT dataset enables us to conclude that the marginal filtering criterion does have an impact on \rap's present error.
However, the coefficients (and their corresponding $p$-values) in the full regression model do not indicate any clear, statistically significant trend for how the present error is impacted by the workload size when comparing the filtered vs.\ unfiltered \rap\ evaluations.
Moreover, regardless of the workload size, due to the lack of significance in many of the coefficients' $p$-values, we are unable to use this model to confidently determine the marginal filtering criterion's impact on \rap's present error.
Thus, although we are able to conclude that incorporating the marginal filtering criterion into \rap's evaluation does impact its present error, we are unable to confirm Aydore et al.'s hypothesis on the precise nature of this impact.

\paragraph{Influence of Number of Queries on Utility}
We now perform a more direct analysis of the marginal filtering criterion's impact on \rap's utility.
Our previous regression analysis assessed Aydore et al.'s hypothesis regarding the filtering criterion's influence on \rap's present error as a function of workload size.
However, the filtering criterion does not affect workload size directly --- it only affects the total number of queries consistent with the marginals in the workload.
As such, we believe that a more informative assessment would be to analyze the marginal filtering criterion's influence on \rap's present error as a function of the total number of consistent queries that it evaluates.

\begin{figure}
\centering
\begin{tabular}{cc}
\hspace*{-0.65cm}
  \includegraphics[width=.45\linewidth]{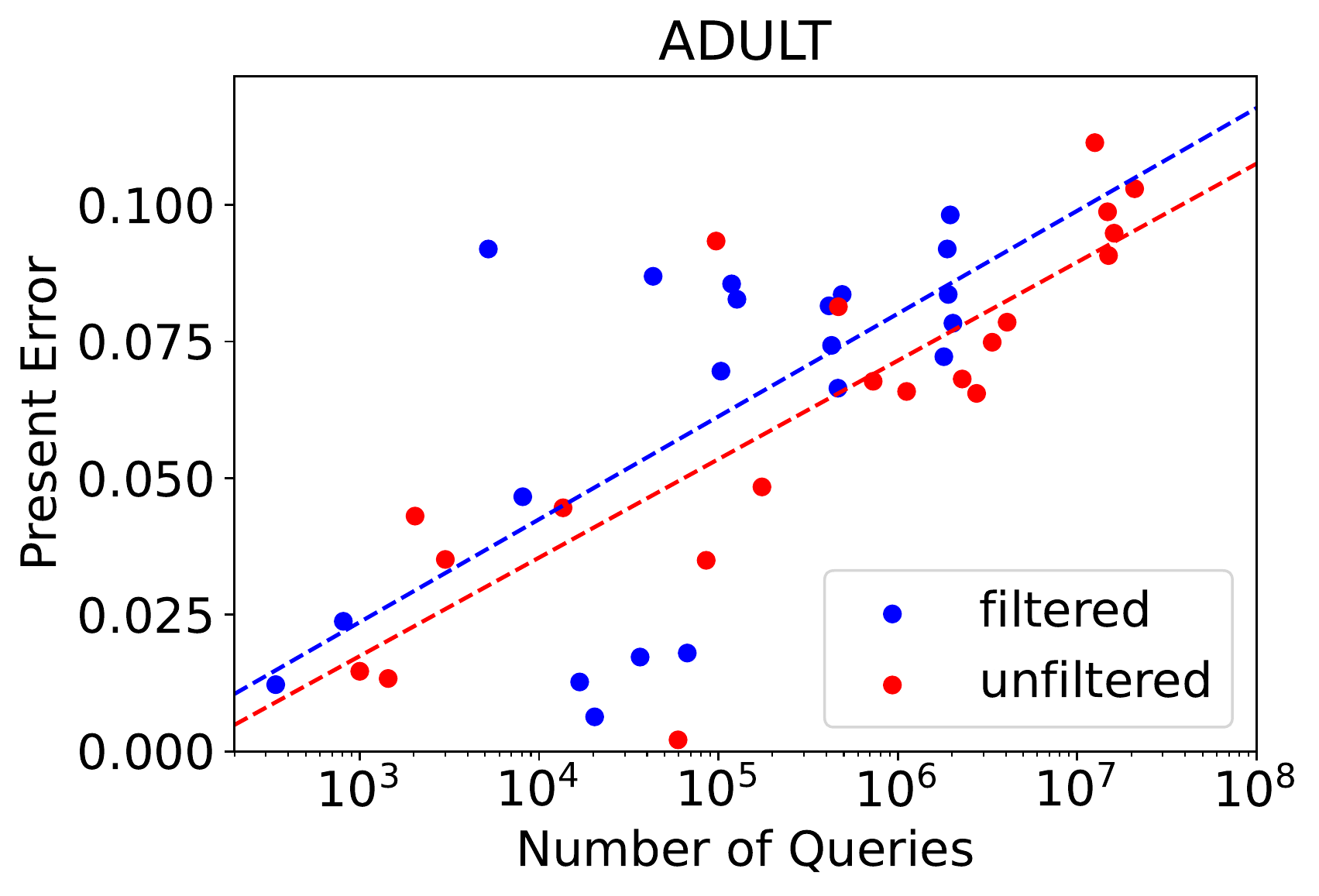} & \includegraphics[width=.45\linewidth]{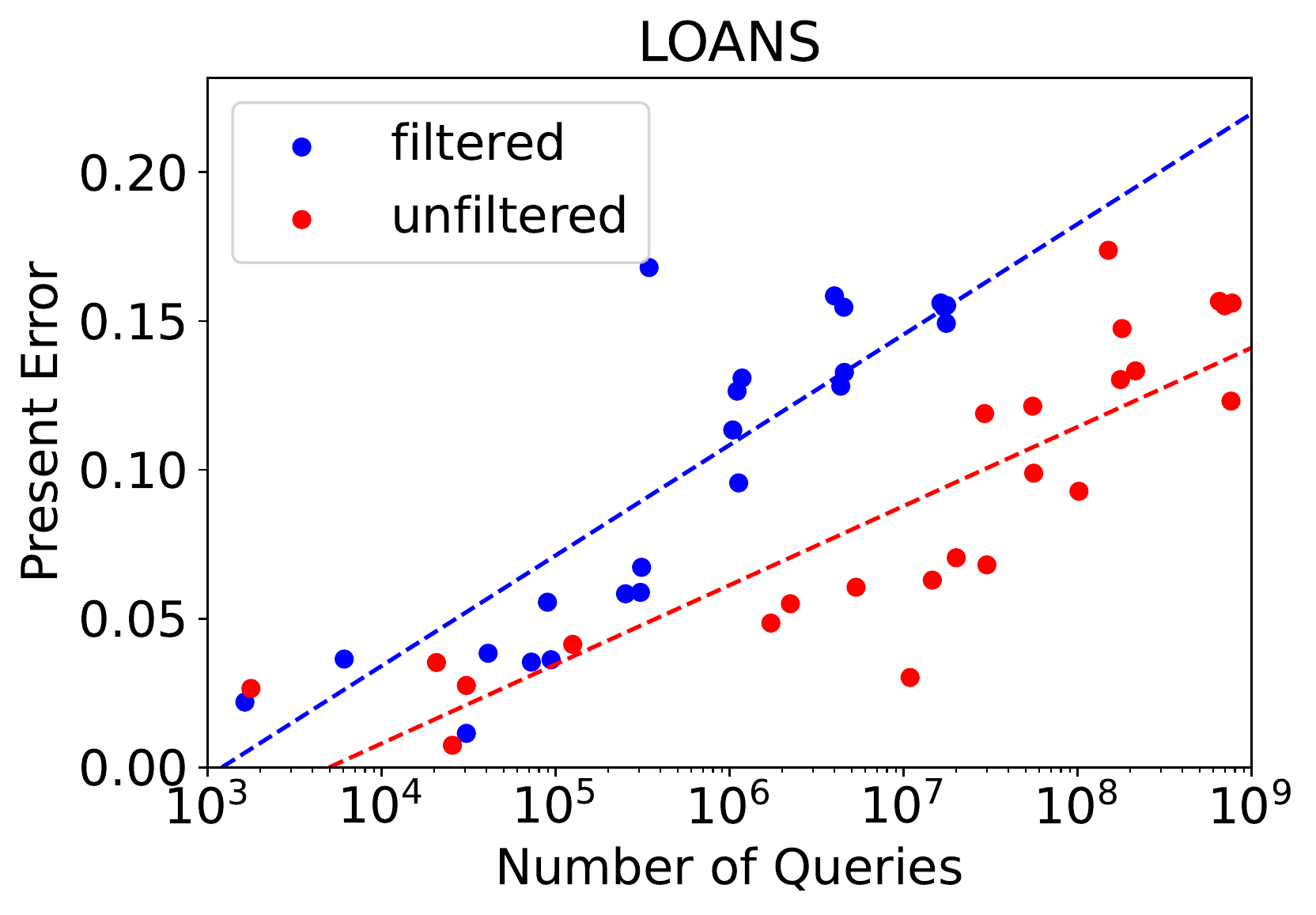}
\end{tabular}
\caption{Regression models for each dataset of \rap's present error vs.\ number of queries for results from filtered and unfiltered marginals, at $\epsilon=0.1$.}
\label{fig:nq-regressions}
\end{figure}

We perform this assessment using precisely the same statistical analysis and regression models as before, only now having the full and restricted models account for the total number of queries rather than workload size.
Figure~\ref{fig:nq-regressions} shows the fitted full regression models on both datasets with $\epsilon$ fixed at $0.1$.
Again, the full regression models for both datasets fit the \rap\ evaluation results well, allowing us to then test these full models against their corresponding restricted models.
The corresponding $p$-values of the likelihood ratio tests for the models on both the ADULT and LOANS \rap\ evaluations were less than $0.0001$, indicating that the filtering criterion has a statistically significant impact on \rap's present error (for both datasets this time).
The results from the figure for both datasets visually imply that including the filtering criterion increases \rap's present error for any given number of queries, and that this increase worsens as the total number of queries grows.
By examining the coefficients (and their corresponding $p$-values) of the full regression models on both datasets, we confirm that this visual trend holds statistically as well.

These results match intuition: in order for a result with the filtering criterion to have approximately the same number of queries as a result without, the result with filtering would likely have corresponded to a larger sized workload.
A larger size workload with the same number of queries implies a more diverse set of queries, whereas a smaller workload with the same number of queries implies a less diverse set of queries with sparser support (i.e., more of the queries evaluate to 0).
Thus, we conclude that Aydore et al.'s initial evaluation of \rap\ --- especially for the highly filtered 5-way marginals --- likely overestimates \rap's present error.
Moreover, this finding motivates a new branch of work on large-scale query answering for the separate cases of when the queries have dense support vs.\ sparse.

\section{Extending \rap's Applicability} \label{sec:extending-applicability}
In this section, we address our third contribution for the setting where queries are prespecified: extending \rap's applicability by expanding the class of queries that it is able to evaluate.
We begin by discussing the motivation behind this contribution.
We then describe what we expand the query class to ($r$-of-$k$ thresholds) and how we accomplish it.
Finally, we detail the empirical evaluations we perform on \rap\ within this expanded query class to quantify its utility and feasibility, finding that \rap\ efficiently evaluates $r$-of-$k$ thresholds with high utility.

\subsection{Motivation}
We contextualize the motivation for this contribution by considering the contributions of prior works.
Prior work on answering statistical queries in practical settings has been focused on relatively simple classes of statistical queries --- most popularly, $k$-way marginals (Definition~\ref{def:kw}), as these are a useful query class which is evaluable within a reasonable computational budget~\cite{barak2007privacy, thaler2012faster, gupta2013privately, chandrasekaran2014faster}.
Aydore et al.'s claim is that their gradient-based \rap\ mechanism~\cite{aydore2021differentially} is able to answer queries from richer classes.
In addition to evaluating $k$-way marginals, they demonstrated this claim by briefly evaluating a new class of queries, 1-of-$k$ thresholds (Definition~\ref{def:1k}).
However, 1-of-$k$ thresholds are essentially a negation of $k$-way marginals.
As such, Aydore et al.\ were able to evaluate \rap\ on 1-of-$k$ thresholds by reusing virtually the same class of EEDQs and the same underlying implementation as they used for $k$-way marginals.
Thus, although their evaluation demonstrated that \rap\ attains high utility on both query classes, these choices of query classes were not fully convincing in demonstrating that \rap\ is effective for answering truly richer classes of queries.
Therefore, it remained an open question whether \rap\ is able to answer richer, more general query classes.

\subsection{Expanding the Query Class} \label{sec:expanding-queries}
To extend \rap's applicability, we develop the mathematical and computational machinery necessary for \rap\ to evaluate a class of queries which generalizes both $k$-way marginals and 1-of-$k$ thresholds: $r$-of-$k$ thresholds (Definition~\ref{def:rk}).
We first describe this query class in detail, then derive its corresponding EEDQs.
Finally, we show how we optimize the derived EEDQs to be more efficiently evaluable, greatly reducing \rap's per-query evaluation time.

\subsubsection{Generalizing to $r$-of-$k$ Thresholds}
Informally, an $r$-of-$k$ threshold query counts what fraction of datapoints in the dataset have at least $r$ out of the $k$ specified attributes.
Thus, it strictly generalizes both $k$-way marginals (when $r=k$) and 1-of-$k$ thresholds (when $r=1$).
$r$-of-$k$ thresholds are a useful generalization because they allow for more expressive, dynamic queries beyond the rigid ``everything'' ($r=k$) or ``anything'' ($r=1$) queries that were previously studied.

The challenge when expanding \rap's evaluation to $r$-of-$k$ thresholds is deriving corresponding EEDQs.
$r$-of-$k$ thresholds cannot trivially reuse the EEDQs relied upon by Aydore et al.\ to evaluate $k$-way marginals and 1-of-$k$ thresholds.
Thus, we must derive new EEDQs for $r$-of-$k$ thresholds, and we accomplish this by generalizing the EEDQs of $k$-way marginals and 1-of-$k$ thresholds.
Towards this, we first reframe the standard definition of $r$-of-$k$ thresholds to enable explicit accounting of all possible combinations of matching and non-matching terms.

\begin{definition}[$r$-of-$k$ thresholds, Alternative] \label{def:rk-alt}
An $r$-of-$k$ threshold query $q_{\phi_{S,y,r}}$ is a statistical query whose predicate is specified by a positive integer $r \le k$, a set $S$ of $k$ features $f_1 \neq \dots \neq f_k \in [d]$, and a target $y \in (\cX_{f_1} \times \dots \times \cX_{f_k})$.
Let $\cR$ denote the set of all partitions $(R_+, R_-)$ of the $k$ features in $S$, such that each $|R_+| \ge r$ and each corresponding $R_- = S - R_+$.
The predicate $\phi_{S,y,r}$ is then given by
\[ \phi_{S,y,r}(x)=
   \begin{cases} 
	  1 & \text{if }\ \bigvee_{(R_+, R_-) \in \cR} \left(\bigwedge_{i \in R_+} (x_{f_i} = y_i) \bigwedge_{i \in R_-} (x_{f_i} \neq y_i)\right)\\
      0 & \text{otherwise.} 
   \end{cases}
\]
Note that at most one partition in $\cR$ will satisfy the predicate.
\end{definition}

We now use this equivalent definition of $r$-of-$k$ thresholds queries to design corresponding EEDQs.
For $k$-way marginals, Aydore et al.\ used \textit{product queries} (Definition~\ref{def:pq}) as EEDQs, which simply compute the product of a datapoint's values at the $k$ specified indices.
For $r$-of-$k$ threshold queries, we generalize product queries in the following ways.
First, we expand the product queries to explicitly include both positive and negated terms, which we refer to as \textit{generalized product queries}.
\begin{definition}[Generalized Product Query] \label{def:gpq}
Given two disjoint subsets of features $T_+, T_- \subseteq [d']$, the generalized product query $\hat{q}_{\hat{\phi}_{T_+,T_-}}$ is a surrogate query parameterized by $\hat{\phi}_{T_+,T_-}$ which is defined as
\[ \hat{\phi}_{T_+,T_-}(x) = \prod_{i \in T_+} x_i \prod_{i \in T_-} (1-x_i). \]
\end{definition}

\noindent Informally, a generalized product query effectively serves as a ``sub''-EEDQ for the conjunction portion of a single partition of $\phi_{S,y,r}(x)$ in Definition~\ref{def:rk-alt}.

Then, leveraging this alternative definition of $r$-of-$k$ thresholds together with generalized product queries, we define a new class of EEDQs in Definition~\ref{def:ptq}: \textit{polynomial threshold queries}.
\begin{definition}[Polynomial Threshold Query] \label{def:ptq}
Given a subset of features $T \subseteq [d']$ and integer $r$, let $\Upsilon$ denote the set of all partitions $(T_+, T_-)$ of $T$ such that each $|T_+| \ge r$ and each corresponding $T_- = T - T_+$.
The \textit{polynomial threshold query} $\hat{q}_{\hat{\phi}_{T,r}}$ is a surrogate query parameterized by $\hat{\phi}_{T,r}$ which is defined in terms of the generalized product query predicates as
\[ \hat{\phi}_{T,r}(x) = \sum_{(T_+, T_-) \in \Upsilon} \hat{\phi}_{T_+,T_-}(x). \]
\end{definition}

\noindent Informally, a polynomial threshold query computes the sum of generalized product queries across all $\sum_{t=r}^k {k \choose t}$ partitions of $T$, where $T$ is constructed identically as in Lemma~\ref{lem:pq}; i.e., for every $i \in S$, we include in $T$ the coordinate corresponding to $y_i \in \cX_{f_i}$.

\subsubsection{Optimizing the Evaluation of Polynomial Threshold Queries} \label{sec:extending-optimizing}
Evaluating polynomial threshold queries can be computationally expensive due to their combinatorial expansion and summation of generalized product query predicates.
Therefore, optimizing their definition to be efficiently evaluable is of utmost importance for enabling \rap\ to evaluate large sets of $r$-of-$k$ thresholds.
Towards this, we present two optimizations that can be used together, which significantly improve the practical runtime of \rap.

The first optimization is inspired by Aydore et al.'s implicit reduction of 1-of-$k$ threshold queries to $k$-way marginal queries.
They accomplished this by recognizing that a 1-of-$k$ threshold predicate is the negation of a $k$-way marginal predicate on a negated datapoint; i.e., $\phi_{S,y,1}(x) = 1-\phi_{S,y,k}(1-x)$.
This equivalence enabled them to efficiently reuse the $k$-way marginals' EEDQs (product queries) in \rap's evaluation.
Applying this concept more generally to computing an $r$-of-$k$ threshold predicate $\phi_{S,y,r}(x)$, the idea is that when $r \le k/2$, it is logically equivalent to compute the negation of a corresponding predicate (with $r' = k-r+1$) on the negated datapoint; i.e., $\phi_{S,y,r}(x) = 1-\phi_{S,y,r'}(1-x)$.
The benefit of utilizing this equivalence when using a polynomial threshold query as the EEDQ to evaluate $\phi_{S,y,r}(x)$ is that \textit{at most} $\lceil k/2 \rceil$ different partition sizes now need to be computed over, compared to at most $k$ when not utilizing this equivalence.
The computational savings from utilizing the equivalence are especially apparent when $r$ is small, as it leads to an exponential (in $k$) reduction in the required number of predicate evaluations.

For the second optimization, the goal is to eliminate the need to explicitly account for the negated terms in our alternative definition of $r$-of-$k$ thresholds (Definition~\ref{def:rk-alt}), as this in turn necessitates the computation of the product of negated values in generalized product queries (Definition~\ref{def:gpq}).
Removing the conjunction over negated terms from Definition~\ref{def:rk-alt} yields a logically equivalent predicate; i.e.,
\[ \phi_{S,y,r}(x)=
   \begin{cases} 
	  1 & \text{if }\ \bigvee_{(R_+, R_-) \in \cR} \bigwedge_{i \in R_+} (x_{f_i} = y_i)\\
      0 & \text{otherwise.} 
   \end{cases}
\]
However, more than one partition of $\cR$ may now satisfy the predicate.
As a result, analogously eliminating the product of negated values from the generalized product query definition (reducing it to a standard product query) would cause the summation in the polynomial threshold query's definition (Def~\ref{def:ptq}) to overcount.
To eliminate computing the product of negated values while simultaneously remedying this overcount, we utilize the principle of inclusion-exclusion to equivalently redefine polynomial threshold queries purely in terms of standard product queries (Definition~\ref{def:pq}).
\begin{definition}[Polynomial Threshold Query, Inclusion-Exclusion] \label{def:ptq-eff}
Given a subset of features $T \subseteq [d']$ and integer $r$, let $\Upsilon(i)$ denote the set of all $i$-size combinations of features in $T$ for $i=r \dots k$; i.e., each $T_i \in \Upsilon(i)$ is such that $|T_i|=i$ and $T_i \subseteq T$.
The polynomial threshold query $\hat{q}_{\hat{\phi}_{T,r}}$ parameterized by $\hat{\phi}_{T,r}$ can be defined in terms of product query predicates $\hat{\phi}_{T_\cdot}$ as
\[ \hat{\phi}_{T,r}(x) = \sum_{i=r}^k (-1)^{i-r} \binom{i-1}{i-r} \sum_{T_i \in \Upsilon(i)} \hat{\phi}_{T_i}(x). \]
\end{definition}
\noindent Utilizing this redefinition of polynomial threshold queries reduces the number of arithmetic operations by nearly half relative to the original definition (when $r > k/2$, which we assume without loss of generality by simultaneously utilizing the first optimization in this section).
In our subsequent experiments with $r$-of-$4$ thresholds (Section~\ref{sec:evaluating-r-of-k}), this reduction in operations results in a maximal runtime improvement of approximately 40\% for evaluating the polynomial threshold queries.

\subsection{Evaluating \rap\ on $r$-of-$k$ Thresholds} \label{sec:evaluating-r-of-k}
With the class of EEDQs derived, the only question that remains is how well \rap\ is able to utilize the EEDQs to answer prespecified sets of $r$-of-$k$ thresholds.
We investigate this question by evaluating how the various inputs to \rap\ affect its present utility and runtime.

\begin{table}
\centering
\begin{tabular}{ |c|c| } 
 \hline
 Primary Mechanism & \rap \\
 \hline
 Baseline Mechanisms & \bma, \gm \\
 \hline
 Utility Measure & $\err_P$ \\
 \hline
 $D$ & ADULT, LOANS \\
 \hline
 $\epsilon$ & $0.1, 1$ \\
  \hline
 $\delta$ & $1/|D|^2$ \\
 \hline
 $|W|$ & $1, 4, 16, 64, 256$ \\
 \hline
 $n^\prime$ & $500, 1000, 2000$ \\
 \hline
 $T$ & $1, 4, 16, 64$ \\
 \hline
 $K$ & $4, 16, 64, 256$ \\
 \hline
 $r$ & $1,2,3,4$ \\
  \hline
 $k$ & $4$ \\
 \hline
\end{tabular}
\caption{Experimental reference table for evaluating $r$-of-$k$ thresholds with \rap.}
\label{tab:thresholds-experiments}
\end{table}

\subsubsection{Utility on $r$-of-$k$ Thresholds}
To begin, we evaluate the present utility of \rap\ on $r$-of-$k$ thresholds, with $k$ fixed at 4.
As in our prior experiments in Section~\ref{sec:improving-evaluation}, we contextualize \rap's utility by comparing against the utilities of the \bma\ and \gm\ baseline mechanisms.
We then evaluate the utility of each mechanisms across a range of $r$ values, $\epsilon$ values, datasets $D$, workload sizes $|W|$, and synthetic dataset sizes $n'$, and across the same $T,K$ values for \rap\ as before.
Table~\ref{tab:thresholds-experiments} contains a summary of the precise parameter values.

\begin{figure}
\centering
\begin{tabular}{c}
\hspace*{-0.65cm}
  \includegraphics[width=1\linewidth]{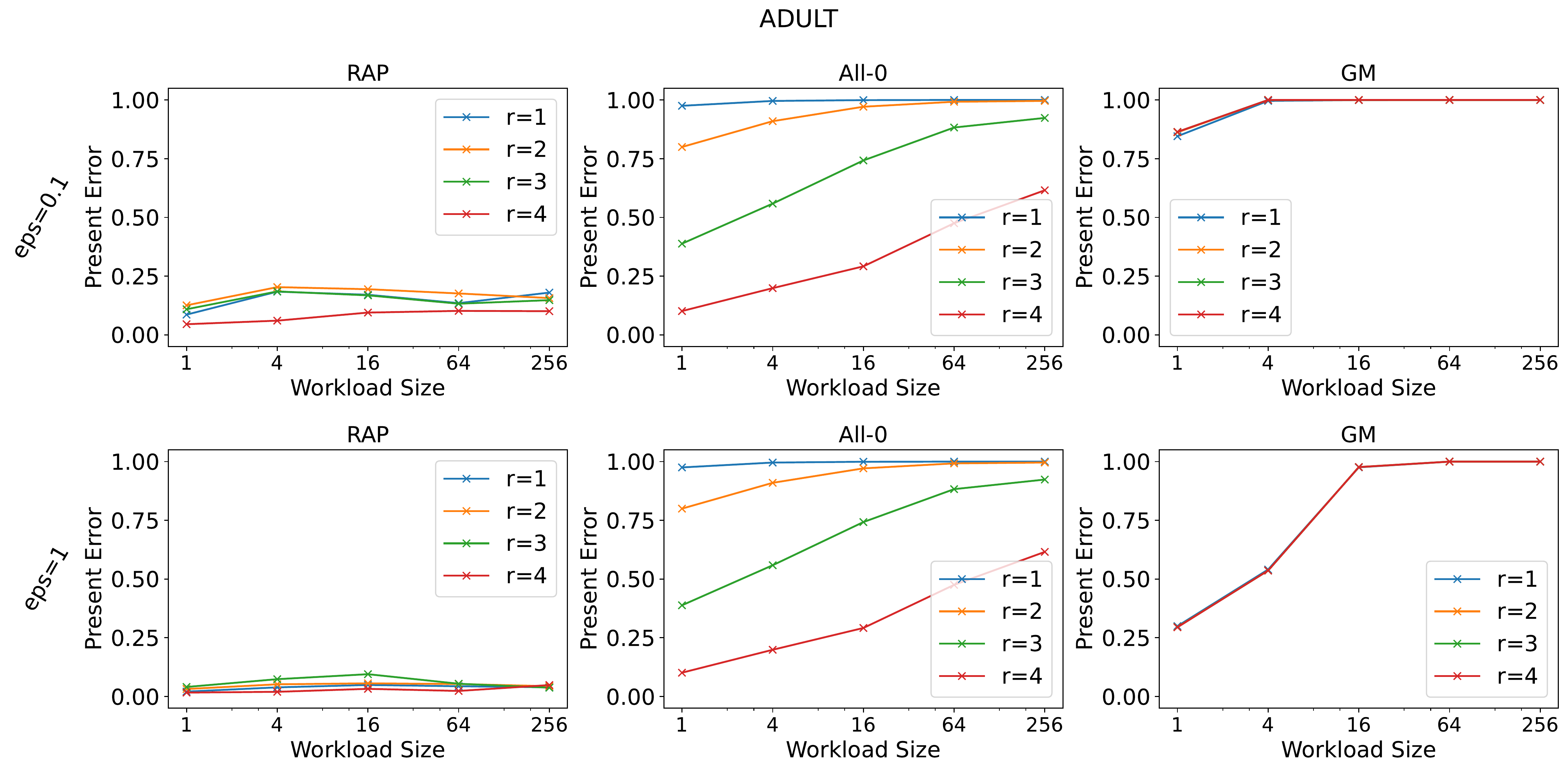} \\
\hline
\hspace*{-0.65cm}
  \includegraphics[width=1\linewidth]{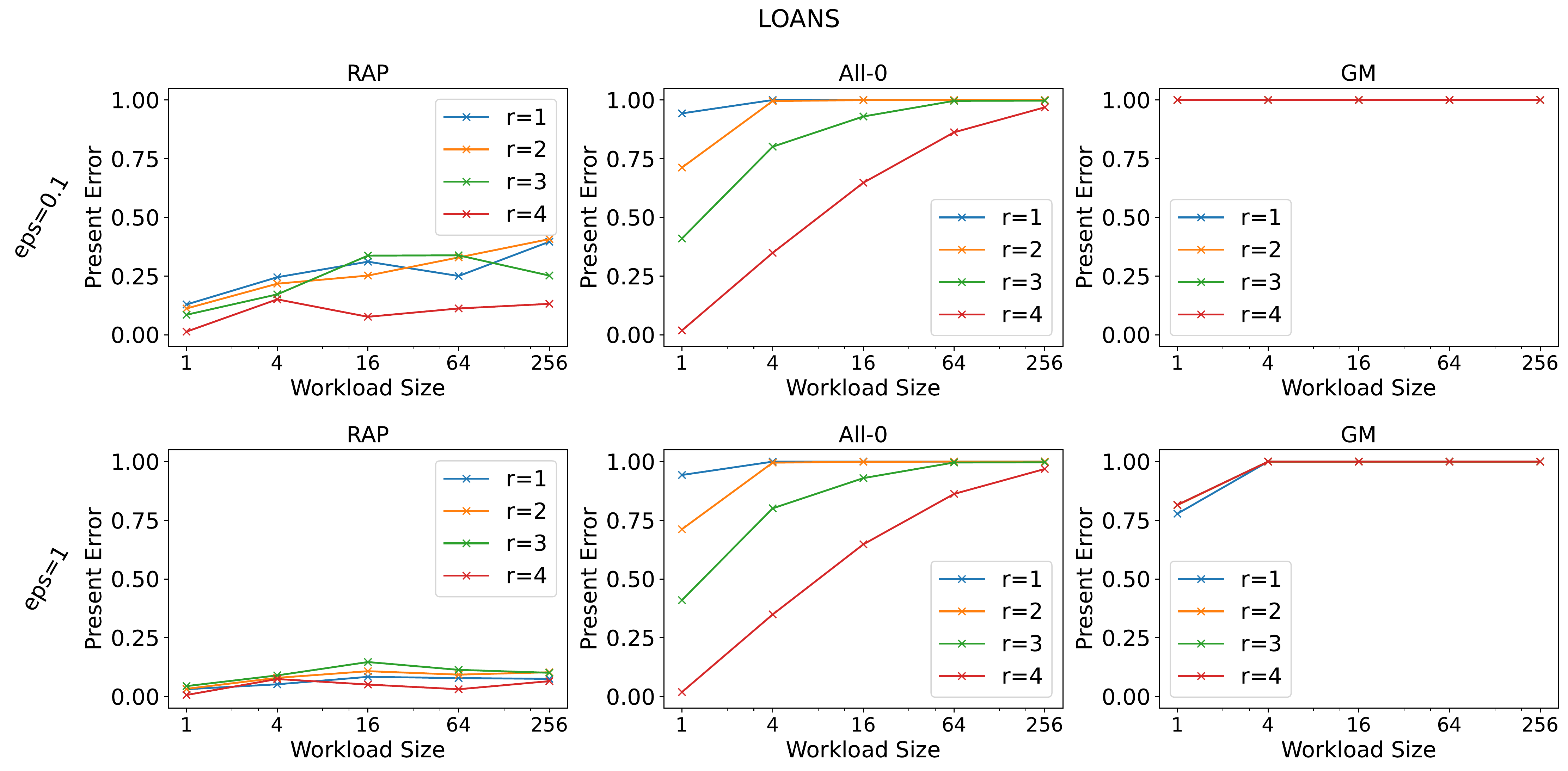}
\end{tabular}
\caption{\rap's minimal present error across all $T, K$ values considered alongside present error of the baseline mechanisms.}
\label{fig:extending-lines}
\end{figure}

Figure~\ref{fig:extending-lines} displays the results of this experiment for $n'=1000$, showing the minimal present error of \rap\ across all $T,K$ values considered alongside the present error of the baseline mechanisms.
The present error of both baseline mechanisms are as expected, with the \bma\ mechanism's present error having a clear and straightforward dependence on $r$, whereas the \gm\ mechanism's present error is independent of $r$.
Immediately, we see that \rap\ significantly outperforms the baseline mechanisms in all settings.
Across the $r$ values, we find that \rap\ achieves its minimal present error at $r=4$ (i.e., 4-way marginals).
Although \rap's present error for $r < 4$ is not much greater than for $r=4$, we find no further obvious relationship between \rap's present error and $r$.

\begin{figure}
\centering
\begin{tabular}{c}
\hspace*{-0.65cm}
  \includegraphics[width=\linewidth]{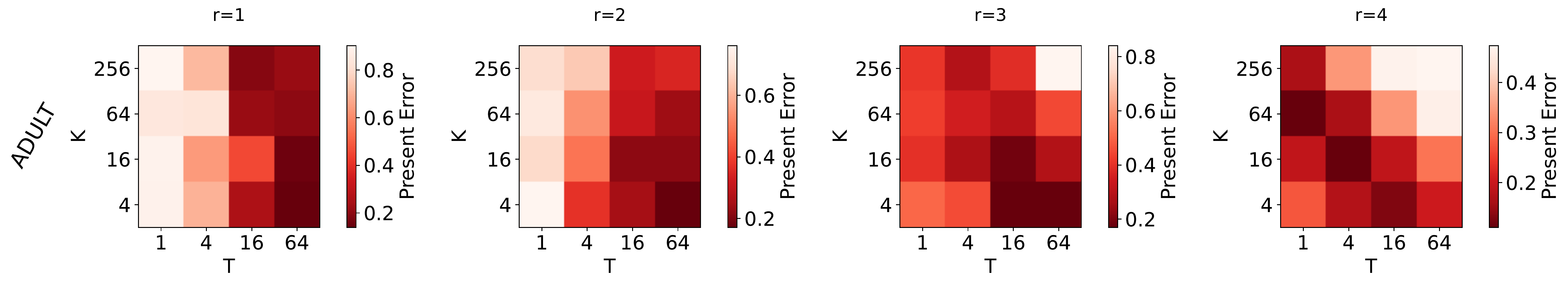} \\
\hline
\hspace*{-0.65cm}
  \includegraphics[width=\linewidth]{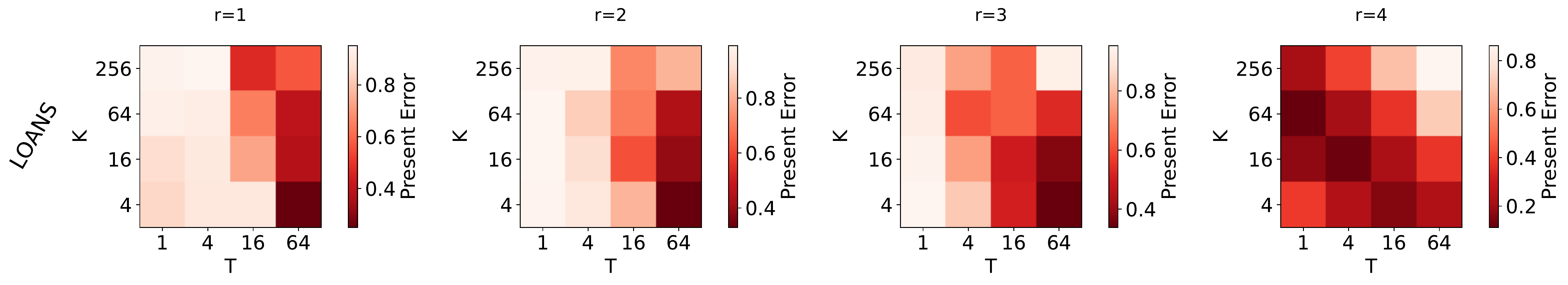}
\end{tabular}
\caption{\rap's present error at each $T, K$ value considered on a workload of 64 $r$-of-$k$ thresholds with $\epsilon = 0.1$.}
\label{fig:extending-adaptivity-r}
\end{figure}

To understand the role of that \rap's adaptivity plays in this experiment, in Figure~\ref{fig:extending-adaptivity-r} we visualize \rap's present error for each individual combination of $T,K$ settings considered.
Just as with 3-way marginals in Section~\ref{sec:reevaluating-role-of-adaptivity}, we find that the same adaptivity behavior emerges with 4-way marginals ($r=4$); i.e., \rap\ primarily needs to evaluate a specific number of queries to achieve low present error, regardless of whether those queries are evaluated jointly in a small number of adaptive rounds or individually across a large number of adaptive rounds.
However, we find that this behavior no longer holds for $r<4$.
Instead, the only consistent pattern that we find for $r<4$ in this figure (which holds across other workload sizes and $\epsilon$ values as well) is that \rap\ achieves its minimal present error when the number of adaptive rounds is relatively large but the number of selected queries per round of adaptivity is relatively small.
Since executing \rap\ for a large number of adaptive rounds is computationally expensive, this finding motivates future work on reducing the necessary number of rounds of adaptivity.
This could potentially be done by more strategically selecting the set of queries in each round --- for instance, by considering their expected joint impact on \rap's present error in the next optimization step, rather than selecting the individual queries with highest present error independently.

\begin{figure}
\centering
\begin{tabular}{c}
\hspace*{-0.65cm}
  \includegraphics[width=.7\linewidth]{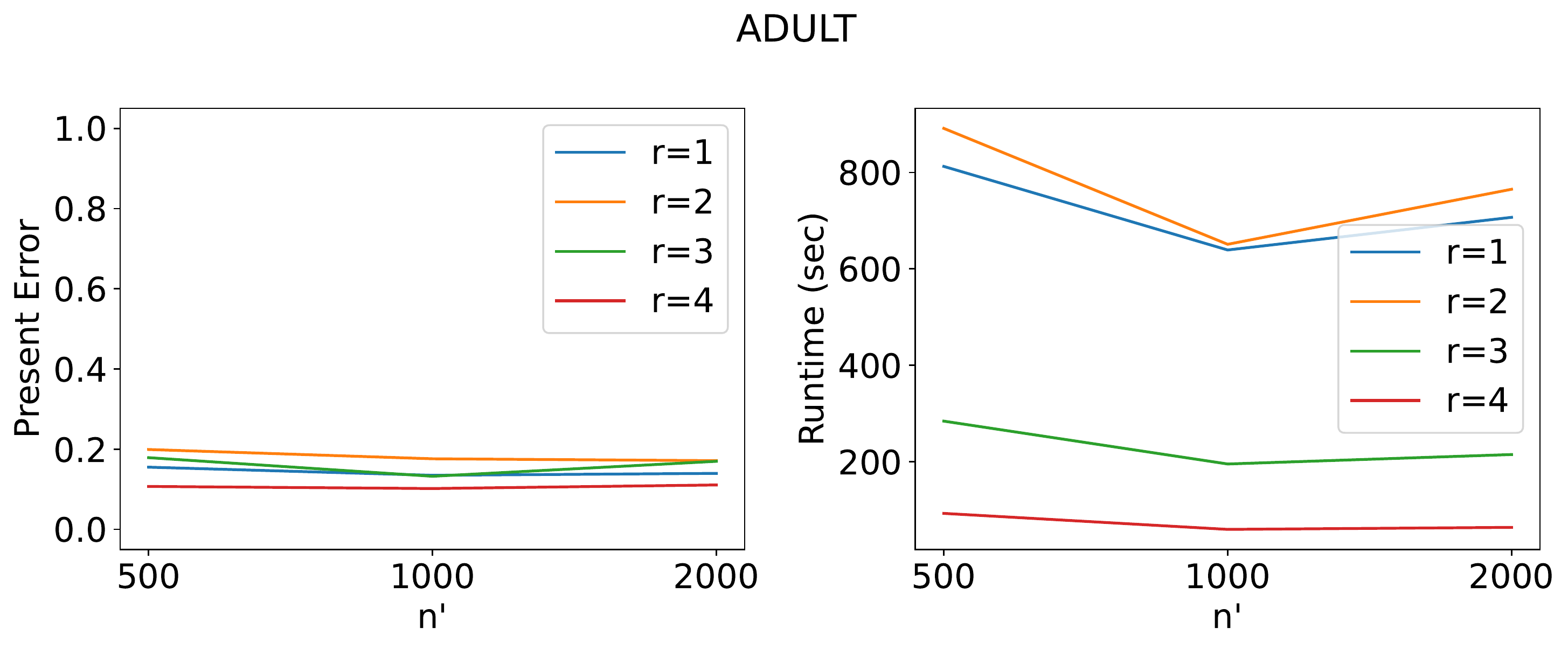} \\
\hline
\hspace*{-0.65cm}
  \includegraphics[width=.7\linewidth]{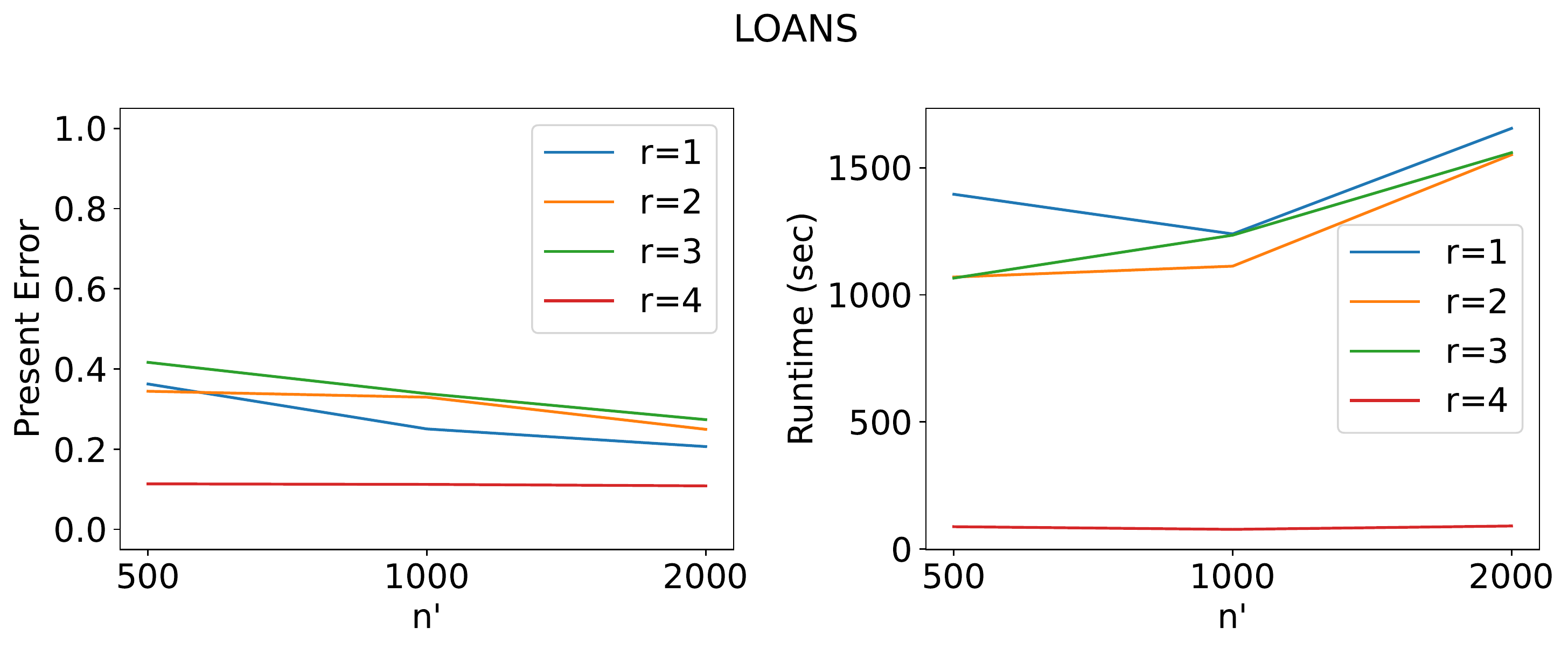}
\end{tabular}
\caption{\rap's present error and runtime as a function of the synthetic dataset size on a workload of 64 $r$-of-$k$ thresholds with $\epsilon = 0.1$.}
\label{fig:extending-synth-data}
\end{figure}

\subsubsection{Effect of Synthetic Dataset Size}
Lastly, we investigate how \rap's synthetic dataset size $n'$ affects its present error and runtime.
Conceptually, $n'$ controls \rap's learning capacity --- the larger $n'$, the better the answers to the queries should be. 
However, since optimizing large synthetic datasets is computationally expensive, $n'$ cannot be taken arbitrarily large.
Similarly, when the synthetic dataset size is too small, the optimization problem becomes underparameterized, which also results in a computationally expensive optimization process.
Aydore et al.\ empirically confirmed this utility--computation trade-off for \rap\ with $k$-way marginals, where they found that setting $n'=1000$ served as a good balance between utility and runtime for (filtered) 3-way and 5-way marginals.

We evaluate this trade-off on (unfiltered) $r$-of-4 thresholds, with the results shown in Figure~\ref{fig:extending-synth-data}.
For each setting of $r$, we find that increasing $n'$ generally results in a mild reduction of \rap's present error, but that at $n'=1000$ \rap\ often attains minimal or near-minimal runtime.
This mirrors Aydore et al.'s results, and thus supports their findings regarding \rap's utility--computation trade-off.
However, one interesting new finding is the effect that $r$ has on \rap's runtime.
Apriori, we expected that \rap\ would have the shortest runtime when evaluating $r$-of-$4$ thresholds with $r\in\{1,4\}$, and that their runtimes would be comparable.
This is because at $r\in\{1,4\}$, \rap\ has the least arithmetic operations to perform in order to evaluate each predicate (compared to $r\in\{2,3\}$, refer to Section~\ref{sec:extending-optimizing} for details on predicate evaluation).
Although we confirm that \rap\ acheives minimal runtime at $r=4$, we find that nearly the opposite holds true for $r=1$, which induces up to a 20x longer runtime.
This increase in runtime is primarily explained by our prior observation that for $r < 4$, \rap\ achieves its maximal utility via a larger number of adaptive rounds (where \rap's runtime appoximately linearly increases with the number of rounds).

However, even with this jump in runtime taken into consideration, we find that \rap\ is a highly performant mechanism for evaluating large sets of queries.
For instance, consider the worst-case runtime at $n'=1000$ in Figure~\ref{fig:extending-synth-data}, which occurs where \rap\ answered a workload of 64 1-of-4 thresholds on the LOANS dataset.
Here, \rap\ answered approximately \num{3.5e7} individual consistent queries in 1,240 seconds --- a rate of over 28,000 queries per second.
Based on these findings, we conclude that \rap\ is highly efficient for answering large sets of $r$-of-$k$ thresholds.

\section{Understanding \rap's Generalizability} \label{sec:understanding-generalizability}
In this final section, we propose a new and realistic intermediate setting that lies between the classic settings of having full knowledge of all queries in advance (i.e., the prespecified queries setting) vs.\ having no knowledge of which queries will be posed.
We begin by concretely defining this new partial knowledge setting along with a generalization-based measure of utility for mechanisms operating within it.
We then address our final contribution by empirically evaluating \rap's utility to determine its suitability in the new setting.

\subsubsection*{Motivation}
In statistical modeling, and especially in the subfield of synthetic data generation, the primary goal is not to generate a model or a synthetic dataset that answers a prespecified set of queries well. 
Rather, the goal is to generate a model or synthetic dataset that \textit{generalizes} well to future queries~\cite{vapnik1999overview, mohri2012foundations}.
When it comes to differentially private mechanisms for answering statistical queries through a synthetic dataset, prior utility analyses have  focused on either: (a) how well those mechanisms answer the prespecified set of queries, or (b)  theoretically bounding how well the mechanisms can answer any class of queries in the worst-case.
For example, the utility of \rap\ (and the related practical mechanisms which preceded it) had previously been based on solely the answers to the prespecified workload; e.g., present utility.
Experimentally evaluating a mechanism's present utility is straightforward: simply report the error of the highest error query from the prespecified query set.
However, in some settings, it may be more useful to understand how well the mechanism can answer future queries.
Towards this, theoretical bounds can provide strong guarantees for the mechanism's worst-case utility across an entire query class~\cite{blum2008learning, dwork2009complexity, dwork2010boosting, hardt2010simple, thaler2012faster}.
The drawback to using these theoretical bounds in practical settings is that they may be overly pessimistic, especially if the queries posed in the future are highly similar to the queries that were used to generate the synthetic dataset.
This apparent disparity between the utility suggested by theoretical analyses and the actual utility that may be observed in practice is nearly identical to the disparity that famously exists between utility analyses in theoretical vs.\ empirical machine learning research~\cite{vapnik1998support, bartlett2006empirical, shalev2014understanding, neyshabur2014search, zhang2021understanding}.
However, for answering statistical queries with DP, the theoretical worst-case bounds are currently the best tool available without introducing additional information or assumptions.

\subsection{Defining the Partial Knowledge Setting}
We now motivate the design of this particular partial knowledge setting, then formally define it.

Much like in the machine learning research literature, we motivate a new partial knowledge setting for the context of differential privacy based on the rationale that in some realistic settings, future queries may be similar to queries posed in the past; i.e., historical queries.
For instance, the U.S.\ Census Bureau periodically collects sensitive data for the decennial census, and routinely allows researchers to securely pose queries directly on the collected data.
Because similar data is being collected each decennial census, it is very likely that some of the queries analysts pose on one census dataset will be similar to the queries that analysts pose on the next census dataset.

We formalize this intuition on partial query repeatability for $r$-of-$k$ thresholds in a general manner in Definition~\ref{def:pksetting}.
For ease of exposition, we first introduce the following notation.
Let $\cT$ be an arbitrary distribution over thresholds, and let $Q \leftarrow \cT$ denote the vector of all consistent queries $Q$ of a threshold randomly drawn from distribution $\cT$.
Similarly, we let $Q \xleftarrow{|W|} \cT$ denote the vector of all consistent queries $Q$ from a $|W|$ size workload of thresholds sampled i.i.d.\ from $\cT$.
\begin{definition}[Partial Knowledge Setting, General] \label{def:pksetting}
Let $\cT_H$ and $\cT_F$ be arbitrarily related distributions over thresholds.
In this setting, DP mechanisms are expected to answer arbitrary future thresholds drawn i.i.d.\ from $\cT_F$.
However, the DP mechanisms are not provided $\cT_F$ explicitly.
Instead, DP mechanisms are provided access to partial knowledge of $\cT_F$ via a workload $W_H$ of ``historical'' thresholds sampled i.i.d.\ from $\cT_H$; i.e., the mechanisms are given access to $Q_H \xleftarrow{|W_H|} \cT_H$.
\end{definition}

Intuitively, in this partial knowledge setting, mechanisms can utilize $Q_H$ to learn about the underlying threshold distribution $\cT_H$, and if $\cT_H$ is similar to $\cT_F$, this will, in turn, inform what areas of the threshold space future thresholds are most likely to be sampled from.
The role of $Q_H$ in this setting is analogous to the role that training data plays in machine learning; i.e., it is the concrete sample of data provided to the mechanism that the mechanism can use to attempt to generalize.

In order for the historical queries $Q_H$ to convey useful information about $\cT_F$ to the DP mechanism, $\cT_H$ and $\cT_F$ should be related.
Towards this, in Definition~\ref{def:pkconc} we specify two concrete instantiations of the partial knowledge setting which make the relationship between $\cT_H$ and $\cT_F$ explicit.
\begin{enumerate}
\item Informally, the first concrete instantiation is the \textit{exact} partial knowledge setting, where historical thresholds are drawn from the same distribution as the future thresholds.
\item The second concrete instantiation is the \textit{drifting} partial knowledge setting, which extends the exact partial knowledge setting.
The drifting partial knowledge setting is inspired by the practical consideration that even if the historical and future thresholds distributions are initially the same, they may gradually drift apart over time.
\end{enumerate}

In both settings, we ground the historical and future thresholds distributions in the observation that in practice, certain features (or combinations of features) are likely to be more relevant to analysts than other features.
For instance, in the ADULT dataset, ``Age'' and ``Years of education'' might be more relevant and useful for analyses than ``Capital loss amount'' and ``Relationship status''.
We model this relevance as a historical probability distribution $\cF_H$ over the \textit{features}, such that the probability mass corresponding to any $r$-of-$k$ threshold in $\cT_H$ corresponds to the (normalized) product of the $k$ features' probabilities; i.e., $\cT_H$ is the sampling distribution of $k$ features from $\cF_H$ without replacement.
Our definition of the drifting partial knowledge setting specifically attempts to capture the practical phenomenon that if (for instance) analysts' interests are concentrated primarily in a small subset of features, then even if their interests drift over time, the analysts' new interests may still be concentrated in a small subset of different features.
Based on this, we now formally define both concrete instantiations of the partial knowledge setting.

\begin{definition}[Partial Knowledge Setting, Exact \& Drifting] \label{def:pkconc}
Let $\cF_H$ be an arbitrary historical distribution over features with $\cT_H$ as its corresponding historical thresholds distribution.
Without loss of generality, assume the features are sorted in descending order of their probability masses under $\cF_H$; i.e., for each feature $f_i$ with probability $p_i$, we have that $p_i \ge p_{i+1}$.
Let $\gamma \in [0,1]$ be a drift parameter, which defines the distributional similarity of the future distribution over features $\cF_F$ (and correspondingly the future thresholds distribution $\cT_F$) as follows.
For each probability $p_i$, associate the corresponding key 
\[
k_i = 
\underbrace{(1-2\gamma)}_{\substack{\text{ordering} \\ \text{weighting}}} \ \cdot
\underbrace{\frac{d-i}{d-1}}_{\substack{\text{relative order,} \\ \text{normalized}}} + \ 
\underbrace{(1-|1-2\gamma|)}_{\substack{\text{shuffling} \\ \text{weighting}}} \ \cdot
\underbrace{u_i}_{\substack{\text{random} \\ \text{shuffling} \\ \text{amount}}},
\]
where $u_i \overset{\text{iid}}{\sim} \textrm{Uniform}[0,1]$.
The feature distribution $\cF_F$ is defined by leaving the features fixed in their original ordering, but reordering the probability masses in descending order of their keys.
This results in a distribution of the same concentration, but with probability masses re-assigned to potentially different features.
The future thresholds distribution $\cT_F$ is therefore the sampling distribution of $k$ features without replacement from $\cF_F$.
When $\gamma = 0$, this procedure yields $\cT_F = \cT_H$, and we refer to this as the \textit{exact} partial knowledge setting.
When $\gamma > 0$, we refer to this as the \textit{drifting} partial knowledge setting.
\end{definition}

This model of drift is designed to maintain the concentration of the initial feature distribution $\cF_H$ while interpolating between the exact partial knowledge setting ($\gamma=0$) and a uniformly random reshuffling of the features' probabilities ($\gamma=1/2$).
For $0 < \gamma < 1/2$, this model induces a weighted amount of random reshuffling of probabilities in conjunction with simultaneously encouraging features' probabilities to remain ``similar'' to what they initially were; e.g., features with large probability masses under $\cF_H$ are likely to retain large probability masses under $\cF_F$.
On the other end of the spectrum is the $\gamma > 1/2$ setting, where the relative orderings of probabilities become more likely to be reversed; e.g., features with large probability masses under $\cF_H$ are likely to be assigned small probability masses under $\cF_F$.
At the extreme of this setting is $\gamma=1$, which induces $\cF_F$ of maximal total variation distance to $\cF_H$ by deterministically reversing the relative ordering of the features' probabilities.
Figures~\ref{fig:generalizing-drift-hists} and \ref{fig:generalizing-drift-tvs} concretely illustrate how the drift amount $\gamma$ affects the distribution of future features.

\begin{figure}
\centering
\begin{tabular}{c}
\hspace*{-0.65cm}
  \includegraphics[width=1\linewidth]{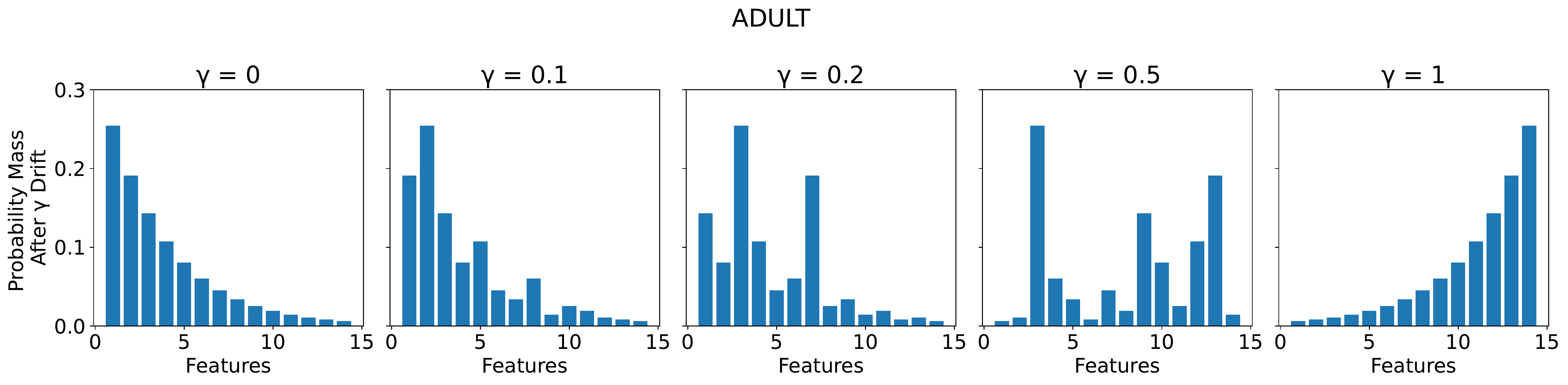} \\
\hline
\hspace*{-0.65cm}
  \includegraphics[width=1\linewidth]{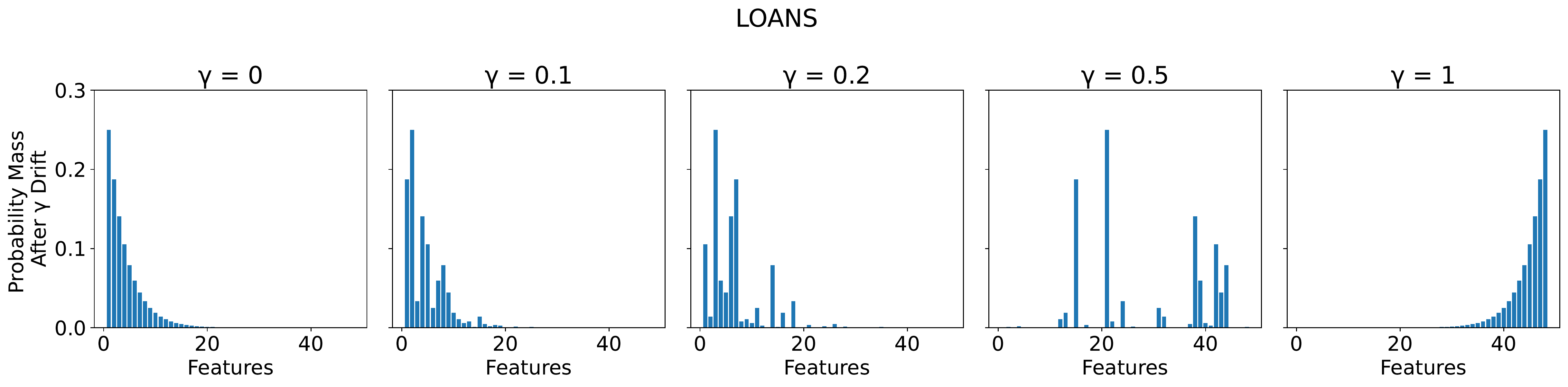}
\end{tabular}
\caption{Examples of drifted feature distributions $\cF_F$ across a range of drift parameters $\gamma$, with an initial Geometric distribution for $\cF_H$ on the ADULT and LOANS datasets. Categorical features are numbered (rather than named) along the $x$-axis.}
\label{fig:generalizing-drift-hists}
\end{figure}

\begin{figure}
\centering
\includegraphics[width=0.5\linewidth]{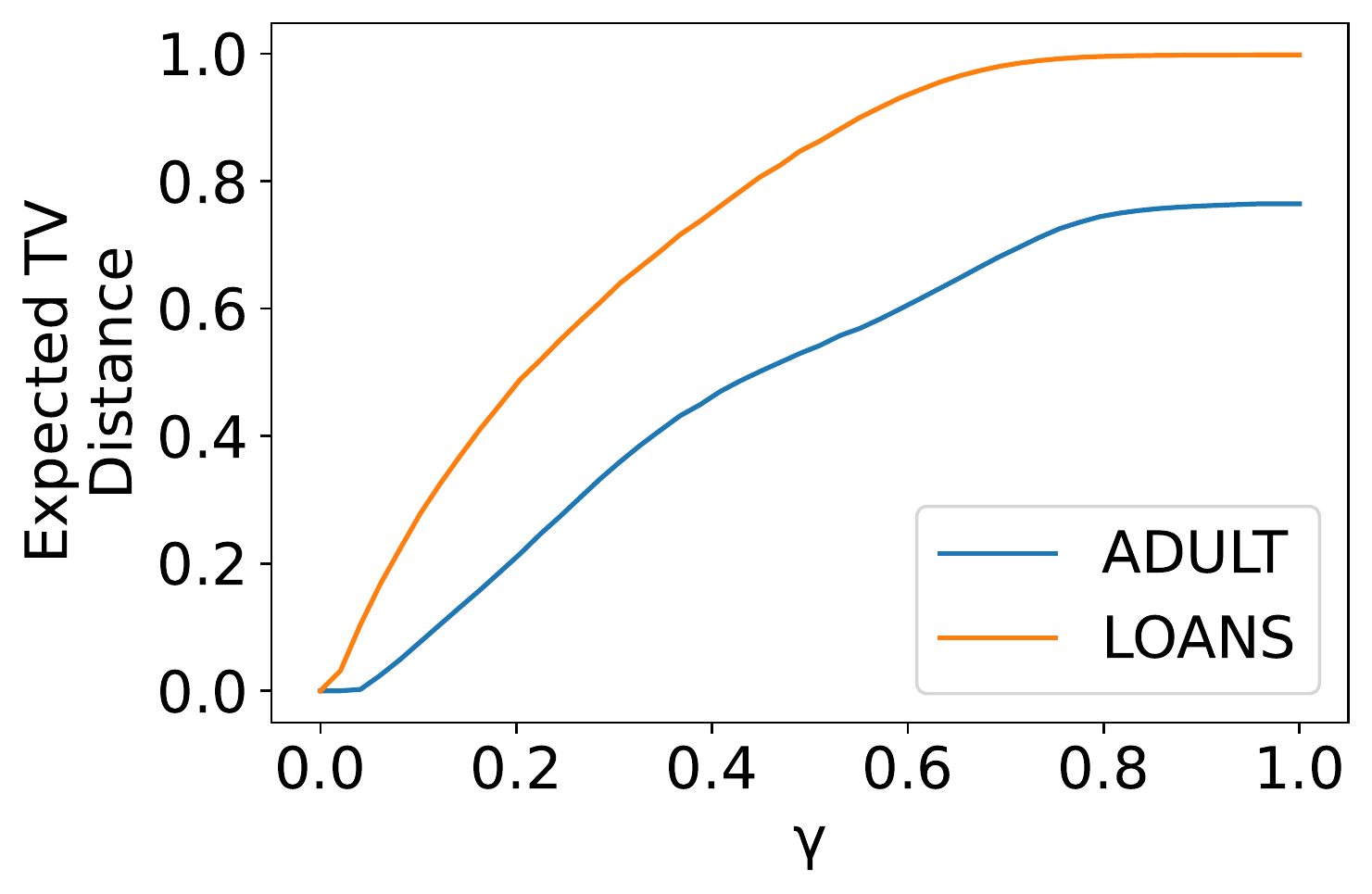}
\caption{Effect of drift parameter $\gamma$ on the total variation distance between the historical features distribution $\cF_H$ and the future features distribution $\cF_F$, with an initial Geometric distribution for $\cF_H$ on the ADULT and LOANS datasets.}
\label{fig:generalizing-drift-tvs}
\end{figure}

\subsection{Measuring and Computing Utility}
Having concretely defined the partial knowledge setting, we formally define a utility measure to quantify how well a mechanism can answer future thresholds based on the historical thresholds it was given access to; i.e., a measure quantifying how well the mechanism generalizes.
We then describe how to empirically evaluate this defined utility measure in an efficient way.

In this setting, we are interested in the mechanism's error across its answers to the consistent queries of $r$-of-$k$ thresholds drawn from $\cT_F$.
This new utility measure is based on the classic utility measure used in the prespecified queries setting (Definition \ref{def:present-err}), with the only difference being that the randomness of the future thresholds distribution $\cT_F$ is now explicitly taken into account.
We thus define \textit{future utility} which we measure in terms of the negative of \textit{future error}; i.e., a mechanism with low future error has high future utility, and vice versa.
Specifically, future error is the expected absolute error taken over the randomness of both $M$ and $\cT_F$, formally defined as follows.
\begin{definition}[Future error]
Let $a = Q(D)$ be the true answers to all queries in $Q$ on $D$, and let $\tilde{a}$ be mechanism $M$'s corresponding answers. Then $\err_F$ is the future error of mechanism $M$, defined as $\err_F(M,D,\cT_F) = \E_{M(D), Q\leftarrow \cT_F} \Vert a - \tilde{a} \Vert_\infty$, where the expectation is over the randomness of both the mechanism and future threshold distribution.
\end{definition}

Theoretically evaluating $\err_F$ of a mechanism on \textit{a priori} unknown threshold distributions without resorting to worst-case bounds is a challenging problem.
Experimentally, however, we are able to efficiently and accurately estimate $\err_F$ for the \rap\ mechanism as follows:
\begin{enumerate}[leftmargin=40pt]
\item Construct feature distributions $\cF_H$ and $\cF_F$ according to real-world phenomena, which in turn define threshold distributions $\cT_H$ and $\cT_F$.
\item Generate a workload $W_H$ of historical thresholds, yielding query vector $Q_H \xleftarrow{W_H} \cT_H$. Independently, generate a workload $W_F$ of future thresholds, yielding query vector $Q_F \xleftarrow{W_F} \cT_F$.
\item Provide $Q_H$ as the input queries to \rap\ in order to generate a synthetic dataset.
\item Use the synthetic dataset to answer $Q_F$, recording the mean error (and optionally, the corresponding confidence intervals to quantify how faithfully $\err_F$ was approximated).
\end{enumerate}
This evaluation approach is analogous to standard practice in empirical machine learning research where data is split into ``training'' and ``test'' sets randomly (to ensure distributional similarity)~\cite{hastie2009elements}.
The model is then learned on the training set, and subsequently evaluated on the test set to measure how well it generalizes.

\subsection{Evaluating \rap's Future Utility}
As our final contribution, we empirically evaluate \rap's future utility for answering $r$-of-$k$ thresholds.
The experiments that we perform on \rap\ to understand its suitability in this new partial knowledge setting are as follows:
\begin{itemize}[leftmargin=30pt]
\item Evaluating the effects that the threshold distribution concentration and the historical threshold workload size $|W_H|$ have on \rap's future utility.
\item Evaluating the effect that ``overfitting'' in the synthetic data optimization step has on \rap's future utility.
\item Evaluating the effect that the distribution drift amount $\gamma$ has on \rap's future utility.
\end{itemize}
These experiments are designed to assess the distinct new ways (beyond those in the previous prespecified queries setting) in which \rap's inputs may influence its future utility.

\subsubsection{Effect of Threshold Distribution Concentration \& Historical Workload Size}
To empirically evaluate \rap's future utility in the exact partial knowledge setting, we must specify the particular threshold distribution $\cT_H = \cT_F$ from which we generate both the input queries $Q_H$ and future queries $Q_F$ used to evaluate $\err_F$.
As previously discussed, we do so by specifying feature distributions $\cF_H$ and $\cF_F$ that, in turn, define the threshold distributions.
As a baseline, we choose what is intuitively the most challenging extreme: setting $\cF_H$ and $\cF_F$ to be the Uniform distribution.
We expect the future utility of this baseline to be the lowest among all possible distributions since it is the least concentrated, implying that it provides the least amount of information possible to the mechanism about any particular region of the threshold space.

In an effort to model the real-world phenomena that certain features are likely to be more relevant to analysts than other features, we utilize the following two feature distributions.
For a highly concentrated distribution, we use the exponentially-tailed Geometric distribution.
For a mildly concentrated distribution, we use the heavy-tailed Zipfian distribution.
Both distributions are commonly used in practice when modeling real-world phenomena; e.g., ~\cite{miller1989modeling, yu2004false, zeng2012topics, okada2018modeling}.
We hypothesize that the highly concentrated Geometric distribution will induce high-utility results, since many of the same features in $Q_H$ will also appear in $Q_F$.
Analogously, we hypothesize that the mildly concentrated Zipfian distribution will induce lower-utility results (although still higher than the Uniform distribution baseline).

With a fixed threshold distribution $\cT_H$ defined by the feature distribution $\cF_H$, we must specify how many thresholds will be randomly sampled to form the historical threshold workload $W_H$ (and corresponding vector of all consistent queries $Q_H$) that \rap\ takes as input.
Obtaining a clear understanding what impact the historical workload size $|W_H|$ has on \rap's future utility is important because there may be a tension between the number of historical $r$-of-$k$ thresholds and \rap's future utility.
On the one hand, the more sampled thresholds there are, the more information \rap\ has about the underlying distribution $\cT_F$ from which future thresholds will be generated.
This suggests that the more historical $r$-of-$k$ thresholds there are, the higher \rap's future utility should be.
On the other hand, to optimize \rap's underlying synthetic dataset, its privacy budget is split between all queries consistent with the historical thresholds.
This implies that the more historical thresholds there are, the more noise will be added to each consistent query's answer, which seems to suggest that this will cause the future utility to be lower.
Thus, we seek to understand whether one of these two possibilities is correct, or whether there is a ``sweet spot'' where a certain number of historical thresholds is just enough for the mechanism to implicitly learn $\cT_F$ but does not result in the privacy budget being spread too thin.

\begin{table}
\centering
\begin{tabular}{ |c|c| } 
 \hline
 Primary Mechanism & \rap \\
  \hline
 Baseline Mechanism & \bma \\
 \hline
 Utility Measure & $\err_F$ \\
 \hline
  $D$ & ADULT, LOANS \\
 \hline
 $\epsilon$ & $0.1$ \\
  \hline
 $\delta$ & $1/|D|^2$ \\
 \hline
 $|W_H|$ & $1, 4, 16, 64, 256, 1024$ \\
 \hline
 $n^\prime$ & $1000$ \\
  \hline
 $T$ & $1, 4, 16, 64$ \\
 \hline
 $K$ & $4, 16, 64, 256$ \\ 
 \hline
 $r$ & $1$ \\
  \hline
 $k$ & $3$ \\
 \hline
 $\cT_H, \cT_F$ & Uniform, Zipf, Geometric \\
 \hline
  $\gamma$ & $0, 0.05, 0.1, 0.2, 0.5, 1$ \\
 \hline
\end{tabular}
\caption{Experimental reference table for evaluating the future utility of \rap\ on $r$-of-$k$ thresholds.}
\label{tab:generalizability-experiments}
\end{table}

To empirically quantify the effect of both the threshold distribution concentration as well as historical workload size on \rap's future utility, we evaluate \rap\ across a range of workload sizes using the three specified distributions over features in both the ADULT and LOANS datasets.
To put \rap's future utility into context, we also evaluate the future utility of the \bma\ baseline mechanism.
Refer to Table~\ref{tab:generalizability-experiments} for a summary of this experiment.

\begin{figure}
\centering
\begin{tabular}{c}
\hspace*{-0.65cm}
  \includegraphics[width=\linewidth]{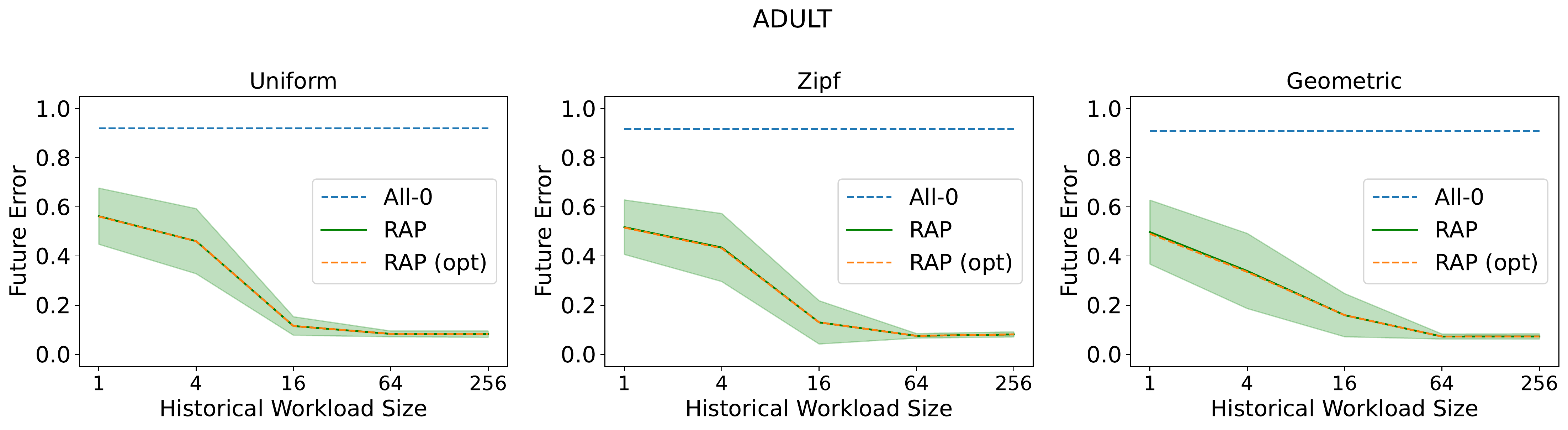} \\
\hline
\hspace*{-0.65cm}
  \includegraphics[width=\linewidth]{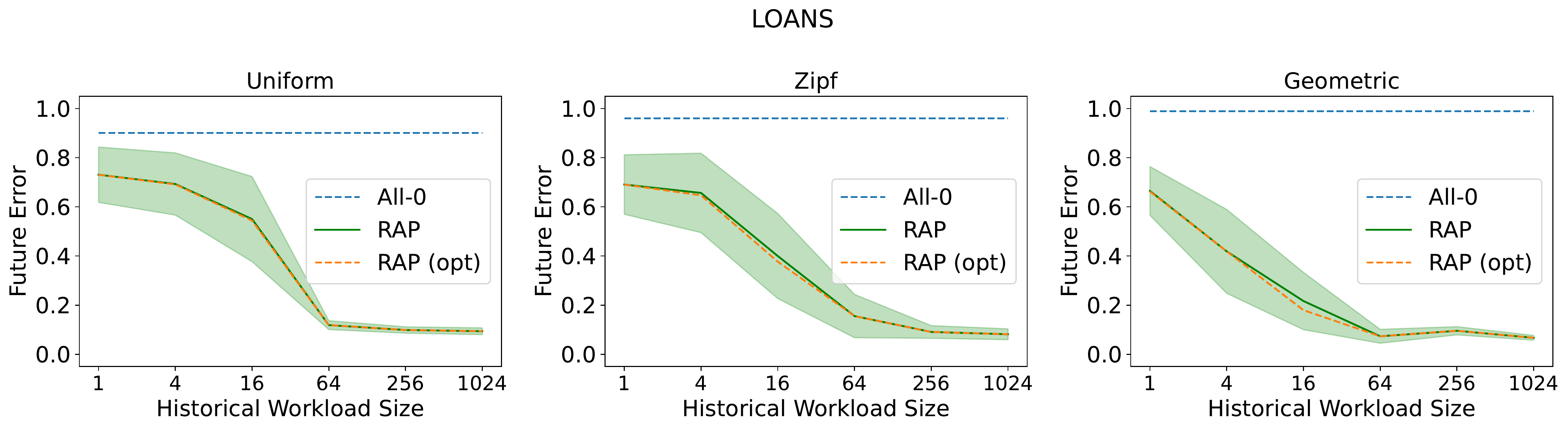}
\end{tabular}
\caption{\rap's future error (and 95\% confidence intervals) across all $T,K$ values considered where \rap\ achieves minimal present error, plotted across a range of workload sizes and historical threshold distributions. ``\rap\ (opt)'' represents \rap's future error across all $T,K$ values considered where \rap\ achieves minimal future error. Future error of \bma\ included as a baseline.}
\label{fig:generalizing-dist-lines}
\end{figure}

\begin{figure}
\centering
\begin{tabular}{cc}
\hspace*{-0.65cm}
  \includegraphics[width=.45\linewidth]{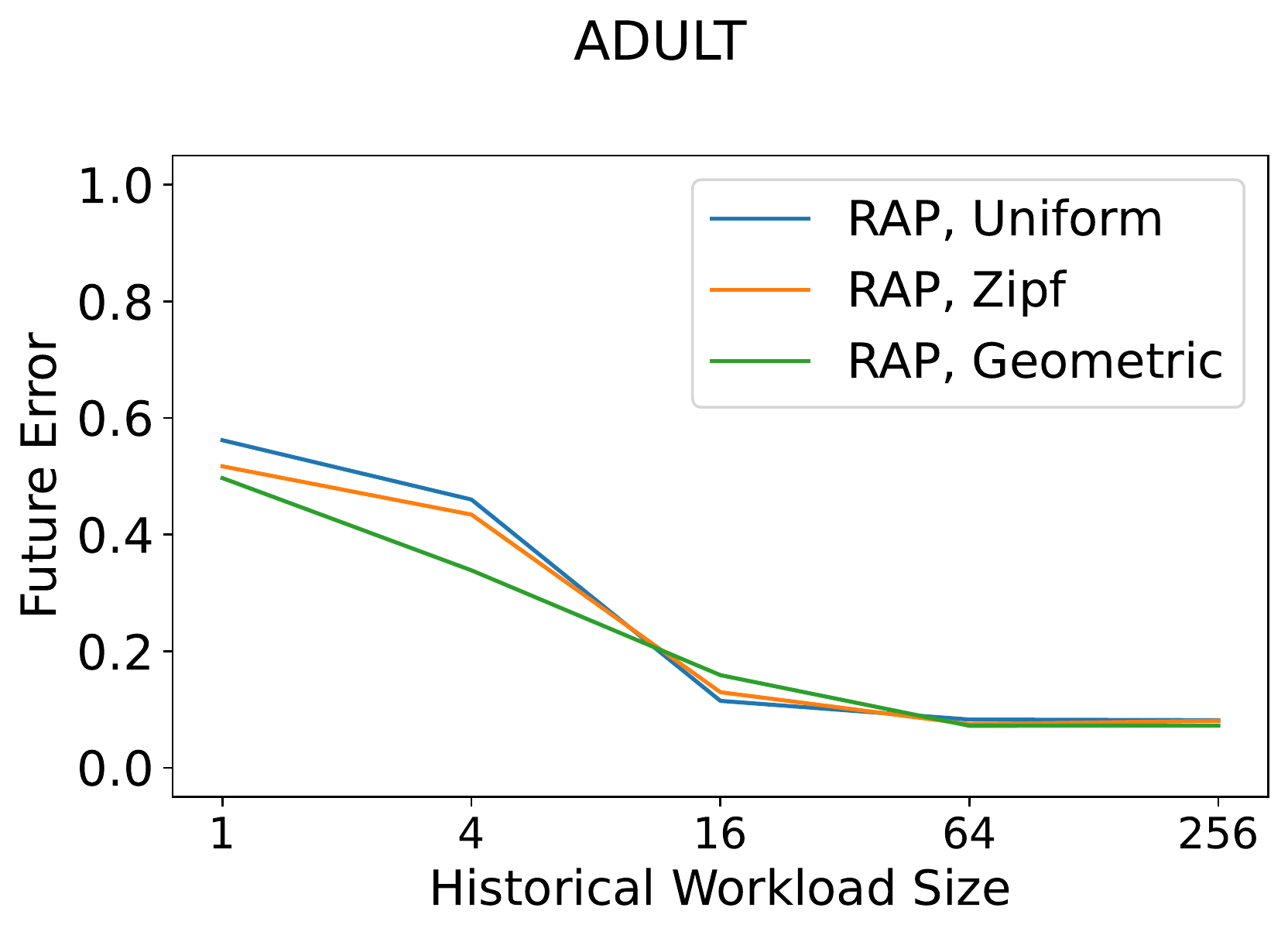} & \includegraphics[width=.45\linewidth]{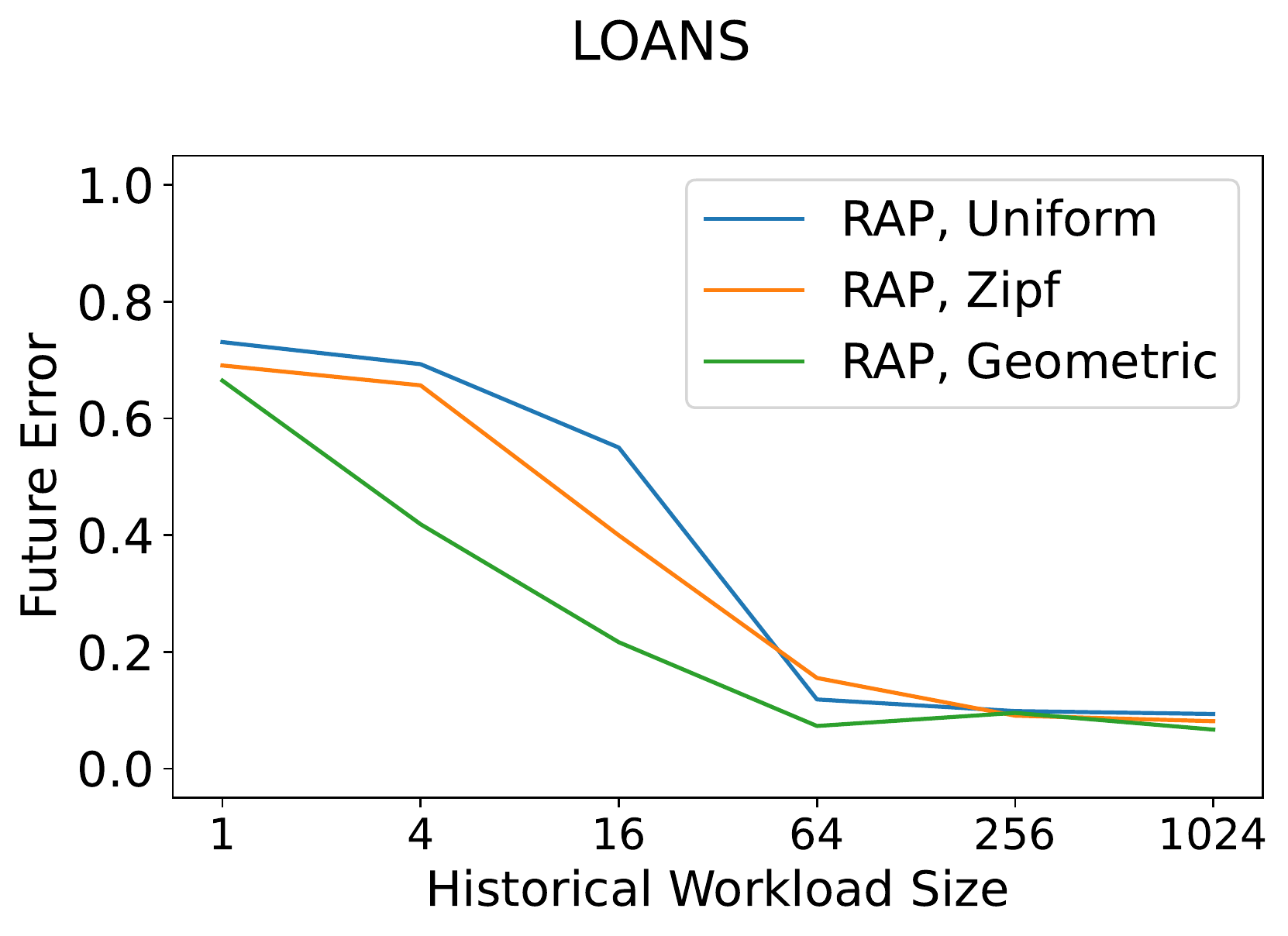}
\end{tabular}
\caption{\rap's future utility on each threshold distribution across a range of workload sizes.}
\label{fig:generalizing-dist-overlays}
\end{figure}

Figure~\ref{fig:generalizing-dist-lines} shows the results of this experiment.
As in our prior experiments, each point of the \rap\ line is taken to be where \rap\ achieves minimal \textit{present error} across all combinations of $T,K$ evaluated.
The future error at this minimizing $T,K$ pair is then evaluated and plotted, along with a corresponding 95\% confidence interval to account for randomness both between independent repetitions and across sampling future thresholds from the threshold distribution.
For real-world applications, this reflects what a practitioner using \rap\ would be able to do; i.e., choose the best performing instance of \rap\ across $T,K$ values on the present error metric (since they would not be able to evaluate future error), and then use that instance to answer future queries.
Ideally though, the practitioner would have omnisciently been able to choose the best performing instance of \rap\ across $T,K$ values on the future error metric directly, as this approach will never have larger future error than the former (feasible) approach.
To understand whether there is a significant difference in the future error between these two scenarios, we additionally plot the latter as ``\rap\ (opt)''.
For each distribution individually, we find the results are as expected.
Namely, \rap's future error is always lower than the baseline mechanism \bma's future error, and \rap's future error decreases as the number of historical thresholds that it is given increases.
Interestingly, we find no evidence that there is any point in which the number of historical thresholds given to \rap\ becomes ``too large'' and causes \rap's future error to begin increasing.
Instead, we find that \rap\ benefits from being provided more historical thresholds when the historical workload size is small, and then eventually reaches a point of diminishing returns.
Additionally, we find that the future error corresponding to the \rap\ instance that attains minimal present error across $T,K$ values is nearly identical to the future error corresponding to the \rap\ instance that attains minimal future error across $T,K$ values.
This indicates that in practice, answering future queries using the \rap\ instance that achieved minimal present error across $T,K$ values will likely also yield the minimal future error.

To better visualize the differences across distributions, \rap's future error lines are overlayed in Figure~\ref{fig:generalizing-dist-overlays} for both the ADULT and LOANS datasets.
From this, we see that the differences between \rap's future error across all three distributions are not as striking as one may expect, although for small historical workload sizes (less than 16 and 64 on the ADULT and LOANS datasets respectively) we find that the results roughly align with our intuition: the least concentrated (Uniform) distribution induces the highest future error, while the most concentrated (Geometric) distribution induces the lowest future error.
These findings, taken together with those of Figure~\ref{fig:generalizing-dist-lines}, yield a simple, useful insight into how to achieve low future error with \rap\ in practice.
Specifically, if the size of the historical workload is small, a practitioner can simply augment it by adding uniformly randomly sampled thresholds from the space of all possible thresholds (regardless of what the underlying threshold distribution $\cT_H$ is).
In the worst case, \rap's future error will be essentially unaffected (if $|W_H|$ was already in the region where returns are diminishing); in the best case, \rap's future error will be reduced significantly.

\subsubsection{Effect of ``Overfitting" the Synthetic Dataset}
When answering a prespecified set of queries using \rap, the goal in the relaxed projection step is to achieve as close to a global minima as possible.
In fact, although such an achievement is unlikely in practice, Aydore et al.'s theoretical utility result relies on a global minima having been reached.
However, when the goal is to learn a model that generalizes to unseen data, it is well known that optimizing the loss function to a global minima will lead to an extremely overfit model.
In the exact partial knowledge setting where future utility is the metric of choice, we seek to determine whether a conceptually similar ``overfitting'' phenomena may be occurring when \rap\ uses the historical threshold workload to generalize to future thresholds.

Towards this, we recall our finding from Figure~\ref{fig:generalizing-dist-lines}.
Specifically, that \rap\ does not seem to noticeably overfit to the historical queries when selecting the adaptivity parameters $T$ and $K$ based on the instance of \rap\ that had minimal present error.
However, this finding does not eliminate the possibility that \rap\ is overfitting to the historical queries during the synthetic dataset optimization procedure itself.
For instance, in Figure~\ref{fig:generalizing-dist-overlays} on the LOANS dataset at a historical workload size of 4, there is a significant difference between \rap's future errors on the Uniform vs.\ Geometric distributions.
This could be explained either by \rap\ overfitting to the historical workload generated from the Uniform distribution (which is relatively less informative regarding which thresholds are likely to be sampled in the future), or it could simply indicate that the historical workload does not contain enough information about the relevant space of thresholds that \rap\ needs in order to generalize well.
To analyze this possibility, we perform the same experiment as above while simultaneously evaluating \rap's future utility not just at the end of the optimization procedure, but after each iteration of the optimization procedure.
Figure~\ref{fig:generalizing-training-progress} displays the results, along with \rap's training loss and present error after each iteration of the optimization procedure.
In classic ML, a canonical symptom of overfitting is observing a point in the training progress where the training error continues decreasing, but where the test error begins steadily increasing.
In our setting, the analogue would be observing a point where the present error continues decreasing, but where the future error begins increasing.
However, we do not observe such behavior in either graph, as the future error steadily decreases throughout the entire training procedure.
The primary difference between the two graphs is that \rap's decrease in future error under the Uniform distribution is much smaller than under the Geometric distribution.
This simply indicates that, as expected, \rap\ is able to take advantage of the significantly more informative (with respect to the relevant portions of the space future thresholds will be drawn from) historical workload from the Geometric distribution.
Viewed differently, in the case of the Uniform distribution, \rap\ did not ``overfit'' to the historical workload --- rather, the historical workload simply did not provide enough information to \rap\ about the relevant remainder of the query space.

The take-away from these findings is that while \rap\ would have benefited from having a larger historical threshold workload, it would not have benefited from introducing analogues to other classic overfitting remedies.
For example, a practitioner may be tempted to reserve a held-out set of thresholds from the historical workload with the intention of using them between iterations as a proxy to estimate future utility, stopping the training early when the error on the held-out set begins increasing.
Not only do these findings indicate that such a strategy would not be beneficial, but combined with the findings from the previous experiment, we conclude that such a strategy would result in relatively \textit{greater} future error due to the reduced historical workload that \rap\ is given.

\begin{figure}
\centering
\begin{tabular}{cc}
\hspace*{-0.65cm}
  \includegraphics[width=.45\linewidth]{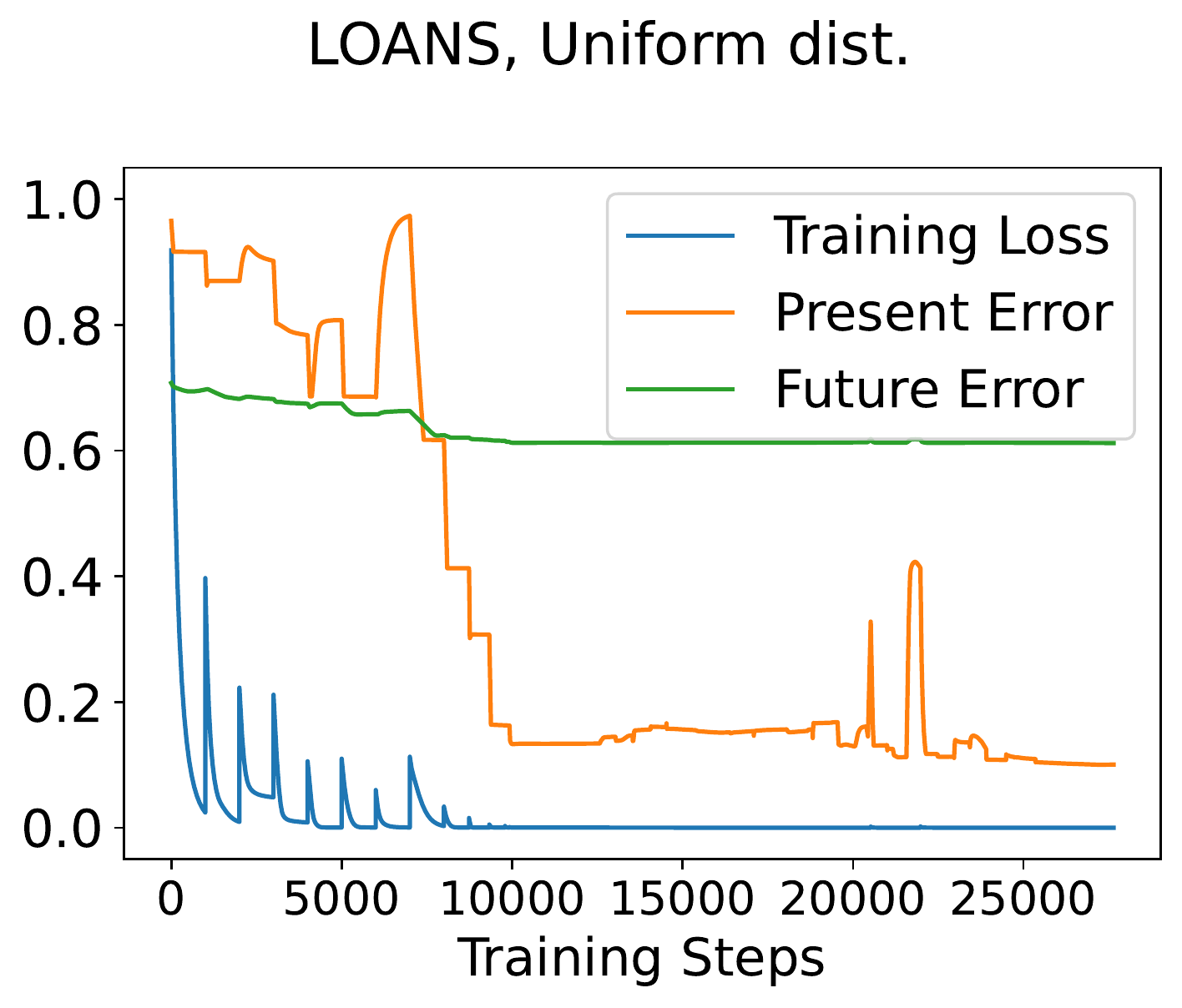} & \includegraphics[width=.45\linewidth]{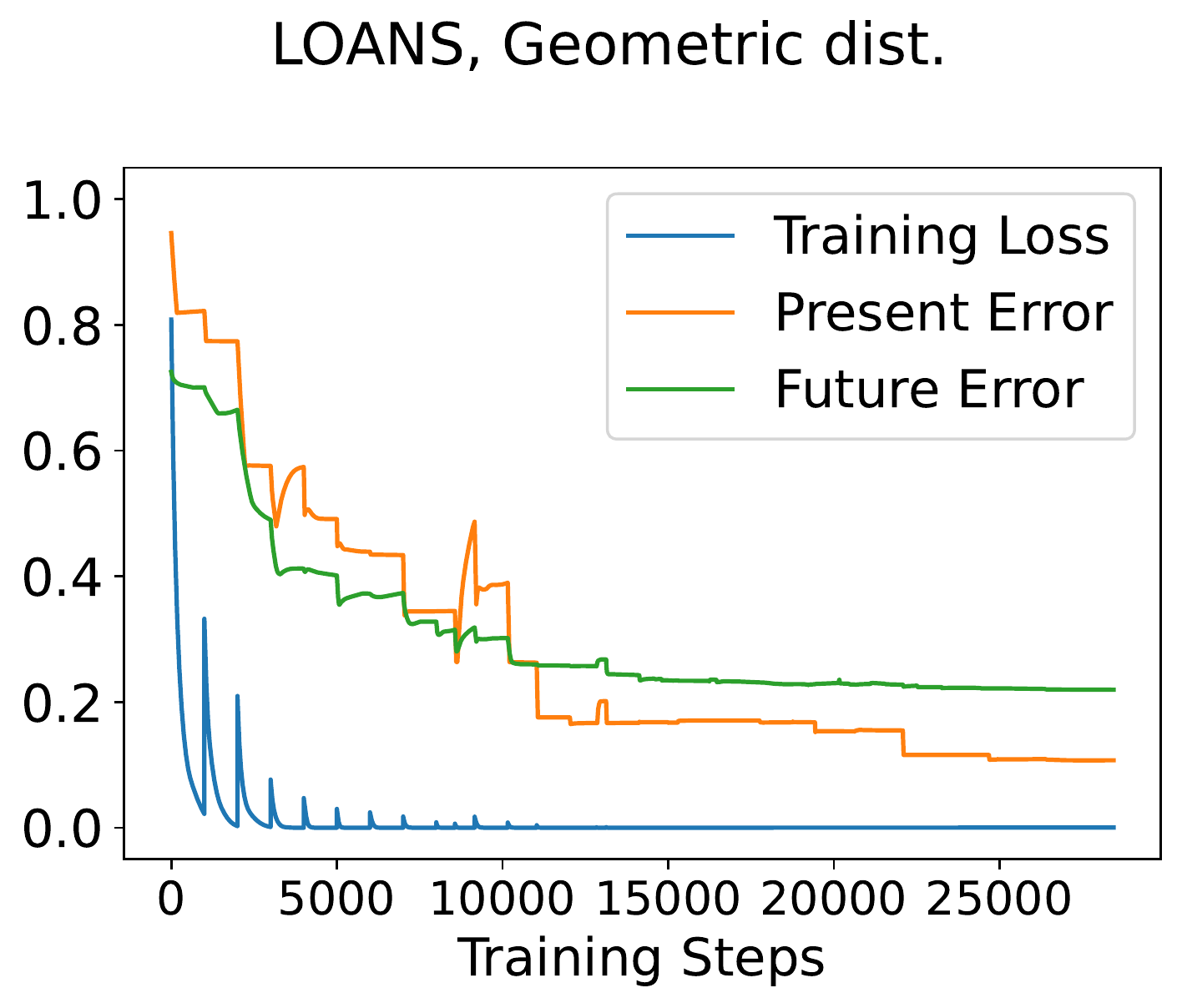}
\end{tabular}
\caption{Training progress across iterations for \rap\ on Uniform vs.\ Geometric distributions over features in LOANS dataset, both with a small historical workload size of 4.}
\label{fig:generalizing-training-progress}
\end{figure}

\subsubsection{Effect of Threshold Distribution Drift}
In the drifting partial knowledge setting, as the future features distribution $\cF_F$ drifts further from the historical features distribution $\cF_H$, it is clear that \rap's future utility should decrease.
However, it is unclear how \textit{sensitive} \rap's future utility is to such drift.
Thus, we seek to quantify the extent to which \rap\ can tolerate distributional drift while maintaining high future utility.

To achieve this, we evaluate \rap's future utility in the following experiment.
We first define the historical features distribution $\cF_H$ using the aforementioned highly concentrated Geometric distribution over features in both the ADULT and LOANS datasets.
We then measure \rap's future error across a range of drift amounts.
Because \rap\ achieved low future error in the exact partial knowledge setting on all distributions when the workload size was large enough, we anticipate that distributional drift will similarly not have a significant impact when the historical workload size is large.
Thus, in Figure~\ref{fig:generalizing-drift-effect}, we evaluate the impact of distributional drift specifically with small historical workload sizes of 4 and 16 on the ADULT and LOANS datasets respectively.

\begin{figure}
\centering
\includegraphics[width=.5\linewidth]{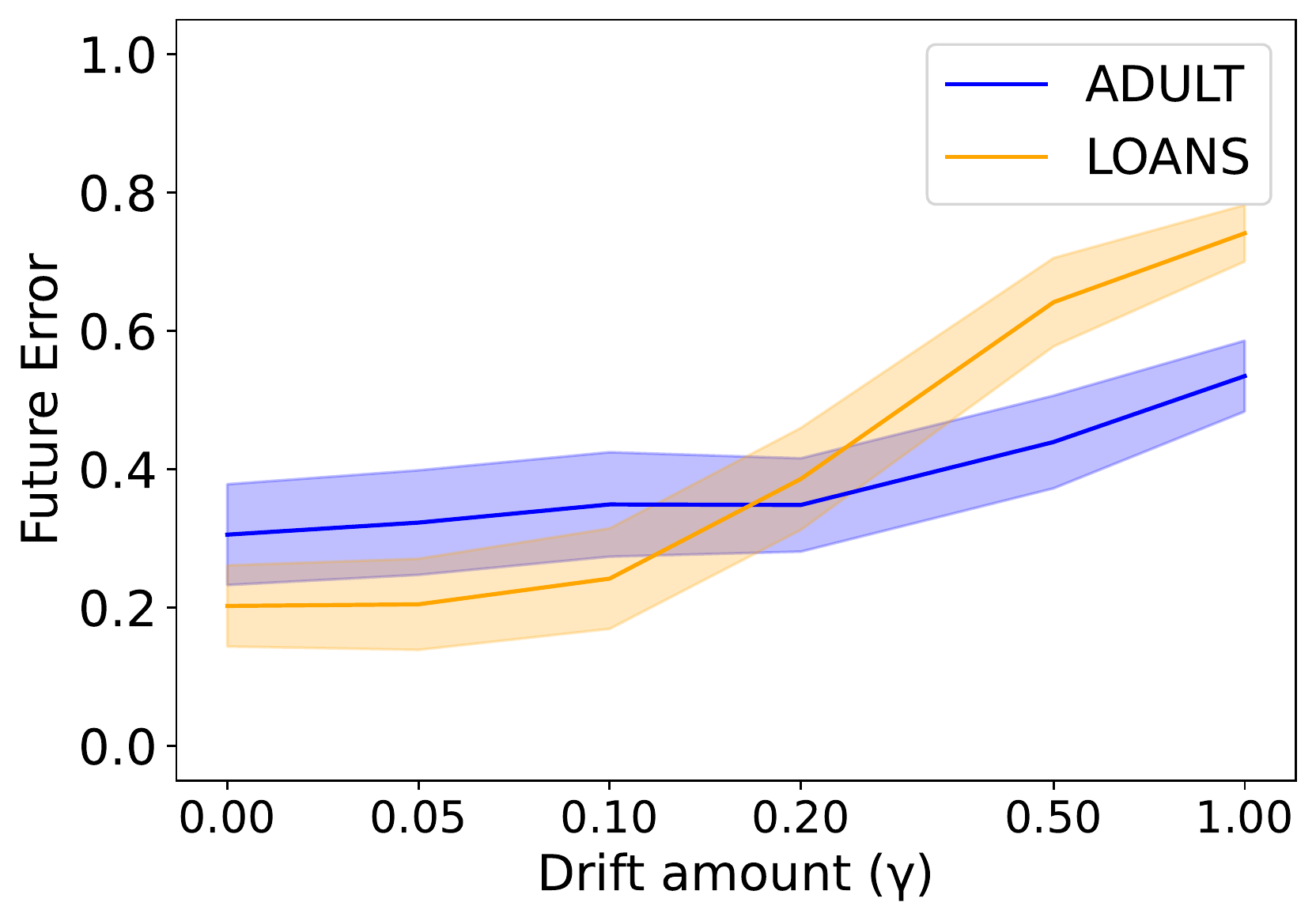}
\caption{Future error of \rap\ across a range of range of distributional drift amounts on the ADULT and LOANS datasets, given small historical workload sizes of 4 and 16 respectively.}
\label{fig:generalizing-drift-effect}
\end{figure}

The results of this experiment reveal that on both datasets, \rap\ is fairly impervious to distributional drift.
\rap's future error only begins to exhibit a significant increase at approximately $\gamma=0.4$ on the ADULT dataset and $\gamma=0.1$ on the LOANS dataset.
Comparing with Figure~\ref{fig:generalizing-drift-tvs}, these points both correspond to an expected total variation distance between the historical and future features distributions of approximately $0.5$ on their respective datasets.
Thus, we are able to conclude that even if the future features distribution drifts from the historical features distribution by a moderate amount, \rap\ can still be expected to maintain high utility.

\section{Additional Related Works} \label{sec:many-queries-related-works}
In this section, in addition to the prior works on large-scale query answering previously discussed (Section~\ref{sec:many-queries-prior-work}), we discuss other important works related to differentially private query answering.
We begin by discussing some works (concurrent with and subsequent to our work) related to answering large sets of prespecified queries.
For the mechanisms defined in these works, a prime direction for future research would be to evaluate them analogously to our evaluation of \rap\ in this work.
For instance, evaluating their scalability to larger query spaces as well as their generalizability for answering queries posed in the future, perhaps in a manner similar to Tao et al.~\cite{tao2021benchmarking}).
We then briefly discuss some lines of research related to the general problem and settings explored in this work.

\subsubsection*{Answering Many Queries}
One closely related work to the goals of this work is that of Liu et al.~\cite{liu2021iterative}, which studies the problem of constructing an algorithmic framework for privately answering a prespecified set of statistical queries --- our first setting of interest.
Concretely, the framework they construct unifies several DP mechanisms which specifically answer queries by building a synthetic dataset through iterative, adaptive updates.
These mechanisms include the previous practical state-of-the-art mechanisms~\cite{gaboardi2014dual, vietri2020new}, as well as a modified variant of a \textit{preliminary} version of the \rap\ mechanism (where a softmax transformation~\cite{bridle1990training} is applied to each row of the synthetic dataset $D$ after each iteration of \rap's optimization procedure).
Liu et al. then leverage their framework to design two new mechanisms for answering prespecified sets of queries, and empirically show that both achieve high utility.
However, in their empirical evaluations, Liu et al. find that the modified \rap\ mechanism's utility is on par with the utility of their two newly proposed mechanisms, and that \rap\ is computationally cheaper to execute.
Thus, we do not additionally evaluate their two new mechanisms in this work.
Moreover, Aydore et al.\ have subsequently updated the \rap\ mechanism to incorporate a similar modification (applying the Sparsemax transformation~\cite{martins2016softmax}, and optionally finishing with randomized rounding) and showed that it further improves utility --- in turn, further justifying our focus on the \rap\ mechanism.

Along similar lines, another closely related work is the recently introduced \textit{Adaptive and Iterative Mechanism} (AIM) by McKenna et al.~\cite{mckenna2022aim}.
AIM is a mechanism for DP synthetic data generation to specifically answer workloads of marginal queries.
The high-level idea of their approach is similar to that of \rap\ and Liu et al.'s work~\cite{liu2021iterative}, adaptively selecting marginals to use to optimize the synthetic dataset.
However, their work takes this a step further by designing a method to more intelligently perform the selection.
Moreover, they develop new techniques to quantify the uncertainty to answers derived from the generated synthetic data.
Empirically evaluating AIM, they show that it generally outperforms prior state-of-the-art mechanisms, including \rap.
However, their evaluation setting was somewhat different; specifically, they reduced the domain size of the datasets by discretizing numerical features into 32 equal-width bins.
This makes the optimization problem significantly easier for all mechanisms they evaluate, which is highly useful when running a wide range of experiments across many random trials.
However, it leaves AIM's utility unclear when the data is unbinned and sparse (e.g., for a numerical attribute with 100 possible values).
Moreover, since the source code of AIM's implementation was never released, we consider a ground-up reimplementation of AIM amenable to large-scale evaluations on large and sparse data spaces to be out of the scope of this work.
Performing such evaluations, especially in connection to the computational resources required by each method (AIM, \rap, and others), is a prime direction for future work.

Another closely related work is the concurrent theoretical work of Nikolov~\cite{nikolov2022private}, which proposes and analyzes a new mechanism for answering sets of prespecified queries with differential privacy.
Their new mechanism is based on randomly projecting the queries to a lower dimensional space, answering the projected queries with a simple DP additive-noise mechanism, then lifting the answers back into their original dimension.
The primary focus and contribution of their work is the thorough mathematical analysis of the mechanism's utility, showing that it achieves optimal worst case sample complexity under an average error metric.
Such results are less directly relevant to our work though, as our focus is on different error metrics for fixed real-world datasets (rather than in the worst case across all possible datasets).
However, conceptually, Nikolov's newly proposed mechanism could be used to tackle the same problem as our work.
Practically though, the runtime of Nikolov's mechanism (although polynomial) would prevent it from being used to answer the large number of queries that we answer with \rap\ in this work.
An intriguing direction for future work would be adapting Nikolov's new mechanism for practical query answering, and determining ways to scale it up to accurately answer queries on a truly large scale.

A final line of closely related work is the subsequent work of Vietri et al.~\cite{vietri2022private}.
The focus of their work is explicitly on enhancing the \rap\ mechanism, creating a new mechanism they call \rap++.
Their goal is orthogonal to the goal of this work, in that they seek to extend the original \rap\ mechanism so that it is able to support numerical features natively.
Prior to their work, \rap\ required one-hot discretization of any numerical features in the dataset.
For features with wide numerical ranges, one-hot discretization greatly increases the dimensionality of \rap's optimization problem, which in turn increases the computational burden and simultaneously decreases the mechanism's overall utility.
In \rap++, Vietri et al.\ incorporate tempered sigmoid annealing and random linear projection queries into \rap\ in order to handle a mixture of categorical and numerical features without any discretization.
They perform several empirical evaluations on \rap++, finding that it achieves state-of-the-art utility and runtime.
Despite their goal being orthogonal to the goal of this work, our findings from this work could be used to further improve the \rap++ mechanism and its evaluation.

\subsubsection*{Related Lines of Research}
One related (but disjoint) line of research is on the \textit{public/private} model of differential privacy, where some data must be protected with differential privacy while the remaining ``public'' data requires no privacy protections \cite{beimel2013private, ji2013differential, hamm2016learning, papernot2016semi, bassily2020private, alon2019limits, liu2021leveraging, tao2021prior}.
These works have shown that mechanisms can be designed which make use of a small amount of public data in order to significantly boost utility.
Our work differs from this model in that it does not directly make use of any public data.
In our newly defined partial knowledge setting, we instead assume that the entire set of user data $D$ is private, but that there exist publicly known historically posed queries $Q_H$ which are not privacy sensitive.
Assuming that $Q_H$ was generated from a random distribution $\cT_H$, we seek to understand the extent to which the \rap\ mechanism is able to take advantage of $Q_H$ using $D$ in order to accurately answer future queries generated from a distribution $\cT_F$ related to $\cT_H$.

The final related line of work is on reconstruction attacks, which studies how accurately sets of queries can be answered before private information in the dataset can be recovered.
The high level results of this research can be summarized through the Fundamental Law of Information Recovery~\cite{dwork2014algorithmic}: ``overly accurate answers to too many questions will destroy privacy in a spectacular way.''
Initial work on reconstruction attacks~\cite{dinur2003revealing} inspired the conception of DP, and subsequent works have improved the computational efficiency of attacks, improved the theoretical analyses of attacks, or crafted highly effective attacks to specific cases~\cite{dwork2007price, dwork2008new, muthukrishnan2012optimal, dwork2017exposed, garfinkel2019understanding}.
Although somewhat related, the focus of this line of work significantly differs from the focus of our work.
In research on reconstruction attacks, the basic goal is to find worst-case sets of queries (or the minimal sizes thereof) such that it is impossible to answer them all accurately while simultaneously maintaining privacy.
In this work, our focus is not on generic worst-case queries, but rather on efficiently and accurately answering practical sets of prespecified or randomly sampled queries with privacy.
Thus, the works on reconstruction attacks are not directly relevant to our problem in either of the two settings we consider.

\section{Conclusions} \label{sec:conclusions}
In this work, we address the high-level research question: \textit{to what extent are differentially private mechanisms able to answer a large number of statistical queries efficiently and with low error?}
We analyze this problem in two settings, the classic prespecified queries setting, and a new setting that we introduced where only partial knowledge of the queries is available to the DP mechanism in advance.
In both settings, our contributions are grounded in the state-of-the-art DP mechanism for answering large numbers of queries, the \rap\ mechanism.
In the prespecified queries setting, we perform a focused but thorough reproducibility study on Aydore et al.'s original evaluation of \rap\ in order to clarify its value and strengthen its adoptability for practical uses.
We also expand the class of queries that \rap\ is capable of evaluating, thus extending \rap's applicability in practice.
Aside from the prespecified queries setting, we concretely specify a new partial knowledge setting where a mechanism is provided with a set of historically posed queries which are similar to queries that will be posed in the future.
In this setting, we define a machine learning inspired utility measure to quantify a mechanism's ability to answer such future queries.
Then, utilizing this utility measure, we evaluate \rap's suitability for generating synthetic datasets to answer queries posed in the future, finding that it is both efficient and effective.
Our findings in this chapter further the state of the art in differentially private large-scale query answering, and additionally open new directions for future work on other problems in differential privacy within our newly defined partial knowledge setting.

\section*{Acknowledgements} \label{sec:acknowledgements}
This work was funded in part by a Privacy Enhancing Technologies Award\footnote{\url{ https://research.facebook.com/research-award-recipients/year/2021/?s=korolova}} from Facebook and NSF awards \#1916153, \#1956435, and \#1943584.

\bibliographystyle{alpha}
\bibliography{references}

\appendix

\section{Appendix} \label{app:many-queries-appendix}

\subsection*{Deferred Regression Analysis Details}

In this portion, we present the details of the setup and results for the regression analysis on the utility impact of filtering ``large'' marginals out of \rap's evaluation.

\subsubsection*{Present Error vs.\ Workload Size}
For this regression analysis on each dataset, we define the following regression variables:
\begin{itemize}
\item $x_1, x_2$: dummy variable encodings for the three levels of $\epsilon$ evaluated. I.e., 
\begin{itemize}[label=$\circ$]
\item $x_1 = x_2 = 0$ represents $\epsilon = 0.01$.
\item $x_1 = 1, x_2 = 0$ represents $\epsilon = 0.1$.
\item $x_1 = 0, x_2 = 1$ represents $\epsilon = 1$.
\end{itemize}
\item $x_3$: logarithm of the workload size.
\item $x_4$: indicator variable representing whether thresholding was applied. I.e., $x_4=0$ if thresholding was not applied, $x_4=1$ if it was.
\item $\zeta$: stochasticity in the process (e.g., from randomness in the \rap\ mechanism due to privacy, from randomness in the marginal selection process across independent trials, etc.).
\end{itemize}
With these variables defined, we state the full regression model with interactions as
\[ \textstyle\err_P = \beta_0 + \beta_1 x_1 + \beta_2 x_2 + (\beta_3 + \beta_4 x_1 + \beta_5 x_2)x_3 + (\beta_6 + \beta_7 x_1 + \beta_8 x_2 + (\beta_9 + \beta_{10} x_1 + \beta_{11} x_2)x_3)x_4 + \zeta, \]
and the restricted regression model as
\[ \textstyle\err_P = \beta_0 + \beta_1 x_1 + \beta_2 x_2 + (\beta_3 + \beta_4 x_1 + \beta_5 x_2)x_3 + \zeta. \]
We then fit both the full and restricted regression models to the results of the \rap\ evaluations for the ADULT and LOANS datasets (separately).
Regression results for the full models (ADULT on left and LOANS on right) are stated below.

\begin{table}[!htb]
    \begin{minipage}{.5\linewidth}
\begin{singlespace}
\scriptsize
\begin{tabular}{lclc}
\toprule
\textbf{Dep. Variable:}                                                                           &   present\_err   & \textbf{  R-squared:         } &     0.963   \\
\textbf{Model:}                                                                                   &       OLS        & \textbf{  Adj. R-squared:    } &     0.959   \\
\textbf{Method:}                                                                                  &  Least Squares   & \textbf{  F-statistic:       } &     266.6   \\
\textbf{Covariance Type:}                                                                         &    nonrobust     & \textbf{  Prob (F-statistic):} &  7.40e-76   \\
\textbf{No. Observations:}                                                                        &         126      & \textbf{  Log-Likelihood:    } &    295.45   \\
\textbf{Df Residuals:}                                                                            &         114      & \textbf{  AIC:               } &    -566.9   \\
\textbf{Df Model:}                                                                                &          11      & \textbf{  BIC:               } &    -532.9   \\
\bottomrule
\end{tabular}\\
\begin{tabular}{lcccccc}
                                                                                                  & \textbf{coef} & \textbf{std err} & \textbf{t} & \textbf{P$> |$t$|$} & \textbf{[0.025} & \textbf{0.975]}  \\
\midrule
$\beta_{0}$                                                                               &       0.0320  &        0.009     &     3.415  &         0.001        &        0.013    &        0.051     \\
$\beta_{1}$                                                                 &      -0.0066  &        0.013     &    -0.495  &         0.621        &       -0.033    &        0.020     \\
$\beta_{2}$                                                                 &      -0.0248  &        0.013     &    -1.869  &         0.064        &       -0.051    &        0.001     \\
$\beta_{3}$                                                                   &       0.0650  &        0.003     &    24.536  &         0.000        &        0.060    &        0.070     \\
$\beta_{4}$                                          &      -0.0528  &        0.004     &   -14.075  &         0.000        &       -0.060    &       -0.045     \\
$\beta_{5}$                                          &      -0.0600  &        0.004     &   -16.015  &         0.000        &       -0.067    &       -0.053     \\
$\beta_{6}$                                                 &       0.0280  &        0.013     &     2.120  &         0.036        &        0.002    &        0.054     \\
$\beta_{7}$                        &      -0.0309  &        0.019     &    -1.649  &         0.102        &       -0.068    &        0.006     \\
$\beta_{8}$                        &      -0.0277  &        0.019     &    -1.482  &         0.141        &       -0.065    &        0.009     \\
$\beta_{9}$                          &      -0.0036  &        0.004     &    -0.952  &         0.343        &       -0.011    &        0.004     \\
$\beta_{10}$ &       0.0052  &        0.005     &     0.988  &         0.325        &       -0.005    &        0.016     \\
$\beta_{11}$ &       0.0040  &        0.005     &     0.762  &         0.448        &       -0.006    &        0.014     \\
\bottomrule
\end{tabular}\\
\begin{tabular}{lclc}
\textbf{Omnibus:}       & 24.270 & \textbf{  Durbin-Watson:     } &    1.693  \\
\textbf{Prob(Omnibus):} &  0.000 & \textbf{  Jarque-Bera (JB):  } &  114.122  \\
\textbf{Skew:}          &  0.434 & \textbf{  Prob(JB):          } & 1.65e-25  \\
\textbf{Kurtosis:}      &  7.581 & \textbf{  Cond. No.          } &     64.4  \\
\bottomrule
\end{tabular}\\
\normalsize
\end{singlespace}
    \end{minipage}%
\begin{minipage}{.05\linewidth}
\phantom{.}
\end{minipage}
    \begin{minipage}{.5\linewidth}
\begin{singlespace}
\scriptsize
\begin{tabular}{lclc}
\toprule
\textbf{Dep. Variable:}                                                                           &   present\_err   & \textbf{  R-squared:         } &     0.942   \\
\textbf{Model:}                                                                                   &       OLS        & \textbf{  Adj. R-squared:    } &     0.937   \\
\textbf{Method:}                                                                                  &  Least Squares   & \textbf{  F-statistic:       } &     193.4   \\
\textbf{Covariance Type:}                                                                         &    nonrobust     & \textbf{  Prob (F-statistic):} &  1.05e-75   \\
\textbf{No. Observations:}                                                                        &         144      & \textbf{  Log-Likelihood:    } &    228.17   \\
\textbf{Df Residuals:}                                                                            &         132      & \textbf{  AIC:               } &    -432.3   \\
\textbf{Df Model:}                                                                                &          11      & \textbf{  BIC:               } &    -396.7   \\
\bottomrule
\end{tabular}\\
\begin{tabular}{lcccccc}                                                                                                  & \textbf{coef} & \textbf{std err} & \textbf{t} & \textbf{P$> |$t$|$} & \textbf{[0.025} & \textbf{0.975]}  \\
\midrule
$\beta_{0}$                                                                                &       0.0372  &        0.019     &     1.982  &         0.050        &     7.32e-05    &        0.074     \\
$\beta_{1}$                                                                 &      -0.0113  &        0.027     &    -0.425  &         0.671        &       -0.064    &        0.041     \\
$\beta_{2}$                                                                 &      -0.0282  &        0.027     &    -1.062  &         0.290        &       -0.081    &        0.024     \\
$\beta_{3}$                                                                   &       0.0966  &        0.004     &    21.626  &         0.000        &        0.088    &        0.105     \\
$\beta_{4}$                                          &      -0.0767  &        0.006     &   -12.134  &         0.000        &       -0.089    &       -0.064     \\
$\beta_{5}$                                          &      -0.0882  &        0.006     &   -13.953  &         0.000        &       -0.101    &       -0.076     \\
$\beta_{6}$                                                 &       0.0215  &        0.027     &     0.812  &         0.418        &       -0.031    &        0.074     \\
$\beta_{7}$                        &      -0.0273  &        0.038     &    -0.729  &         0.467        &       -0.102    &        0.047     \\
$\beta_{8}$                        &      -0.0275  &        0.038     &    -0.733  &         0.465        &       -0.102    &        0.047     \\
$\beta_{9}$                          &      -0.0039  &        0.006     &    -0.619  &         0.537        &       -0.016    &        0.009     \\
$\beta_{10}$ &       0.0039  &        0.009     &     0.437  &         0.663        &       -0.014    &        0.022     \\
$\beta_{11}$ &       0.0051  &        0.009     &     0.574  &         0.567        &       -0.013    &        0.023     \\
\bottomrule
\end{tabular}\\
\begin{tabular}{lclc}
\textbf{Omnibus:}       & 29.738 & \textbf{  Durbin-Watson:     } &    2.677  \\
\textbf{Prob(Omnibus):} &  0.000 & \textbf{  Jarque-Bera (JB):  } &  208.504  \\
\textbf{Skew:}          &  0.355 & \textbf{  Prob(JB):          } & 5.29e-46  \\
\textbf{Kurtosis:}      &  8.852 & \textbf{  Cond. No.          } &     75.6  \\
\bottomrule
\end{tabular}\\
\normalsize
\end{singlespace}
    \end{minipage}%
\end{table}

\subsubsection*{Present Error vs.\ Number of Queries}
For this regression analysis on each dataset, we define the same variables as in before, with the only change being that $x_3$ now represents the logarithm of the total number of consistent queries that \rap\ evaluates (rather than the size of the workload that \rap\ evaluates).
With these variables, we define the same full and restricted regression models as before, and we fit both to the results of the \rap\ evaluations.
Regression results for the full models (ADULT on left and LOANS on right) are stated below.

\begin{table}[!htb]
    \begin{minipage}{.5\linewidth}
\begin{singlespace}
\scriptsize
\begin{tabular}{lclc}
\toprule
\textbf{Dep. Variable:}                                                                         &   present\_err   & \textbf{  R-squared:         } &     0.889   \\
\textbf{Model:}                                                                                 &       OLS        & \textbf{  Adj. R-squared:    } &     0.879   \\
\textbf{Method:}                                                                                &  Least Squares   & \textbf{  F-statistic:       } &     83.19   \\
\textbf{Covariance Type:}                                                                       &    nonrobust     & \textbf{  Prob (F-statistic):} &  3.83e-49   \\
\textbf{No. Observations:}                                                                      &         126      & \textbf{  Log-Likelihood:    } &    227.07   \\
\textbf{Df Residuals:}                                                                          &         114      & \textbf{  AIC:               } &    -430.1   \\
\textbf{Df Model:}                                                                              &          11      & \textbf{  BIC:               } &    -396.1   \\
\bottomrule
\end{tabular}\\
\begin{tabular}{lcccccc}
                                                                                                & \textbf{coef} & \textbf{std err} & \textbf{t} & \textbf{P$> |$t$|$} & \textbf{[0.025} & \textbf{0.975]}  \\
\midrule
$\beta_{0}$                                                                              &      -0.3210  &        0.043     &    -7.438  &         0.000        &       -0.406    &       -0.235     \\
$\beta_{1}$                                                               &       0.2882  &        0.061     &     4.722  &         0.000        &        0.167    &        0.409     \\
$\beta_{2}$                                                               &       0.3014  &        0.061     &     4.939  &         0.000        &        0.181    &        0.422     \\
$\beta_{3}$                                                                   &       0.0472  &        0.004     &    12.856  &         0.000        &        0.040    &        0.054     \\
$\beta_{4}$                                          &      -0.0390  &        0.005     &    -7.516  &         0.000        &       -0.049    &       -0.029     \\
$\beta_{5}$                                          &      -0.0436  &        0.005     &    -8.398  &         0.000        &       -0.054    &       -0.033     \\
$\beta_{6}$                                               &       0.1198  &        0.057     &     2.110  &         0.037        &        0.007    &        0.232     \\
$\beta_{7}$                      &      -0.1237  &        0.080     &    -1.540  &         0.126        &       -0.283    &        0.035     \\
$\beta_{8}$                      &      -0.1189  &        0.080     &    -1.480  &         0.142        &       -0.278    &        0.040     \\
$\beta_{9}$                          &      -0.0127  &        0.005     &    -2.742  &         0.007        &       -0.022    &       -0.004     \\
$\beta_{10}$ &       0.0123  &        0.007     &     1.886  &         0.062        &       -0.001    &        0.025     \\
$\beta_{11}$ &       0.0124  &        0.007     &     1.894  &         0.061        &       -0.001    &        0.025     \\
\bottomrule
\end{tabular}\\
\begin{tabular}{lclc}
\textbf{Omnibus:}       & 53.796 & \textbf{  Durbin-Watson:     } &    1.512  \\
\textbf{Prob(Omnibus):} &  0.000 & \textbf{  Jarque-Bera (JB):  } &  189.737  \\
\textbf{Skew:}          & -1.528 & \textbf{  Prob(JB):          } & 6.30e-42  \\
\textbf{Kurtosis:}      &  8.177 & \textbf{  Cond. No.          } &     572.  \\
\bottomrule
\end{tabular}\\
\normalsize
\end{singlespace}
    \end{minipage}%
\begin{minipage}{.05\linewidth}
\phantom{.}
\end{minipage}
    \begin{minipage}{.5\linewidth}
\begin{singlespace}
\scriptsize
\begin{tabular}{lclc}
\toprule
\textbf{Dep. Variable:}                                                                         &   present\_err   & \textbf{  R-squared:         } &     0.887   \\
\textbf{Model:}                                                                                 &       OLS        & \textbf{  Adj. R-squared:    } &     0.877   \\
\textbf{Method:}                                                                                &  Least Squares   & \textbf{  F-statistic:       } &     93.96   \\
\textbf{Covariance Type:}                                                                       &    nonrobust  & \textbf{  Prob (F-statistic):} &  7.68e-57   \\
\textbf{No. Observations:}                                                                      &         144      & \textbf{  Log-Likelihood:    } &    180.50   \\
\textbf{Df Residuals:}                                                                          &         132      & \textbf{  AIC:               } &    -337.0   \\
\textbf{Df Model:}                                                                              &          11      & \textbf{  BIC:               } &    -301.4   \\
\bottomrule
\end{tabular}\\
\begin{tabular}{lcccccc}
                                                                                                & \textbf{coef} & \textbf{std err} & \textbf{t} & \textbf{P$> |$t$|$} & \textbf{[0.025} & \textbf{0.975]}  \\
\midrule
$\beta_{0}$                                                                              &      -0.6398  &        0.070     &    -9.171  &         0.000        &       -0.778    &       -0.502     \\
$\beta_{1}$                                                               &       0.5254  &        0.099     &     5.326  &         0.000        &        0.330    &        0.721     \\
$\beta_{2}$                                                               &       0.5873  &        0.099     &     5.952  &         0.000        &        0.392    &        0.782     \\
$\beta_{3}$                                                                   &       0.0779  &        0.005     &    14.839  &         0.000        &        0.068    &        0.088     \\
$\beta_{4}$                                          &      -0.0618  &        0.007     &    -8.321  &         0.000        &       -0.076    &       -0.047     \\
$\beta_{5}$                                          &      -0.0709  &        0.007     &    -9.550  &         0.000        &       -0.086    &       -0.056     \\
$\beta_{6}$                                               &       0.1453  &        0.096     &     1.509  &         0.134        &       -0.045    &        0.336     \\
$\beta_{7}$                      &      -0.1293  &        0.136     &    -0.949  &         0.344        &       -0.399    &        0.140     \\
$\beta_{8}$                      &      -0.1474  &        0.136     &    -1.082  &         0.281        &       -0.417    &        0.122     \\
$\beta_{9}$                          &      -0.0240  &        0.007     &    -3.647  &         0.000        &       -0.037    &       -0.011     \\
$\beta_{10}$ &       0.0195  &        0.009     &     2.088  &         0.039        &        0.001    &        0.038     \\
$\beta_{11}$ &       0.0227  &        0.009     &     2.431  &         0.016        &        0.004    &        0.041     \\
\bottomrule
\end{tabular}\\
\begin{tabular}{lclc}
\textbf{Omnibus:}       & 18.588 & \textbf{  Durbin-Watson:     } &    1.775  \\
\textbf{Prob(Omnibus):} &  0.000 & \textbf{  Jarque-Bera (JB):  } &   78.893  \\
\textbf{Skew:}          & -0.142 & \textbf{  Prob(JB):          } & 7.39e-18  \\
\textbf{Kurtosis:}      &  6.615 & \textbf{  Cond. No.          } &     726.  \\
\bottomrule
\end{tabular}\\
\normalsize
\end{singlespace}
    \end{minipage}%
\end{table}

\end{document}